\documentclass[ejsv2]{imsart}
\def\isXr{0}
\input{myPreliminary}

\begin{document}
\ifnum\isXr=1{\externaldocument{supp}}\fi
\begin{frontmatter}
\title{Difference-based variance estimators with repeated measurements}
\runtitle{Difference-based variance estimators}

\begin{aug}
\author[A]{\fnms{Chak Ming}~\snm{Lee}\ead[label=e1]{chakminglee@link.cuhk.edu.hk}}
\and
\author[A]{\fnms{Kin Wai}~\snm{Chan}\ead[label=e2]{kinwaichan@cuhk.edu.hk}}
\address[A]{Department of Statistics and Data Science, The Chinese University of Hong Kong\printead[presep={,\ }]{e1,e2}}

\runauthor{Lee and Chan}
\end{aug}

\begin{abstract}
In this paper, we formulate a general differencing framework for variance estimation across a range of settings.
We demonstrate that conventional difference-based noise variance estimators cannot achieve the desired bias-correcting power in nonparametric regression with repeated measurements.
A new high-order bias-corrected differencing scheme, adapted to repeated measurements, is proposed by interlacing inter-group and intra-group differencing.
The theoretical properties of the new sequences and estimators are studied.
Our proposals are particularly efficient in finite samples and under high signal-to-noise ratio scenarios,
where asymptotic convergence has not yet fully taken effect, due to their strong bias-correcting power.
\end{abstract}



\end{frontmatter}

\section{Introduction \label{section_review_diff}}
Repeated measurements facilitate noise variance estimation
in nonparametric regression
since
fitting of the regression function is not needed
when a design point has multiple observations.
However,
the conventional difference-based variance estimator fails to achieve the
desired high-order corrected bias in regression data with repeated measurements.
We derive the minimal difference order $m$ such that the desired bias property is achievable
and propose a new difference sequence adaptive to repeated data.
We show that repeated measurements boost both the robustness and efficiency of the proposed difference-based methods.

Consider a nonparametric regression problem of the form
\begin{equation}\label{eqt:model}
    Y_{h,i}=g(X_{h,i})
    +\epsilon_{h,i}, \qquad(h=1,\ldots,s(i)
    \; \text{and} \; i=1,\ldots,n)
\end{equation}
where $X_{1,i}=\cdots=X_{s(i),i} \equiv X^\star_i {\in \mathbb{R}}$ for each $i$,
{$s(i)\geq 1$\linelabel{def:s(i)} is the number of repeated measurements at $X^\star_i$,}
{the unknown regression function $g:\mathbb{R}\rightarrow\mathbb{R}$ is continuous},
and the noises $\epsilon_{h,i}$'s are independent and identically distributed random variables
with zero mean, variance $\sigma^2<\infty$, and
$\lambda_4=\E(\epsilon_{h,i}^4) <\infty$. \linelabel{def:lambda4}
The independence assumption will be relaxed in \S~\ref{section_extension}.
If $s(1) = \cdots = s(n) = 1$,
i.e., without repeated measurements,
then it is called a standard design; otherwise,
it is called a repetitional design.
If $s(1) = \cdots = s(n)$, the repetitional design is said to be balanced; otherwise, it is unbalanced.
We stack the data in (\ref{eqt:model}) as follows:
\begin{align}
    Y=(Y_1, \ldots, Y_N)
        &=(
        \underbracket{Y_{1,1},\ldots,Y_{s(1),1}}_{Y_1, \ldots, Y_{s(1)}},\quad
        \ldots, \quad
        \underbracket{Y_{1,n},\ldots,Y_{s(n),n}}_{Y_{S(1,n-1)+1}, \ldots, Y_{N}}
        ) , \label{eqt:stackY} \\
    X=(X_1, \ldots, X_N)
        &= (\underbracket{X_{1,1},\ldots,X_{s(1),1}}_{\equiv X^\star_1},\quad\ldots,\quad \underbracket{X_{1,n},\ldots,X_{s(n),n}}_{\equiv X^\star_n}),
        \label{eqt:stackX}\\
    \epsilon=(\epsilon_1, \ldots, \epsilon_N)
        &= (\epsilon_{1,1},\ldots,\epsilon_{s(1),1},\quad\ldots,\quad\epsilon_{1,n},\ldots,\epsilon_{s(n),n}),\label{eqt:stackEpsilon}
\end{align}
where $N = S(1,n)$ is the sample size and $S(a,b)=\sum_{i=a}^bs(i)$.
The setup (\ref{eqt:model}) is applicable to univariate covariates;
see \S~\ref{section_extension} for the extension to the multivariate case.
Without loss of generality, assume that
$0\leq X^\star_1< \cdots<X^\star_n\leq 1$.

The estimation of the variance $\sigma^2$ is important in, e.g.,
construction of confidence intervals of the nonparametric regression function $g$ \citep{H1992},
local bandwidth selection \citep{WA1988},
and nonparametric tests \citep{Y1992, Z2009, KAP2015}.
In the absence of repeated measurements,
{the sample size is $n=N$.}
Typical estimators of $\sigma^2$ admit a quadratic form:
\begin{equation} \label{quadratic_form}
    \hat{\sigma}_{\text{QF}}^2 =Y^{\T}AY / \tr(A),
\end{equation}
where $A$ is an $n\times n$ positive definite and symmetric matrix
and $\tr(A)$ is the trace of $A$.
Estimators in the form of (\ref{quadratic_form}) in
the standard regression problem have a long history.
Examples include
the residual-based method \citep{W1978, CES1992, HC1989, HM1990, N1994, MSW2003} and
the difference-based method \citep{R1984, GSJ1986, HKT1990};
see
{\ifnum\isXr=1{\S~\ref{rem:resBasedVarEst}}\else{\S~B.2}\fi}
for a detailed review.
{This article studies the difference-based
estimators that
admit the form in (\ref{quadratic_form}) with
\begin{align}\label{eq:Wmat_diffseq}
    A=\DiffSeqMat^{\T}\DiffSeqMat
    \quad\text{and} \quad
    \DiffSeqMat
    = \begin{pmatrix}
        v_1 \cdots v_n
    \end{pmatrix}
    = \begin{pmatrix}
        d_{1,0}  & \cdots & d_{1,m} & \cdots & 0 & \cdots & 0 \\
        \vdots &  \ddots  & \vdots  & \cdots &\vdots &\ddots & \vdots \\
        0 & \cdots & 0 & \cdots &d_{n-m,0} & \cdots & d_{n-m,m}
    \end{pmatrix} ,
\end{align}}
where $d_{(i)}=(d_{i,0},\ldots,d_{i,m})^\T$\linelabel{def:d(i)}
is a $m$th order difference sequence for locally differencing
the data $Y_{i},\ldots,Y_{i+m}$ such that
\begin{align}
	\sum^m_{j=0}d_{i,j}=0 \label{diff_poly_cond0}
	\qquad \text{and} \qquad
	\sum^m_{j=0}d^2_{i,j}=1 .
\end{align}
If $(X_i)_{i=1}^n$ are equidistant,
it suffices to do global differencing by
using
$d_{(1)}=\cdots=d_{(n-m)} \equiv d$
with a common
$d= (d_0, \ldots, d_m)^{\T}$.
Generally, using $d_{(i)}$, we form the difference statistic
\begin{align}\label{eqt:def_D}
	D_i=\sum^m_{j=0}d_{i,j}Y_{i+j}=D^g_i+D^\epsilon_i,
         \quad \text{where} \quad
        D^g_i=\sum^m_{j=0}d_{i,j}g(X_{i+j})
        \quad  \text{and} \quad
        D^\epsilon_i=\sum^m_{j=0}d_{i,j}\epsilon_{i+j}
\end{align}
for $i=1, \ldots, n-m$.
Then $\hat{\sigma}_{\text{QF}}^2=\sum^{n-m}_{i=1}D^2_i/(n-m)$ can be rewritten as
\begin{align}\label{diff_est_general}
    \hat{\sigma}^2_{*}
    	=\frac{1}{n-m}\sum^{n-m}_{i=1}\left(\sum^m_{j=0}d_{i,j}Y_{i+j}\right)^2
	\quad\text{and}\quad
	\hat{\sigma}^2=\frac{1}{n-m}\sum^{n-m}_{i=1}\left(\sum^m_{j=0}d_jY_{i+j}\right)^2
\end{align}
if local differencing and global differencing are used, respectively.
The component $D_i^g$ is asymptotically negligible when $g(\cdot)$ is {Lipschitz continuous or high-order differentiable;
see Proposition \ref{prop:poly_cancelling} for the details.}
Note that $D^\epsilon_i$ serves as a pseudo noise
since $\E(D^\epsilon_i)=0$ and $\Var(D^\epsilon_i)=\sigma^2$.
As a result, $\hat{\sigma}^2_{*}$ and $\hat{\sigma}^2$ are weakly consistent for $\sigma^2$.
In the literature, there are two classes of difference sequences: ordinary and optimal.

The ordinary difference sequence,
which we call the bias-optimal difference sequence,
uses all free parameters in $d_{(i)}$ to perform differencing
\citep{R1984,GSJ1986,DMW1998}.
The variance-optimal difference sequence uses all free parameters in $d_{(i)}$ to minimize
$\MSE(\hat{\sigma}_{}^2) = \E\{(\hat{\sigma}_{}^2-\sigma^2)^2\}$ after doing the first-order differencing \citep{HKT1990}.
These two classes of sequences
and our proposal to be introduced
are summarized in Table \ref{table:diff_seq}.

Ideally, differencing removes the
nuisance
$D_i^g$ in (\ref{eqt:def_D})
so that $\hat{\sigma}_{*}^2$ measures solely the noise variability.
However, by Taylor's expansion of $g(x)$ at $x=X_i$,
\begin{align}\label{eq:TE_Di}
    D_i^g
    =
    	\sum_{h=0}^{m} E_h
		+o\big\{(X_{i+m}-X_{i})^{m}\big\} ,
\end{align}
where $E_h = g^{(h)}(X_i)/h!\sum_{j=0}^md_{i,j}(X_{i+j}-X_{i})^h.$
The $m$th order ordinary sequences ensure $E_0=\cdots = E_{m-1} = 0$
{due to the bias-correcting constraints to be introduced in (\ref{diff_poly_random_cond}) with $r=m-1$},
leading to a smaller bias, while
the variance is not optimized.
The optimal sequences ensure only $E_0=0$
from the first constraint in
(\ref{diff_poly_cond0}).
It creates a large bias,
although the variance is optimized.
So,
neither method satisfactorily balances bias and variance.
Example \ref{eg:compare} illustrates this.

\begin{figure}[t]
  \begin{center}
  \includegraphics[width=.8\textwidth]{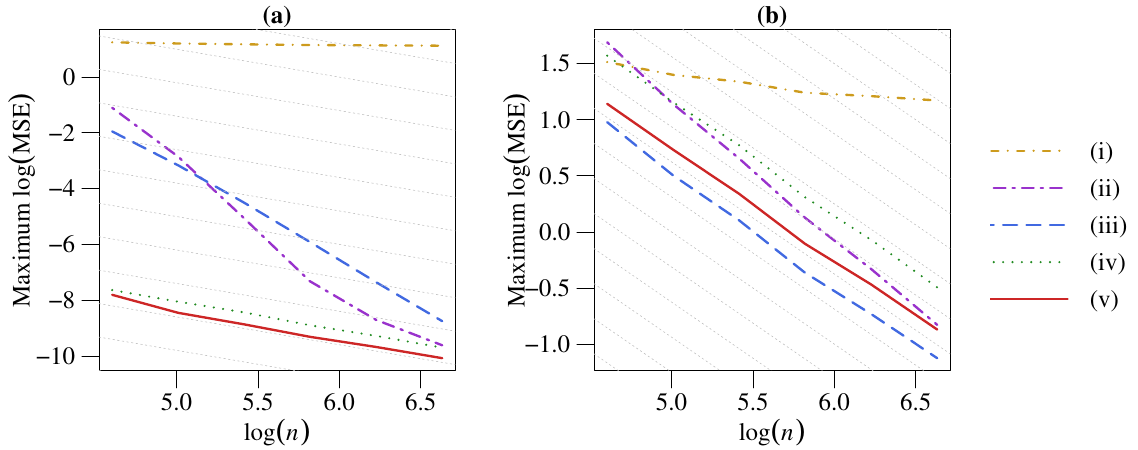}
  \vspace{-0.2cm}
  \caption{
  The maximum log MSE across 4 regression functions are plotted against $\log(n)$.
  The estimators (i)--(v) are defined in the example in \S~\ref{section_review_diff},
  where (v) is our proposal.
  Plots (a) and (b) correspond to $\sigma^2=0.1$ and $\sigma^2=10$, respectively.
  The reference dotted lines are of slope $-1$.
  Each experiment is performed with $2048$ replications.
  }
  \label{fig_compare_diff_seq}
  \end{center}
\end{figure}

\begin{example}\label{eg:compare}
Let
$g(x) = \sum^{k}_{s=0}\binom{k+\alpha}{k-s}\binom{k+\beta}{s} (x-1)^s(x+1)^{k-s}/2^k$,
i.e., a Jacobi polynomial,
with $\alpha=1$, $\beta=10$ and $k=2,4,6,8$.
Let $\epsilon_i\sim \Normal(0,\sigma^2)$ independently
with $\sigma^2=0.1,10$ and
$X_i = i/n$ ($i=1, \ldots, n$).
Five estimators are compared:
(i) the parametric estimator:
a residual-based estimator that fits
a quartic trend,
(ii) the nonparametric estimator: a residual-based estimator that fits a local linear regression with Epanechnikov kernel and bandwidth selected by {the R function \texttt{nprobust::lpbwselect} \citep{CCF2019}},
(iii) the variance-optimal third-order difference-based estimator \citep{HKT1990},
(iv) the bias-optimal third-order difference-based estimator, and
(v) the proposed third-order difference-based estimator with order $r=1$ to be introduced.
Figure \ref{fig_compare_diff_seq} shows their mean squared errors (MSE).
The variance-optimal and the bias-optimal estimators perform poorly when the signal-to-noise ratio is
high and low,
respectively.
The parametric and nonparametric estimators have poor performances in both cases.
Our proposal performs decently
regardless of the signal-to-noise ratio.
\end{example}

For regression with repeated measurements,
as far as we know, \citet{DMTZ2015} is the only existing approach on difference-based variance estimation,
despite that many real problems fall into this setting
\citep{GHD2010,L2015}.
They proposed three methods to utilize the repeated observations.
However, they are not optimally adapted to repeated data.
This motivates this research.

In this article, we contribute
in three aspects.
First, under some conditions, for any user-specified bias order, the variance-optimal difference sequence is derived for regular and random designs.
Second, an angular representation system of such optimal difference sequences is derived for fast computation.
Third, we propose difference-based variance estimators that work for repeated measurements. They achieve appealing theoretical and empirical results.

The rest of the article is organized as follows.
In \S~\ref{section_gds},
we introduce our general framework for defining our proposed difference
sequence.
In \S~\ref{sec:repeated},
we extend the settings to data with repeated measurements and propose a new differencing scheme adapted to repeated measurements.
\S~\ref{section_theory} states the theoretical results of our proposal estimators.
\S~\ref{section_algorithm} provides details for implementing our proposals.
In \S~\ref{section_extension},
we discuss some possible extensions of the proposal difference sequences, including multivariate data and time series data.
In \S~\ref{section_experiment},
we compare our proposed estimators with existing methods in simulation experiments.
We also apply one of our proposals in an existing test to reveal the impact of our proposed difference sequences in a simulation experiment and a real data example.

In summary,
to arrive at our final proposed estimator,
it involves six intermediate steps:
\begin{itemize}
    \item[(i)] $m$th order differencing for reducing variance, i.e., the existing $\hat{\sigma}^2_*$ and $\hat{\sigma}^2$ in (\ref{diff_est_general}),
    \item[(ii)] $r$th order bias correction for reducing bias, i.e., $\hat{\sigma}^2(m,r)$ in (\ref{eqt:hat_m_r}),
    \item[(iii)] local differencing for handling unequally spaced covariates, i.e., $\hat{\sigma}^2_*(m,r)$ in (\ref{eqt:hat_m_r_star}),
    \item[(iv)] intra-group differencing for handling repeated measurements, i.e., $\hat{\sigma}^2_\textsc{r}(m,r)$ in (\ref{eq:rep_diff_est}),
    \item[(v)] $V$-statistic aggregation for boosting efficiency, i.e., $\hat{\sigma}^2_{\textsc{r}+}(m,r)$ in (\ref{eq:repeated_vb}), and
    \item[(vi)] adjusting $m$ for handling imbalanced design, i.e., $\hat{\sigma}^2_{\textsc{r}+}(m,r)$ for $m\geq m_{\min}'$ in (\ref{eq:repeated_m_prime}).
\end{itemize}
Table \ref{table:diff_seq_summary} provides a summary of
of the above estimators.
For ease of referencing,
Table \ref{table:dictionary} provides a summary of notation that is used throughout the article.
The \texttt{R} function \texttt{dvar} for computing the proposed estimator can be found in the online supplementary materials.

\begin{table}[t]
    \def~{\hphantom{0}}
    \caption{\label{table:diff_seq_summary}
    Summary table of difference-based variance estimators.
    The higher the order $m\in\{1,2,\ldots\}$ the smaller the variance of the estimator.
    The higher the order $r\in\{0,1,\ldots,m-1\}$ the stronger the bias-correcting power.
    The adaptability and key properties of each proposal to various design points are also stated.
    We abbreviate the references as follows:
    R1984 \citep{R1984}, HKT1990 \citep{HKT1990},
    GSJ1986 \citep{GSJ1986}, and DMW1998 \citep{DMW1998}.
    }
    \setlength{\tabcolsep}{2pt}
    \begin{adjustbox}{width=\textwidth}
        \begin{tabular}{llllllll}
            \toprule
            &&& \multicolumn{3}{c}{Design} &\\
            \cline{4-6}
            Estimator & $m$ & $r$ & Regular & Random & Repetitional   & Key property\\[4pt]
            \midrule
            $\hat{\sigma}^2(m,r)$ & Any & $0$  & Yes & No& No & Variance-optimal
            (R1984, HKT1990)
            \\
            &Any & $m-1$  & Yes & No& No&
            Bias-optimal
            (GSJ1986, DMW1998)
            \\
            &Any & Any  & Yes & No& No& (New) High-order bias corrected and variance-optimal via (\ref{diff_poly_cond0}, \ref{diff_poly_cond})
            \\
            $\hat{\sigma}_{*}^2(m,r)$ &Any & Any & Yes& Yes & No& (New) Local differencing with angular representation  (Lemma \ref{lemma:diff_sol}) \\
            $\hat{\sigma}_{\Long}^2(m,r)$ & $\geq m_{\min}$ & Any  & Yes& Yes& Yes& (New) Interlacing intra-group and inter-group differencing via (\ref{eq:rep_diff_est})\\
            $\hat{\sigma}_{\Long+}^2(m,r)$ & $\geq m_{\min}$ & Any & Yes& Yes& Yes& (New) Efficiency boosting via V-statistics via (\ref{eq:repeated_vb})\\
            $\hat{\sigma}_{\Long+}^2(m,r)$ &$\geq m'_{\min}$ & Any  & Yes& Yes& Yes& (New)  Adaptive to imbalanced designs to boost efficiency by (\ref{eq:repeated_m_prime}) \\
            \bottomrule
        \end{tabular}
    \end{adjustbox}
\end{table}

{
\begin{table}[t]
    \def~{\hphantom{0}}
    \setlength{\tabcolsep}{2pt}
    \caption{\label{table:dictionary}Summary of notation.}
    \begin{adjustbox}{width=\textwidth}

        \begin{tabular}{lll}
            \toprule
             & Notation & Definition  \\
            \midrule
            \S 1 & $Y_{h,i}, X_{h,i}, \epsilon_{h,i}$ & The $h$th records of the response, covariate, and noise at the $i$th design point; see (\ref{eqt:model})\\
            & $Y_i,X_i,\epsilon_i$ & The $i$th elements in $Y,X,\epsilon$; see (\ref{eqt:stackY}), (\ref{eqt:stackX}), and (\ref{eqt:stackEpsilon}) \\
            & $n$ & Total number of distinct design points \\
             & $N$ & Total sample size\\
             & $s(i)$ & Number of measurements at the $i$th distinct design point; see line \ref{def:s(i)} on page \pageref{def:s(i)} \\
             & $\lambda_4=\E(\epsilon_{h,i}^4)$ & Fourth moment of the noises; see line \ref{def:lambda4} on page \pageref{def:lambda4}\\
             & $d_{(i)}=(d_{i,0},\ldots,d_{i,m})$ & Difference sequence of order $m$ for the data $(X_{i},\ldots,X_{i+m})$; see line \ref{def:d(i)} on page \pageref{def:d(i)}\\
             & $D_i$ & Difference statistic computed with $(X_i,\ldots,X_{i+m})$ and $d_{(i)}$; see (\ref{eqt:def_D}) \\
             & $A,V$ & The matrices that define the quadratic form of an estimator; see (\ref{quadratic_form}) and (\ref{eq:Wmat_diffseq}) \\ \bottomrule
            \S 3
             & $\mathcal{X}_i = \{ X_{i}, \ldots, X_{i+m} \}$ & Set of designs points of data in $D_i$; see line \ref{def:X_i} on page \pageref{def:X_i} \\
            & $\mathcal{X}^\star_i=\{ X_{i,1}^\star,\ldots,X_{i,m(i)+1}^\star \}$ & Set of distinct design points in $\mathcal{X}_i$; see line \ref{def:Xstari} on page \pageref{def:Xstari} \\
             & $m(i)=|\mathcal{X}^\star_i|-1$ & Actual diferencing order of $d_{(i)}$; see line \ref{def:m(i)} on page \pageref{def:m(i)} \\
            & $\psi(\cdot,h)$ & $h$th permutation function of the indices in $\{Y_1,\ldots,Y_N\}$; see line \ref{def:psi(dot,h)} on page \pageref{def:psi(dot,h)}\\
             & $\Psi$ & Set of all possible permutation functions $\psi(\cdot,h)$; see line \ref{def:Psi} on page \pageref{def:Psi}\\ \bottomrule
            \S 4 & $S_n$ & Total number of design points with single measurements; see line \ref{def:Sn} on page \pageref{def:Sn} \\
             & $m_{\min}$ & Minimal $m$ such that all $D_i$ are repetitional-bias-corrected; see line \ref{def:m_min} on page \pageref{def:m_min} \\
             & $m_{\min}'$ & Adjusted value of $m_{\min}$ for handling imbalanced design; see
             line \ref{def:m_min'} on page \pageref{def:m_min'}\\ \bottomrule
            \S 5 & $\mathcal{L}_i'$ & Reduced row echelon form of the matrix $\mathcal{L}_i$; see line \ref{def:Lrref} on page \pageref{def:Lrref}\\
            & $t(i,j)$ & Matching index such that $X_{i+j}=X^\star_{i,t(i,j)}$; see line \ref{def:t(i,j)} on page \pageref{def:t(i,j)}\\
            & $\Lambda_j(i)$ & Set of indices $j'\in\{0,\ldots,m\}$ such that $X_{i+j'}=X^\star_{i,1+j}$; see (\ref{def:Lambdaj(i)}) \\
            & $\mathcal{I}_m$ & Set of indices $i\in\{1,\ldots,N-m\}$ such that all design points in $\mathcal{X}_i$ are repeated; see (\ref{def:Im})\\
            & $\mathcal{Q}_i=\{t:X_t=X_i\}$ & Set of indices of design points that are equal to $X_i$; see line \ref{def:Qi} on page \pageref{def:Qi}\\
             & $\alpha(i)$ & Matching index such that $X_{\alpha(i)}^\star=X_i$; see line \ref{def:alpha(i)} on page \pageref{def:alpha(i)}\\
            \bottomrule
        \end{tabular}
    \end{adjustbox}

\end{table}
}

\section{High-order accurate differencing \label{section_gds}}
In this section,
we consider differencing for standard design, i.e., all design points $X$ are distinct.
Define $r\in\{0, 1, \ldots, m-1\}$ as a bias-correcting parameter.
The first constraint in (\ref{diff_poly_cond0}) can be generalized
to the $r$th order
bias correcting
constraints:
\begin{align}
    \sum^m_{j=0}d_{i,j}(X_{i+j}-X_{i})^h&=0 \qquad (h=1,\ldots,r), \label{diff_poly_random_cond} \\
    \sum^m_{j=0}d_jj^h &=0 \qquad (h=1,\ldots,r). \label{diff_poly_cond}
\end{align}
The constraints (\ref{diff_poly_random_cond}) and
(\ref{diff_poly_cond})
force $E_1=\cdots = E_r=0$ in (\ref{eq:TE_Di})
under arbitrary design and equidistant design, respectively.
A $m$th order difference sequence
(i.e., $d_{i,0}, \ldots, d_{i,m}$ in (\ref{diff_poly_random_cond})
or $d_0, \ldots, d_m$ in (\ref{diff_poly_cond}))
contains $m-1$ degrees of freedom in view of (\ref{diff_poly_cond0}).
These $m-1$ degrees of freedom can be divided to achieve two goals:
(a) $r$ of them are used to enhance robustness against trends, and
(b) $m-r-1$ of them are used to boost
efficiency.
Exhausting all the degrees of freedom on either side results in the inferior performance of the two traditional approaches demonstrated in Example \ref{eg:compare}.
Flexibly choosing the parameter $r$
opens up the possibility of striking a better balance between robustness and efficiency.
Hence, we propose estimators that satisfy
the high-order accurate differencing
schemes in (\ref{diff_poly_random_cond})--(\ref{diff_poly_cond})
with $0\leq r \leq m-1$.
{Recall that $d_{(i)}$
is a $m$th order difference sequence for locally differencing
the data $Y_{i},\ldots,Y_{i+m}$.}

\begin{assumption}[Two types of difference sequences]\label{ass:typeDiffSeq}
Let $m\in\mathbb{N}$ and $r\in\mathbb{N}_0$ such that $r+1\leq m$.
\begin{enumerate}[topsep=2pt,label={(\alph*)}]
    \item \label{assumption:design_adapt_con}
    The local difference sequences $d_{(1)}, \ldots, d_{(n-m)}$ satisfy
    (\ref{diff_poly_cond0}) and (\ref{diff_poly_random_cond}).

    \item \label{assumption:regular_con}
    The global difference sequence $d$ satisfies
    (\ref{diff_poly_cond0}) and (\ref{diff_poly_cond}).
\end{enumerate}

\end{assumption}

\begin{assumption}[Regular design]\label{assumption:basic}
    {$\max_{i=1, \ldots, n-1}|X_{i+1}-X_{i}| = O_p(n^{-1})$.}
\end{assumption}

Assumption \ref{ass:typeDiffSeq}\ref{assumption:regular_con} is equivalent to Assumption \ref{ass:typeDiffSeq}\ref{assumption:design_adapt_con}
either when
$X_{i}=i/n$ for $i=1,\ldots,n$
or $r=0$.
Assumption \ref{assumption:basic} regularizes a random design,
so that the gap between any consecutive $X_i$ and $X_{i-1}$
is small asymptotically.
It ensures that $D_i$ is defined with outcomes $Y_i, \ldots, Y_{i+m}$
corresponding to nearby covariates that are at most $O_p(n^{-1})$ apart from each other.
Otherwise, the difference statistic $D_i$
will not have the desired mean cancellation properties.
The order of magnitude $O_p(n^{-1})$
in Assumption \ref{assumption:basic} is sensible and mild.
By Theorem 5.1 in \citet{P1965}, Assumption \ref{assumption:basic} is satisfied if $(X_i)_{i=1}^n$ are generated from a cumulative distribution function $F$ whose $F^{-1}(u)$ is uniquely defined for $0<u<1$.
It can be verified in practice by inspecting the empirical distribution of $(X_i)_{i=1}^n$.
Assumption \ref{assumption:basic} includes the equidistant design as a special case.
Our proposed estimators are
\begin{align}
    \hat{\sigma}^2_{*}(m,r)
    &=\frac{1}{n-m}\sum^{n-m}_{i=1}\left(\sum^m_{j=0}d_{i,j}Y_{i+j}\right)^2, \quad\text{where $d_{i,j}$'s satisfy Assumption \ref{ass:typeDiffSeq}\ref{assumption:design_adapt_con},} \label{eqt:hat_m_r_star}\\
    \hat{\sigma}^2(m,r)
    &=\frac{1}{n-m}\sum^{n-m}_{i=1}\left(\sum^m_{j=0}d_jY_{i+j}\right)^2,\quad\text{where $d_j$'s satisfy Assumption \ref{ass:typeDiffSeq}\ref{assumption:regular_con}.} \label{eqt:hat_m_r}
\end{align}
The properties of difference statistics used in the above estimators
are as follows.

\begin{proposition}\label{prop:poly_cancelling}
    Suppose $(X_i)_{i=1}^n$ follows a random design.
    \begin{enumerate}
        \item Under Assumption \ref{ass:typeDiffSeq}\ref{assumption:design_adapt_con}, we have
        (a) $D_i^g=0$ if $g(x)=\sum_{\ell=0}^{r}a_\ell x^\ell$ for some $a_0, \ldots, a_r$; and
        (b) $D_i^g=O_p(n^{-r-1})$
        if $g^{(r+1)}$ exists and is continuous.
        \item Under Assumption \ref{ass:typeDiffSeq}\ref{assumption:regular_con}, we have
        (c) $D_i^g=0$ if $g(x)$ is a constant function; and
        {(d) $D_i^g=O_p(n^{-1})$ if $g(x)$ is continuously differentiable or Lipschitz continuous.}
    \end{enumerate}
\end{proposition}

{The bias-correcting parameter $r$ determines how close $D_i^g$ is to zero.
The magnitude of $D_i^g$ will then control the
biases of $\hat{\sigma}^2_{*}(m,r)$ and $\hat{\sigma}^2(m,r)$.
However, as stated in Proposition \ref{prop:poly_cancelling},
the bias correcting effect requires
that $g(\cdot)$ is sufficiently smooth.
In particular,
if the regression function $g(\cdot)$ is a $r$th order polynomial,
then differencing can perfectly cancel $g(\cdot)$.
If $g(\cdot)$ is only $r$th order differentiable,
then $g(\cdot)$ is canceled up to $r$th order.
If $g(\cdot)$ is continuously differentiable or Lipschitz continuous,
then $g(\cdot)$ is asymptotically negligible.
The proof of Proposition \ref{prop:poly_cancelling} relies on
approximating all $g(X_{i+j})$ by $g(X_{i})$ under the respective condition.}
The results of Proposition \ref{prop:poly_cancelling} (a) and (b) remain true even
under Assumption \ref{ass:typeDiffSeq}\ref{assumption:regular_con}
instead of Assumption \ref{ass:typeDiffSeq}\ref{assumption:design_adapt_con}
provided that $(X_i)_{i=1}^n$ is
equidistant.
We refer to those $d_{(i)}$ as
high-order bias correcting variance-optimal difference sequences.
We will show in \S~\ref{section_theory} that
the property of high-order accuracy leads to a high-order
corrected bias for variance estimator.

\section{Differencing for repeated data \label{sec:repeated}}
When repeated measurement exists,
applying $\hat{\sigma}^2_{*}(m,r)$
on (\ref{eqt:stackY})--(\ref{eqt:stackX})
is sub-optimal.
This is because
it does not utilize the property that $g(X_{1,i})=\cdots=g(X_{s(i),i})$ and
that repeated values in $(X_{i'})_{i\leq i'\leq i+m}$
reduce degrees of freedom for differencing;
see Example \ref{eg:diff_long} for an illustration.

\begin{example}\label{eg:diff_long}
{Let $(X_1,X_2,X_3,X_4)=(1,1,2,2)/n$ and $m=3$.
Note that $s(1)=2$ and $s(2)=2$.
The constraint (\ref{diff_poly_random_cond}) becomes
\begin{align*}
    \sum^3_{j=0}d_{1,j}(X_{1+j}-X_1)^h
    &=(d_{1,0}+d_{1,1})\left(\frac{0}{n}\right)^h+(d_{1,2}+d_{1,3})\left(\frac{1}{n}\right)^h
    =\sum^1_{j=0}d^\star_{1,j}(X^\star_{1+j}-X^\star_1)^h,
\end{align*}
where $(X^\star_1,X^\star_2)=(1,2)/n$,
$(d^\star_{1,0},d^\star_{1,1})=(d_{1,0}+d_{1,1},d_{1,2}+d_{1,3})$.
So, the two repeated values in $(X_1,X_2,X_3,X_4)$ reduce the effective differencing order from $m=3$ to $m=1$.
Since $r<m$,
this implies that
the bias correcting order $r$ must not be $1$ or $2$.
Consequently,
a larger $m$ is necessary for achieving a higher-order ($r\geq 1$) bias correcting effect
as in Proposition \ref{prop:poly_cancelling}.}
\end{example}

{
To address the problem in Example \ref{eg:diff_long},
we consider two cases to be defined below.
We handle them by (i) intra-group differencing and (ii) inter-group differencing.
Figure \ref{fig:sec:repeat:1} visualizes the concepts.

\begin{figure}[t]
        \centering
        \includegraphics[width=0.6\linewidth]{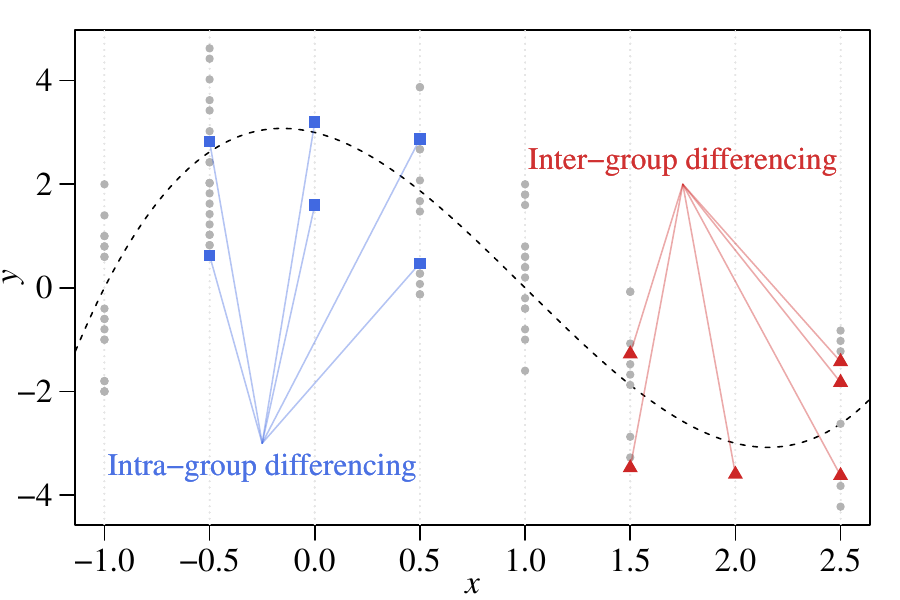}
        \caption{The nodes represent data $\{(X_i,Y_i)\}_{i=1}^{n}$, while the dashed curve denotes the regression function $y=g(x)$.
        The squares and triangles constitute
        two particular difference statistics
        corresponding to cases (i) and (ii), respectively.
        For case (i), intra-group differencing is applied to $D_i$
        by independently applying differencing to the three columns of squares.
        It cancels out the regression function $g(\cdot)$ exactly.
        In contrast,
        for case (ii),
        intra-group differencing is not feasible for the triangles because
        there is only one observation at $x=2$. Thus, inter-group differencing is needed.
        This involves averaging the observations in each column of triangles independently.
        The differencing method in \S~\ref{section_gds} can then be applied to
        the averages.
        }

        \label{fig:sec:repeat:1}
\end{figure}

First, suppose each distinct element appears at least twice in $\{X_{i},\ldots,X_{i+m}\}$.
We can perform intra-group differencing for each repeated measurement,
and then sum them up across distinct measurements.
To see it,
let $\mathcal{X}^\star_i=\{ X_{i,1}^\star,\ldots,X_{i,m(i)+1}^\star \}$\linelabel{def:Xstari}
be the set of distinct elements of $\mathcal{X}_i = \{X_{i},\ldots,X_{i+m}\}$\linelabel{def:X_i},
where $|\mathcal{X}^\star_i|=m(i)+1$ is the number of distinct elements in $\mathcal{X}^\star_i$ and $|\cdot|$ is the cardinality of a finite set.
Suppose that
each $X^\star_{i,k}$ appears $\Gamma_k-\Gamma_{k-1}$ times,
so we can partition $\{X_{i},\ldots,X_{i+m}\}$ as
\begin{align*}
        \underbracket{X_{i+\Gamma_0} , \ldots, X_{i+\Gamma_1-1}}_{=X^\star_{i,1}} , \quad
        \underbracket{X_{i+\Gamma_1} ,\ldots, X_{i+\Gamma_2-1}}_{=X^\star_{i,2}} , \quad\ldots, \quad
        \underbracket{X_{i+\Gamma_{m(i)}} ,\cdots, X_{i+\Gamma_{m(i)+1}-1}}_{=X^\star_{i,m(i)+1}} ,
\end{align*}
where $\Gamma_0=0$ and $\Gamma_{m(i)+1}=m+1$.
Performing intra-group differencing ensures that
\begin{align*}
        D^{ g}_i
        &=\underbracket{\sum_{j=0}^{m(i)}
        \Bigg\{ \underbracket{ g(X^\star_{i,j+1}) \sum_{k=\Gamma_j}^{\Gamma_{j+1}-1} d_{i,k}}_{\text{differencing}} \Bigg\}}_{\text{averaging}} =0,
\end{align*}
provided that $d_{i,0}, \ldots, d_{i,m}$ satisfy
\begin{align}
	\sum_{j=\Gamma_0}^{\Gamma_1-1}d_{i,j}
	= \sum_{j=\Gamma_1}^{\Gamma_2-1}d_{i,j}
	= \cdots
	= \sum_{j=\Gamma_{m(i)}}^{\Gamma_{m(i)+1}-1} d_{i,j}
	= 0 .\label{eqt:intra_diff}
\end{align}
In other words, the difference sequence $\{d_{i,0}, \ldots, d_{i,m}\}$
can be partitioned into $m(i)+1$ sub-sequences
$\{d_{i,\Gamma_0}, \ldots, d_{i,\Gamma_1-1}\},
\{d_{i,\Gamma_1}, \ldots, d_{i,\Gamma_2-1}\},\ldots,
\{d_{i,\Gamma_{m(i)}}, \ldots, d_{i,\Gamma_{m(i)+1}-1}\}$,
each of which is also a difference sequence.
Hence, it is called intra-group differencing.

Second,
suppose one of the elements in $\{X_{i},\ldots,X_{i+m}\}$ only appears once.
We perform weighted averaging for each repeated measurement and then perform differencing across them over the distinct design points.
Performing inter-group differencing ensures that
\begin{align}\label{eq:inter_group_diff}
        D^{ g}_i
        &=\frac{1}{c_i} \underbracket{\sum_{j=0}^{m(i)}
        \Bigg\{ \underbracket{ g(X^\star_{i,j+1})  \sum_{k=\Gamma_j}^{\Gamma_{j+1}-1} (c_id_{i,k})
          }_{\text{averaging}}
        \Bigg\}}_{\text{differencing}}
        =O_p(n^{-r-1}) ,
\end{align}
provided that $d_{i,0}, \ldots, d_{i,m}$ satisfy
\begin{align}
	\sum_{j=\Gamma_0}^{\Gamma_1-1}d_{i,j} = d_{i,0}^{\star}/c_i, \qquad
	\sum_{j=\Gamma_1}^{\Gamma_2-1}d_{i,j} = d_{i,1}^{\star}/c_i, \qquad \cdots , \qquad
	\sum_{j=\Gamma_{m(i)}}^{\Gamma_{m(i)+1}-1} d_{i,j} = d_{i,m(i)}^{\star}/c_i, \label{eqt:inter_diff}
\end{align}
where $(d_{i,0}^\star,\ldots,d_{i,m(i)}^\star)$ is a $(m(i),r)$th order difference sequence\linelabel{def:m(i)} satisfying Assumption \ref{ass:typeDiffSeq}, and
$c_i$ is a normalizing constant so that $\sum_{j=1}^md_{i,j}^2=1$.
The right-most equality in (\ref{eq:inter_group_diff}) follows from Proposition \ref{prop:poly_cancelling}.
So, the difference sequence $\{d_{i,0}, \ldots, d_{i,m}\}$
can be partitioned into $m(i)+1$ sub-sequences
$\{d_{i,\Gamma_0}, \ldots, d_{i,\Gamma_1-1}\},
\{d_{i,\Gamma_1}, \ldots, d_{i,\Gamma_2-1}\},\ldots,
\{d_{i,\Gamma_{m(i)}}, \ldots, d_{i,\Gamma_{m(i)+1}-1}\}$
so that the sub-sequence sums form a difference sequence after being normalized by $c_i$.

The $D_i$ whose difference sequence satisfies the following assumption
is said to be repetitional-bias-corrected.
\begin{assumption}[Repetitional adaptive differencing]\label{assumption:intra_inter_diff}
    Let $m\in\mathbb{N}$.
    \begin{enumerate}
        \item When $\min_{i\in\{1,\ldots,m(i)+1\}}(\Gamma_i-\Gamma_{i-1})>1$,
        the difference sequence $d_{(i)}$ satisfies (\ref{eqt:intra_diff}).
        \item When $\min_{i\in\{1,\ldots,m(i)+1\}}(\Gamma_i-\Gamma_{i-1})=1$,
        the difference sequence $d_{(i)}$ satisfies (\ref{eqt:inter_diff}).
    \end{enumerate}
\end{assumption}
Then our proposed repetitional adaptive difference-based estimator is
\begin{align}\label{eq:rep_diff_est}
    \hat{\sigma}^2_{\Long}(m,r)
    	=\frac{1}{N-m}\sum_{i=1}^{N-m}\left(\sum_{j=0}^md_{i,j}Y_{i+j}\right)^2, \quad\text{where $d_{i,j}$'s satisfy Assumption \ref{assumption:intra_inter_diff}.}
\end{align}

There are
$\prod_{i=1}^n \{s(i)!\}$
arrangements of $(Y_i,X_i)_{i=1}^N$
that define
identically distributed versions of $\hat{\sigma}^2_{\Long}(m,r)$
by the exchangeability of $Y_{1,i},\ldots,Y_{s(i),i}$ {for each $i$}.
Let $\Psi$\linelabel{def:Psi} be all such permutations of $\{Y_1,\ldots,Y_N\}$,
where only permutations within each repeated set
$\{Y_{1,i},\ldots,Y_{s(i),i}\}$ are allowed.
Denote $\psi(\cdot,h)\in\Psi$\linelabel{def:psi(dot,h)} as the $h$th permutation function.
So, $|\Psi|=\prod_{i=1}^n \{s(i)!\}$ and
$\{Y_{\psi(1,h)},\ldots,Y_{\psi(N,h)}\}$ is the $h$th permuted dataset.
We propose an efficiency-boosted estimator $\hat{\sigma}^2_{\VB}(m,r)$
by averaging all $|\Psi|$ versions of $\hat{\sigma}^2_{\Long}(m,r)$, where
\begin{align}\label{eq:repeated_vb}
    \hat{\sigma}^2_{\VB}(m,r)
    =\frac{1}{|\Psi|}\sum_{h=1}^{|\Psi|}
    \left\{\frac{1}{N-m}\sum_{i=1}^{N-m}\left(\sum_{j=0}^md_{i,j}Y_{\psi(i+j,h)}\right)^2
    \right\}.
\end{align}
Although $\hat{\sigma}^2_{\VB}(m,r)$ has the same bias properties as $\hat{\sigma}^2_{\Long}(m,r)$,
Theorem \ref{thm:repeat_varboost} in the next section suggests that
$\hat{\sigma}^2_{\VB}(m,r)$ is more precise than $\hat{\sigma}^2_{\Long}(m,r)$; see
Figure \ref{fig:sim_longitudinal} for simulation evidence.
The procedures for computing $\hat{\sigma}^2_{\Long}(m,r)$ and $\hat{\sigma}^2_{\VB}(m,r)$ are deferred to \S~\ref{sec:compute_repeat}.
}

\section{Theory \label{section_theory}}
In this section, all results are conditional upon
the design points
$X$.
Denote
\[
    B(m,r) = \frac{1}{n-m}\sum^{n-m}_{i=1}\Bigg\{\frac{g^{(r+1)}(X_{i})}{(r+1)!}\sum^m_{j=0}d_{i,j}(X_{i+j}-X_i)^{r+1}\Bigg\}^2+o_p(n^{-2r-2}) .
\]
The following theorem concerns standard designs, i.e., all design points $X_i$'s are distinct.
Two cases, random and equidistant $X_i$'s, are studied.

\begin{theorem}
    \label{thm_diff_bias}
    Let
    $m\in\mathbb{N}$ and $r\in\mathbb{N}_0$ such that $r+1\leq m$.
    If $(X_i)_{i=1}^n$ satisfies Assumption \ref{assumption:basic}
    and $g^{(r+1)}$ exists and is continuous,
    then $B(m,r) = O_p(n^{-2r-2})$ and the following results hold:
    Under a random standard design $(X_i)_{i=1}^n$,
                \[
                \Bias\left\{\hat{\sigma}^2_{*}(m,r)\right\}=B(m,r)
                \quad \text{and} \quad
                \Bias\left\{\hat{\sigma}^2(m,r)\right\}=B(m,0).
                \]
                If $g$ is a
                polynomial of degree $k\leq r$, then
                $\Bias\left\{\hat{\sigma}^2_{*}(m,r)\right\}=0$.
                If $g$ is a constant function, then
                $\Bias\left\{\hat{\sigma}^2(m,r)\right\}=0
                .$
    Under an equidistant standard design $(X_i)_{i=1}^n$,
                \[
                \Bias\left\{\hat{\sigma}^2_{*}(m,r)\right\}=B(m,r)
                \quad \text{and} \quad
                \Bias\left\{\hat{\sigma}^2(m,r)\right\}=B(m,r).
                \]
                If $g$ is a
                polynomial of degree $k\leq r$, then
                $\Bias\left\{\hat{\sigma}^2_{*}(m,r)\right\}=\Bias\left\{\hat{\sigma}^2(m,r)\right\}=0.$
\end{theorem}

Tables \ref{table:bias1}
summarizes the results of Theorem \ref{thm_diff_bias}.
Theorem \ref{thm_diff_bias} reduces to two existing results when taking
$r=0$ (i.e., variance optimal-estimator) and $r=m-1$ (i.e., bias-optimal estimator).
Corollary \ref{coro:equiv_diff_est}
shows the equivalence between
$\hat{\sigma}_*^2(m,r)$ and two traditional difference-based estimators.
Denote $\hat{\sigma}_\textsc{h}^2(m)$ and $\hat{\sigma}_\textsc{g}^2$ as the estimators
proposed by \citet{HKT1990} and \citet{GSJ1986}, respectively.
\begin{corollary}\label{coro:equiv_diff_est}
    Suppose Assumption \ref{ass:typeDiffSeq}\ref{assumption:design_adapt_con} holds, then we have (i) $\hat{\sigma}_*^2(m,0)\equiv\hat{\sigma}^2(m,0)\equiv\hat{\sigma}_\textsc{h}(m)$; and (ii) $\hat{\sigma}_*^2(2,1)\equiv\hat{\sigma}_\textsc{g}^2$.
\end{corollary}
When $r=0$, the estimator $\hat{\sigma}^2(m,0)$ is
the variance-optimal difference-based estimator proposed in \cite{HKT1990},
where $\Bias\{\hat{\sigma}^2(m,0)\} = O_p(n^{-2})$.
When $r=m-1$, the estimator $\hat{\sigma}^2_*(m,m-1)$
is the bias-optimal difference-based estimator in \citet{DMW1998},
where $\Bias\{\hat{\sigma}_{*}^2(m,m-1)\} = O_p(n^{-2m})$.

    \begin{table}[t]
    \setlength{\tabcolsep}{2pt}
    \caption{\label{table:bias1}The value of $\Bias\left\{\hat{\sigma}^2_{*}(m,r)\right\}$
    stated in Theorem \ref{thm_diff_bias}
    under various $g$, designs, and differencing methods.}
    \begin{adjustbox}{width=\textwidth}
        \begin{tabular}{ccccc}
            \toprule
             & \multicolumn{2}{c} {local differencing} & \multicolumn{2}{c} {global differencing}\\
            regression function $g$ & random design& equidistant design& random design& equidistant design\\
            \hline
            $r+1$ times continuously differentiable
            & $B(m,r)$ & $B(m,r)$ & $B(m,0)$ & $B(m,r)$\\
            polynomial of degree $r$ & $0$ & $0$ & $B(m,0)$ & $0$\\
            constant & $0$ & $0$ & $0$ & $0$\\
            \bottomrule
        \end{tabular}
    \end{adjustbox}

    \end{table}

We contribute for the cases where $0<r<m-1$.
Our proposed estimator $\hat{\sigma}_{*}^2(m,r)$ controls
the bias at order $O_p(n^{-2r-2})$, which is substantially smaller than $O_p(n^{-2})$.
So, our proposal performs better than the variance-optimal estimator $\hat{\sigma}^2(m,0)$ in finite samples,
where the bias contributed from the regression function has not yet diminished.
Meanwhile, our proposal reduces the cost of performing high-order bias correcting differencing
as we shall see in Theorem \ref{thm_diff_var} that the variance of our proposal
can be optimized and smaller than that of the bias-optimal estimator $\hat{\sigma}^2(m,m-1)$.
Theorem \ref{thm_diff_bias} also illustrates the difference between local and global differencing.
Global differencing only mitigates a part of the trend and the resulting bias remains the order of $O_p(n^{-2})$, irrespective of $r$.
Therefore, using a larger value of $r$ provides no advantage.
It is recommended to use $\hat{\sigma}_*^2(m,r)$ over $\hat{\sigma}^2(m,r)$ under a random design.

Then, we derive the variance of our estimators under standard designs.
Denote $\Diag(A)$ as the diagonal matrix with the same diagonal elements as matrix $A$
and $\lambda_h=\E\{(\epsilon_i/\sigma)^h\}$.
Let $B_1=\frac{1}{n-m}\sum_{i=1}^{n-m}D^g_i$,
$B_2=\frac{1}{n-m}\sum_{i=1}^{n-m}(D^g_i)^2$ and
$B_3=\frac{1}{n-m}\sum_{\ell=1}^m\sum_{i=1}^{n-m-\ell}D^g_iD^g_{i+\ell}$.

\begin{assumption}\label{assumption:Dgi_mean}
    For some $\kappa>0$,
    $B_1^2+B_2 + |B_3|=O_p(n^{-2\kappa})$.
\end{assumption}

The terms $B_1$ and $B_2$ are the first and second sample moments of the differenced regression function $D^g_i$,
whereas the term $B_3$ is the sum of the lagged cross moments of
$D^g_i D^g_{i+\ell}$ over the lags $\ell = 1, \ldots, m$.
They control the order of the cross terms in $\Var\{\hat{\sigma}_*^2(m,r)\}$
and $\Var\{\hat{\sigma}^2(m,r)\}$.
Assumption \ref{assumption:Dgi_mean} ensures that
$B_1$, $B_2$ and $B_3$ asymptotically contribute only $O_p(n^{-\kappa-1})$
to $\Var\{\hat{\sigma}_*^2(m,r)\}$ and $\Var\{\hat{\sigma}^2(m,r)\}$
in Theorem \ref{thm_diff_var}.
While the smoothness conditions mentioned in Theorem \ref{thm_diff_bias} are sufficient for Assumption \ref{assumption:Dgi_mean} to hold,
Assumption \ref{assumption:Dgi_mean} still holds under some weaker conditions, for example, under the existence of finitely many change points.
However,
Assumption \ref{assumption:Dgi_mean} may not hold when differencing fails to cancel a non-negligible portion of the regression function, for example, when there are infinitely many change points.
Further discussions on change points can be found in \ifnum\isXr=1{\S~\ref{section:cp}}\else{\S~B.3}\fi.

\begin{theorem}
    \label{thm_diff_var}
    Suppose $(X_i)_{i=1}^n$ satisfies Assumptions \ref{assumption:basic} and \ref{assumption:Dgi_mean}.
    Let
    $m\in\mathbb{N}$ and $r\in\mathbb{N}_0$ such that $r+1\leq m$.
    For $\tilde{\sigma}^2 \in \{ \hat{\sigma}_*^2(m,r), \hat{\sigma}^2(m,r)
    \}$,
    \begin{equation*}
        \Var(\tilde{\sigma}^2)
        =\frac{\sigma^4}{(n-m)^2}\theta(A,\lambda_4)+O_p(n^{-\kappa-1}),
    \end{equation*}
    {where
    $\theta(A,\lambda_4)=2\tr(A^2)+(\lambda_4-3)\tr\{A\Diag(A)\}$, and
    the matrix $A$ is defined in (\ref{eq:Wmat_diffseq}),
    whose value depends on the specfic form of $\tilde{\sigma}^2$.}
\end{theorem}
Theorem \ref{thm_diff_var} is an application of the variance formula shown in \citet{SGW1993}.
It can be easily seen that $\Var\left\{\hat{\sigma}_*^2(m,r)\right\}=O_p(n^{-1})$.
Assumption \ref{assumption:Dgi_mean} states the minimal order of $B_1$, $B_2$ and $B_3$ so that $\sigma^4\theta(A,\lambda_4)/(n-m)^2$ is the sole leading term.
If $g(\cdot)$ is differentiable or piece-wise-differentiable, i.e., $g(\cdot)$ satisfies
\ifnum\isXr=1{(\ref{cp_trend}) in \S~\ref{section:cp}}\else{(B.3) in \S~B.3}\fi,
Assumption \ref{assumption:Dgi_mean} and thus Theorem \ref{thm_diff_var} holds.

Next, we study the properties of
$\hat{\sigma}^2_{\Long}(m,r)$
for repeated data.
For each $r$, denote
\begin{align*}
    \ell_1&=\left\{i\in [1,n-r+1]\cap\mathbb{Z}
    :
    \min_{i'=i,\ldots,i+r-1}s(i')>1
    \right\}
    ,\\
    \ell_2&=\left\{
    i\in[1,n-r]\cap\mathbb{Z}
    :
    \min_{i'=i,\ldots,i+r-1}s(i')=1
    \right\},\\
    \mathcal{S}_1&=\left\{S(i,i+r-1)\mathbb{1}(r>0)+1:i\in \ell_1\right\}, \\
    \mathcal{S}_2&=\left\{S(i,i+r)\mathbb{1}(r>0)+\mathbb{1}(r=0):i\in \ell_2\right\}.
\end{align*}
Let $S_n=\sum_{i=1}^n\mathbb{1}\left\{s(i)=1\right\}$\linelabel{def:Sn} and
$m_{\min}=\max(\mathcal{S}_1\cup\mathcal{S}_2).$\linelabel{def:m_min}
By definition, $m_{\min}\geq r+1$.
The elements in $\mathcal{S}_1$ and $\mathcal{S}_{2}$ are
the minimal $m$ such that $D_{\ell(i)}$ satisfies
the first result of Proposition \ref{prop:poly_cancelling} $(i\in\ell_1,\ell_2)$, respectively,
i.e. $D^{ g}_{\ell(i)}=O(n^{-r-1})$.
$\ell(i)$ is a matching index such that $Y_{\ell(i)}=Y_{1,i}$ $(i=1,\ldots,n)$.

\begin{theorem}\label{thm:long_diff_bias}
    {Assume that $\sup_{i=1,\ldots,N-m}m(i)<\infty$}
    and suppose that all assumptions in Theorem \ref{thm_diff_bias} hold.
    Let $(X_i)_{i=1}^N$ be
    either balanced or unbalanced repetitional designs
    with any $s(i)\geq 1$ ($i=1, \ldots, n$).
    If $m\geq m_{\min}$,
    then $D_i$ is repetitional-bias-corrected for $i=1,\ldots,N-m$
    and
    \begin{align}
        \Bias\big\{\hat{\sigma}^2_{\Long}(m,r)\big\}
        	&= O_p\left\{(n^{-2r-1}+n^{-2r-2}S_n m)/(N-m)\right\}, \label{eq:thm2:bias}\\
        \Var\left\{\hat{\sigma}_\Long^2(m,r)\right\}
        &=\frac{\sigma^4}{(N-m)^2}\theta(A,\lambda_4)+O_p\left\{\frac{n^{-r}m+n^{-r-1}S_nm^2}{(N-m)^2}\right\},\label{eq:thm2:var}
    \end{align}
    {where $A$ is defined in (\ref{eq:Wmat_diffseq})
    and its value depends on the difference sequence $\{d_{(i)}\}_{i=1}^{N-m}$.}
\end{theorem}
Theorem \ref{thm:long_diff_bias} implies that
$\Bias\big\{\hat{\sigma}^2_{\Long}(m,r)\big\}$ is negligible even when $n<\infty$, provided that $N\rightarrow\infty$ and {$S_nm/N\rightarrow0$}.
Moreover, (\ref{eq:thm2:var}) shows that the leading order terms of
the asymptotic variances of
$\hat{\sigma}^2_{\Long}(m,r)$,
$\hat{\sigma}^2_{*}(m,r)$ and $\hat{\sigma}^2(m,r)$ are equal.
If $s(i)>1$ for all $i$, the order
in (\ref{eq:thm2:bias})
becomes $O_p\left\{n^{-2r-1}/(N-m)\right\}$.
The presence of $D_i$ with $s(i)=1$
deteriorates $\hat{\sigma}^2_{\Long}(m,r)$ since their existence increases the number of non-zero $D^{ g}_i$ and leads to a
larger bias.
However, the simulations in Figure \ref{fig:sim_longitudinal} show that the increase in bias is insignificant.
We emphasize that our estimator and results apply also to equidistant and random designs and are free of tuning parameters.

Then, we derive the variance of $\hat{\sigma}^2_{\VB}(m,r)$.
Denote $\mathcal{Q}^\star_i=\{t:X_t=X_i^\star\}$.
Let
\begin{gather*}
    \mathcal{A}_n=
    \sum_{i=1}^n \frac{1}{s(i)} \mathop{\sum\sum}_{h,k\in\mathcal{Q}^\star_i} \frac{(v_k^\T v_k- v_h^\T v_h)^2}{2},\\
    \mathcal{B}_n=
    \mathop{\sum\sum}_{i,j=1}^n \frac{1}{s(i)\left\{s(j)-\mathbb{1}(i=j)\right\}}
    \sum_{k\in\mathcal{Q}^\star_i}\sum_{h\in\mathcal{Q}^\star_j:k\neq h}
    \sum_{p\in\mathcal{Q}^\star_i}\sum_{q\in\mathcal{Q}^\star_j:q\neq p}
    \frac{(v_k^\T v_h -v_p^\T v_q)^2}{2},
\end{gather*}
where $v_k$ is the $k$th column vector of the matrix $\DiffSeqMat$ from $\hat{\sigma}^2_{\Long}(m,r)=Y^\T V^\T V Y/ \tr(V^\T V)$.

\begin{theorem}\label{thm:repeat_varboost}
    Suppose that all assumptions in Theorem \ref{thm_diff_bias} hold.
    Let $(X_i)_{i=1}^N$ be
    either balanced or unbalanced repetitional designs
    with any $s(i)\geq 1$ ($i=1, \ldots, n$).
    (i) If $m\geq m_{\min}$, then
    \begin{align*}
        \Var\left\{\hat{\sigma}_\VB^2(m,r)\right\}
        &=\Var\left\{\hat{\sigma}_\Long^2(m,r)\right\}
        -\frac{\sigma^4\left\{\mathcal{A}_n(\lambda_4-1)+2\mathcal{B}_n\right\}}{(N-m)^2}
        +O_p\left\{\frac{n^{-r}m+n^{-r-1}S_nm^2}{(N-m)^2}\right\}  .
    \end{align*}
    (ii) If $n=1$ and $0<m\leq N-1$, then
    $$
        \hat{\sigma}_\VB^2(m,r)\equiv\frac{1}{N-1} \sum_{i=1}^N( Y_i-\bar{Y} )^2,
        \qquad \text{where} \qquad
        \bar{Y} = \sum_{j=1}^NY_{j}/N.
    $$
\end{theorem}
Theorem \ref{thm:repeat_varboost} implies that
$\Var\{\hat{\sigma}_\VB^2(m,r)\}\leq\Var\{\hat{\sigma}_\Long^2(m,r)\}$
asymptotically because $\lambda_4\geq 1$ by Jensen's inequality, $\mathcal{A}_n\geq 0$, and $\mathcal{B}_n\geq 0$.
The amount of efficiency boosted by $\hat{\sigma}_\VB^2(m,r)$ depends on $\mathcal{A}_n$ and $\mathcal{B}_n$.
The expressions $\mathcal{A}_n$ and $\mathcal{B}_n$ can be
interpreted as
weighted sum of sample variances in the form of V-statistics:
\begin{gather*}
    \mathcal{A}_n=
    \sum_{i=1}^ns(i)\frac{1}{s(i)^2}\sum_{k\in\mathcal{Q}^\star_i}\sum_{h\in\mathcal{Q}^\star_i}\frac{(v_k^\T v_k- v_h^\T v_h)^2}{2},\\
    \mathcal{B}_n=
    \sum_{i=1}^n\sum_{j=1}^n
    \frac{s(i)\left\{s(j)-\mathbb{1}(i=j)\right\}}{s(i)^2\left\{s(j)-\mathbb{1}(i=j)\right\}^2}\sum_{k\in\mathcal{Q}^\star_i}\sum_{h\in\mathcal{Q}^\star_j:k\neq h}\sum_{p\in\mathcal{Q}^\star_i}\sum_{q\in\mathcal{Q}^\star_j:p\neq q}\frac{(v_k^\T v_h -v_p^\T v_q)^2}{2}.
\end{gather*}
From the above expressions, $\mathcal{A}_n$ and $\mathcal{B}_n$
are the weighted sum of sample variances of $(v_k^\T v_k)_{k\in\mathcal{Q}^\star_i}$ and $(v_k^\T v_h)_{k\in\mathcal{Q}^\star_i,h\in\mathcal{Q}^\star_j:k\neq h}$ for $i,j=1,\ldots,n$, respectively.
So, the more dispersed the $v_k$'s are
the more precise the $\hat{\sigma}_\VB^2(m,r)$ is.
When $n=1$, i.e., $X_1=\cdots=X_N$,
$\hat{\sigma}_\VB^2(m,r)$ does not depend on $r$ since intra-group differencing is applied to all difference statistics.
Also, $\hat{\sigma}_\VB^2(m,r)$ does not depend on $m$ since it reduces to the sample variance.

Besides, when the repetitional design is extremely imbalanced, for example, $s(i)=5$ for all $i$ except $s(1)=100$, the lower bound $m_{\min}$ stated in Theorem \ref{thm:long_diff_bias} may be too large and lead to possibly increased bias.
To address such issue, we propose to exclude the $s(i)$ that are larger than a threshold $s_{\max}$,
i.e.,
computing $m_{\min}$ based on
$\{s'(i)\}_{i=1}^{n-|\mathcal{T}_s|} \equiv \{s(i)\}_{i\in [1,\ldots,n]\cap\mathbb{Z}\setminus \mathcal{T}_s},$
where ${\mathcal{T}_s}=\{i\in[1,\ldots,n]\cap\mathbb{Z}:s(i)>s_{\max}\}.$
Denote the adjusted lower bound for $m$ as $m'_{\min}$.
After that, set $D_i=0$ if $D_i$ is no longer
repetitional-bias-corrected with $m\geq m'_{\min}$.
In other words, set $(d_{i,0},\ldots,d_{i,m})^\T=(0,\ldots,0)^\T$ if $m(i)<r+1$.
More explicitly,
define
\begin{align*}
    \ell'_1&=\left\{i\in\{1,\ldots,n-|\mathcal{T}_s|-r+1\}
    :
    \min_{i'=i,\ldots,i+r-1}s'(i')>1
    \right\}
    ,\\
    \ell'_2&=\left\{
    i\in\{1,\ldots,n-|\mathcal{T}_s|-r\}
    :
    \min_{i'=i,\ldots,i+r-1}s'(i')=1
    \right\},\\
    \mathcal{S}'_1&=\left\{S'(i,i+r-1)\mathbb{1}(r>0)+1:i\in \ell'_1\right\},\\
    \mathcal{S}'_2&=\left\{S'(i,i+r)\mathbb{1}(r>0)+\mathbb{1}(r=0):i\in \ell'_2\right\},
\end{align*}
where $S'(i,i+a)=\sum_{k=i}^{i+a}s'(k)$.
Then $m'_{\min}=\max(\mathcal{S}'_1\cup\mathcal{S}'_2).$
\linelabel{def:m_min'}
The proposed estimator becomes
\begin{align}\label{eq:repeated_m_prime}
    \hat{\sigma}^2_{\Long}(m'_{\min},r)=\frac{1}{|\mathcal{I}|}\sum_{i\in\mathcal{I}}D_i^2,
\end{align}
where the set $\mathcal{I}$ consists of indices $i$ such that $D_i$ is repetitional-bias-corrected.
Our suggested rule of thumb is to set $s_{\max}$ as
$s_{0.95}$,
the $95\%$ upper quantile of $\{s(i)\}_{i=1}^n$.
Corollary \ref{coro:long_adj_bias}
states that the asymptotic properties are the same as in Theorem \ref{thm:long_diff_bias}.
\begin{corollary}\label{coro:long_adj_bias}
    Suppose that all assumptions in Theorem \ref{thm:long_diff_bias} holds.
    Assume further that
    $S'_n/N\rightarrow 0$,
    where
    \begin{align*}
        S'_n=\sum_{i\in\mathcal{T}_s}\sum_{a=1}^rs(i-a)^{b(i,-,a)}+s(i+a)^{b(i,+,a)}
        \quad \text{and}\quad
        b(i,\pm,a)=\mathbb{1}\left\{\min_{k=i\pm 1,\ldots,i\pm a}s(k)=1\right\}.
    \end{align*}
    If $m\geq m'_{\min}$,
    then the bias of $\hat{\sigma}^2_\Long(m'_{\min},r)$ is of the order stated in (\ref{eq:thm2:bias}) and its variance is asymptotically equal to the right hand side of (\ref{eq:thm2:var}).
\end{corollary}
The expression $S'_n$ is
the maximum number of $D_i$ that are dropped.
For example, if $r=1$ and $s_{\max}=s_{0.95}$,
then $S_n' = 2(0.05n)$.
We suggest picking $(m,r)=(m'_{\min},1)$ which yields decent results according to our empirical experience.
See Figure \ref{fig:long_supp_R} for simulation evidence.

\section{Implementation \label{section_algorithm}}
\subsection{Angular representation}
In this section, we present a method to compute $\{d_{(i)}\}_{i=1}^{N-m}$ for $\hat{\sigma}^2_{*}(m,r)$ as well as $d$ for $\hat{\sigma}^2(m,r)$.
First, we introduce a representation of our proposed high-order accurate difference sequences.
Denote $\mathcal{L}_i'$\linelabel{def:Lrref} as the reduced row echelon form of
any matrix $\mathcal{L}_i$ \citep[Chapter~2]{M2023}.
For simplicity, we write the $(p,q)$th elements of $\mathcal{L}_i$ and $\mathcal{L}_i'$
as $L_{i,p,q}$ and $L_{i,p,q}'$, respectively.
Also denote $L'_{i,p}=(L'_{i,p,r+2},\ldots,L'_{i,p,m+1})^{\T}$.
The first constraints in (\ref{diff_poly_cond0}) and (\ref{diff_poly_random_cond}) can be represented in a matrix form:
    $\mathcal{L}_id_{(i)}
    =
     (0,\ldots,0)^\T
    $
where $L_{i,p,q}=(X_{i+q-1}-X_{i})^{p-1}$ for $p=1,\ldots,r+1$ and $q=1,\ldots,m+1$.
Note that the submatrix $\mathcal{M}$ formed by the first $r+1$ columns of $\mathcal{L}$ is a square Vandermonde matrix \citep[Example 4.3.4]{M2023}.
Since $0\leq X_1< \cdots<X_n\leq 1$ for standard design,
$\mathcal{M}$ is non-singular and hence $\mathcal{L}_i$ is always full row rank.
The system of equation $\mathcal{L}_i'd_{(i)}=(0,\ldots,0)^\T$, which admits the form:
\begin{align*}
    \begin{pmatrix}
        1 & 0  & \cdots & 0 & L'_{i,1,r+2} & \cdots & L'_{i,1,m+1}\\
        0 & 1 & \cdots  & 0 & L'_{i,2,r+2} & \cdots & L'_{i,2,m+1}\\
        \vdots & \vdots & \ddots & \vdots  & \vdots & \ddots & \vdots \\
        0 & 0 & \cdots & 1 & L'_{i,r+1,r+2} & \cdots & L'_{i,r+1,m+1}
    \end{pmatrix}
    \begin{pmatrix}
        d_{i,0}\\
        \vdots\\
        d_{i,m}
    \end{pmatrix}
    =\begin{pmatrix}
        0\\
        \vdots\\
        0
    \end{pmatrix},
\end{align*}
where the submatrix of the first $r+1$ columns is an identity matrix.
This form facilitate us to solve the systems of equation $\mathcal{L}_id_{(i)}=(0,\ldots,0)^\T$.
See  {\ifnum\isXr=1{\S~\ref{miscell:RREF}}\else{\S~B.6}\fi} for
a more detailed review of reduced row echelon form.
We can express $d_{i,0},\ldots,d_{i,r}$ as a function of the free parameters.
Lemma \ref{lemma:diff_sol} represents
each of $\{d_{(i)}\}_{i=1}^{n-m}$
in a $m$-dimensional angular coordinate system.
\begin{lemma}\label{lemma:diff_sol}
For $j\in\{ 0, r+1\}$, define the function $u_j :[0,\pi]^{m-j} \rightarrow \mathbb{R}^{m-j+1}$ as
\begin{equation}\label{patternvector}
    {u}_{j}(\phi_{j+1},\ldots,\phi_{m}) =
    \begin{pmatrix}
        \cos\phi_{j+1}  \\
        \sin\phi_{j+1}\cos\phi_{j+2} \\
        \vdots \\
        \sin\phi_{j+1}\sin\phi_{j+2}\cdots \cos\phi_{m}\\
        \sin\phi_{j+1}\sin\phi_{j+2}\cdots \sin\phi_{m}\\
    \end{pmatrix}.
\end{equation}
The solution to (\ref{diff_poly_cond0}) and (\ref{diff_poly_random_cond})
is given by
\[
    (d_{i,0}, \ldots, d_{i,m})^{\T}
    = u_{0}(\phi_{i,1},\ldots,\phi_{i,m}),
\]
where
the parameters $\phi_{i,m}, \ldots, \phi_{i,r+2} \in [0, \pi]$
are free to vary if $m\geq r+2$,
while
the remaining parameters $\phi_{i,r+1}, \ldots, \phi_{i,1}$ are
determined iteratively for $j=r+1,\ldots,1$ as follows:
\begin{align}\label{independent_phi_formula}
    \varphi_{i,j} &= \prod_{b=\min(r+1,j+1)}^{r+1}(\sin\phi_{i,b})^{\mathbb{1}(j<r+1)}, \nonumber \\
    \phi_{i,j} &= \left\{
    \begin{array}{ll}
    \arctan\left[ \left\{ -{{u}_{r+1}(\phi_{i,r+2},\ldots,\phi_{i,m})^\T L'_{i,j} \varphi_{i,j} } \right\}^{-1}\right]
    & \text{if $m>r+1$}; \\[2ex]
    \arctan\left\{ \left( -{L'_{i,j,m+1}\varphi_{i,j} }\right)^{-1} \right\}
    & \text{if $m=r+1$}.
    \end{array}
    \right.
\end{align}

\end{lemma}
The second constraint in (\ref{diff_poly_cond0}) is automatically satisfied by the representation (\ref{patternvector}).
The polar angles in (\ref{independent_phi_formula}) ensures $d_{(i)}$ satisfies the first constraints in (\ref{diff_poly_cond0}) and all constraints in (\ref{diff_poly_cond}).
Lemma \ref{lemma:diff_sol} provides a computationally feasible way to solve $d_{(i)}$ iteratively:
$\phi_{i,j}$ can be solved from $j=r+1$ to $j=1$ and then $d_{(i)}=u_{i,0}$ gives our target difference sequence.
The following example illustrate the application of Lemma \ref{lemma:diff_sol}.
\begin{example}\label{eg:rref}
    Suppose $i=1$, $m=3$, $r=1$ and $(X_1,X_2,X_3,X_4)=(0.1,0.3,0.4,0.7)$.
    Then the bias-correcting constraints in (\ref{diff_poly_cond0}) and (\ref{diff_poly_random_cond}) can be written in
    matrix form as
    \begin{align*}
        \begin{pmatrix}
            1 & 1 & 1 & 1 \\
            0 & 0.2 & 0.3 & 0.6
        \end{pmatrix}
        \begin{pmatrix}
            d_{1,0}\\
            d_{1,1}\\
            d_{1,2}\\
            d_{1,3}\\
        \end{pmatrix}
        &=
        \begin{pmatrix}
            0\\
            0
        \end{pmatrix}
    \qquad \text{or} \qquad
        \begin{pmatrix}
            1 & 0 & -0.5 & -2 \\
            0 & 1 & 1.5 & 3 \\
        \end{pmatrix}
        \begin{pmatrix}
            d_{1,0}\\
            d_{1,1}\\
            d_{1,2}\\
            d_{1,3}\\
        \end{pmatrix}
        =
        \begin{pmatrix}
            0\\
            0
        \end{pmatrix},
    \end{align*}
    where the second matrix equation is
    the reduced row echelon form of the first matrix equation.
    By (\ref{independent_phi_formula}),
    $\phi_{1,3}$ is a free parameter, whereas
    $\phi_{1,1}$ and $\phi_{1,2}$ are computed as
    \begin{align*}
        \phi_{1,2}&=\arctan\left[\left\{
        -\begin{pmatrix}
            \cos \phi_{1,3} & \sin \phi_{1,3}
        \end{pmatrix}
        \begin{pmatrix}
            1.5\\
            3
        \end{pmatrix}
        \right\}^{-1}\right], \\
        \phi_{1,1}&=\arctan\left[\left\{
        -\begin{pmatrix}
            \cos \phi_{1,3} & \sin \phi_{1,3}
        \end{pmatrix}
        \begin{pmatrix}
            -0.5\\
            -   2
        \end{pmatrix}
        \sin\phi_{1,2}
        \right\}^{-1}
        \right].
    \end{align*}
    Finally, using (\ref{patternvector}),
    the difference sequence can be represented as follows:
    \[
    d_{(1)}
    =\begin{pmatrix}
        \cos \phi_{1,1}\\
        \sin \phi_{1,1}\cos \phi_{1,2}\\
        \sin \phi_{1,1}\sin \phi_{1,2}\cos \phi_{1,3}\\
        \sin \phi_{1,1}\sin \phi_{1,2}\sin \phi_{1,3}
    \end{pmatrix}.
    \]

\end{example}

\begin{table}[t]
    \def~{\hphantom{0}}
    \caption{\label{table:diff_seq}Optimal $m$th order global difference sequences of $r$th order smoothness.
    The variance-optimal global differencing \citep{HKT1990} is achieved when $r=0$.
    The bias-optimal global differencing is achieved when $r=m-1$.
    The estimators are sorted in ascending
    order of the asymptotic value of $\Var\left\{\hat{\sigma}^2(m,r)\right\}$, where $\lambda_4' = \lambda_4-3$ and $\lambda_4=\E\{(\epsilon_i/\sigma)^4\}$.
    }
    \begin{adjustbox}{width=\textwidth}
        \begin{tabular}{llllll}
            \toprule
            $m$ & $r$ & $(d_0,\ldots, d_m)$ & Variance & Bias  & Reference\\[4pt]
            \midrule
             $4$ & $0$ & $(0.271,-0.014,0.691,-0.486,-0.462)$ & $\sim \sigma^4(2.250+\lambda_4')/n$ & $O(n^{-2})$ & \citet{HKT1990}\\
             $3$ & $0$ & $(0.194,0.281,0.383,-0.858)$ & $\sim \sigma^4(2.333+\lambda_4')/n$ & $O(n^{-2})$ &  \citet{HKT1990}\\
             $2$ & $0$ & $(-0.809,0.500,0.309)$ & $\sim \sigma^4(2.500+\lambda_4')/n$ & $O(n^{-2})$ &  \citet{HKT1990}\\
             \multirow{2}{*}{$11$}
             & \multirow{2}{*}{$3$} & $(0.470, -0.795, -0.052, 0.223, 0.228, 0.117,$ & \multirow{2}{*}{$\sim \sigma^4(2.537+\lambda_4')/n$} & \multirow{2}{*}{$O(n^{-8})$} &  \multirow{2}{*}{Proposal}\\
              & & $-0.007, -0.087, -0.104, -0.068, 0, 0.076)$ & & &\\
             $4$ & $1$ & $(0.616,-0.732,-0.185,0.103,0.198)$ & $\sim \sigma^4(2.685+\lambda_4')/n$ & $O(n^{-4})$ &  Proposal\\
             $1$ & $0$ & $(-0.707,0.707)$& $\sim \sigma^4(2.999+\lambda_4')/n$ & $O(n^{-2})$ &  \citet{R1984} \\
             $3$ & $1$ & $(0.535,-0.802,0.000,0.267)$ & $\sim \sigma^4(3.001+\lambda_4')/n$ & $O(n^{-4})$ &  Proposal (Default choice)\\
             $4$ & $2$ & $(0.326,-0.793,0.424,0.227,-0.184)$ & $\sim \sigma^4(3.433+\lambda_4')/n$ & $O(n^{-6})$ &  Proposal\\
             $5$ & $3$ & $(0.190, -0.644, 0.676, -0.065, -0.273, 0.116)$ & $\sim \sigma^4(3.816+\lambda_4')/n$ & $O(n^{-8})$ & Proposal\\
             $2$ & $1$ & $(-0.408,0.816,-0.408)$ &  $\sim \sigma^4(3.884+\lambda_4')/n$ & $O(n^{-2})$ &  \citet{GSJ1986} \\
             $3$ & $2$ & $(0.224,-0.671,0.671,-0.224)$ & $\sim \sigma^4(4.627+\lambda_4')/n$ & $O(n^{-6})$ &  \citet{DMW1998}\\
             $4$ & $3$ & $(0.120,-0.478,0.717,-0.478,0.120)$ & $\sim \sigma^4(5.256+\lambda_4')/n$ & $O(n^{-8})$ &  \citet{DMW1998}\\
            \bottomrule
        \end{tabular}
    \end{adjustbox}

\end{table}

\subsection{Intra-group differencing, inter-group differencing and variance boosting \label{sec:compute_repeat}}
{First, we outline the implementation details for intra-group differencing and inter-group differencing.
The difference statistic $D_i$ based on the stacked dataset (\ref{eqt:stackY})
may correspond to duplicated design points
$\mathcal{X}_i = \{ X_{i}, \ldots, X_{i+m} \}$.
Recall that
$\mathcal{X}^\star_i=\{ X_{i,1}^\star,\ldots,X_{i,m(i)+1}^\star \}$
is defined as the set of distinct elements of $\mathcal{X}_i$,
where $|\mathcal{X}^\star_i| = m(i)+1$ is the number of distinct elements and $m(i)$ can be interpreted as the actual differencing order of $d_{(i)}$.
Let $t(i,j)$\linelabel{def:t(i,j)} be a matching index such that $X_{i+j}=X^\star_{i,t(i,j)}$.
Also let
\begin{align}
    \Lambda_j(i)&=\left\{
    j'\in[0,m]\cap\mathbb{Z}:
    X_{i+j'}=X^\star_{i,1+j}
    \right\}
    \quad(j=0,\ldots,m(i)), \label{def:Lambdaj(i)}\\
    \mathcal{I}_m&=\left\{
    i\in[1,N-m]\cap\mathbb{Z} :
    \min_{j=0, \ldots, m(i)} |\Lambda_j(i)|> 1
    \right\}. \label{def:Im}
\end{align}
So,
$\{i,\ldots,i+m\}$ can be partitioned as $\{i,\ldots,i+m\}=\bigcup_{j=0}^{m(i)}\Lambda_j(i)$.
The $(j+1)$th distinct element in $\mathcal{X}^\star_i$ appears $|\Lambda_j(i)|$ times.
The set $\mathcal{I}_m$ contains all $i$ such that
each design point in $\mathcal{X}_i$ is duplicated.
According to (\ref{eqt:intra_diff}) and (\ref{eqt:inter_diff})},
our proposed repetitional difference sequence $d_{(i)} = (d_{i,0}, \ldots, d_{i,m})^{\T}$
is defined
in two cases, $i\in\mathcal{I}_m$ (intra-group differencing) and $i\not\in\mathcal{I}_m$ (inter-group differencing):
\begin{align}\label{eqt:long_ds}
    d_{i,j}=d^\circ_{i,j}/c_i,
    \qquad \text{where} \qquad
    d^{\circ}_{i,j}
        =\left\{
        \begin{array}{ll}
        w_{i,j} & \quad\text{if $i\in\mathcal{I}_m$};\\
        w_{i,j}d^\star_{i,t(i,j)} & \quad\text{if $i\not\in\mathcal{I}_m$},
        \end{array}
        \right.
\end{align}
$c_i = \{ {\sum_{j=1}^m(d^{\circ}_{i,j})^2} \}^{1/2}$ is a normalizing constant so that
$\sum_{j=1}^md_{i,j}^2=1$, and
$(d^\star_{i,0},\ldots,d^\star_{i,m(i)})$ is either a local or global difference sequence defined in Assumption \ref{ass:typeDiffSeq}.
The parameters $w_{i,j}$ in (\ref{eqt:long_ds})
are free parameters to be optimized and satisfy the constraint
\begin{align}
    \sum_{t\in\Lambda_j(i)}w_{i,t}
    =\left\{
        \begin{array}{ll}
        0 & \quad\text{if $i\in\mathcal{I}_m$};\\
        1 & \quad\text{if $i\not\in\mathcal{I}_m$},
        \end{array}
        \right.
    \quad
    (j=0,\ldots,m(i),i=1,\ldots,N-m).\label{eq:repeated_weight_constraints}
\end{align}

Then we present the details for computing $\hat{\sigma}_{\VB}^2(m,r)$.
The estimator $\hat{\sigma}^2_{\VB}(m,r)$ can be rewritten in the quadratic form (\ref{quadratic_form}):
$\hat{\sigma}_{\VB}^2 (m,r)=Y^{\T}\Bar{A} Y / \tr(\Bar{A})$.
Denote $\Bar{a}_{i,j}$ as the $(i,j)$th elements of $\Bar{A}$,
$\mathcal{Q}_i=\{t:X_t=X_i\}$\linelabel{def:Qi} as the set of indices of covariates that are equal to $X_i$
and $\alpha(i)$\linelabel{def:alpha(i)} as the matching index such that $X_{\alpha(i)}^\star=X_i$.
For $i\neq j$,
\begin{gather}\label{eq:vb_Amat}
\begin{aligned}
    \Bar{a}_{i,i}&=
        \frac{1}{s\left\{\alpha(i)\right\}}\sum_{k\in \mathcal{Q}_i}v_k^\T v_k, \\
        \Bar{a}_{i,j} &=\frac{1}{
        s(\alpha(i))
        \left[s(\alpha(j))-\mathbb{1}\left\{\alpha(i)=\alpha(j)\right\}\right]
        }\sum_{k\in\mathcal{Q}_i}\sum_{\ell\in\mathcal{Q}_j:\ell\neq k}v_k^\T v_\ell,
\end{aligned}
\end{gather}
where $v_k$ is the $k$th column vector of matrix $\DiffSeqMat$ from $\hat{\sigma}^2_{\Long}(m,r)=Y^\T V^\T V Y/ \tr(V^\T V)$.

\subsection{Optimization}
It remains to find the best set of free parameters for $\hat{\sigma}^2_{*}(m,r)$, $\hat{\sigma}^2(m,r)$ and $\hat{\sigma}^2_{\Long}(m,r)$.
Define
\begin{align}
    \Phi_*&=\left\{\phi_{i,j}:i=1,\ldots,n-m,j=r+2,\ldots,m\right\}, \label{eqt:freePara1}\\
    \Phi&=\left\{\phi_{\cdot,j}:j=r+2,\ldots,m\right\},
    \label{eqt:freePara2}
    \\
    \Phi_{\Long}&= \left\{(\phi_{i,r+2},\ldots,\phi_{i,m(i)})^\T\in[0,\pi]^{m(i)-r-1}:i\notin\mathcal{I}_m,
     r< m(i)-1
    \right\},
    \label{eqt:freePara3}\\
     \mathcal{W} &= \left\{(w_{i,0},\ldots,w_{i,m})^\T\in\mathbb{R}^{m+1}:i=1,\ldots,N-m\right\}
    \label{eqt:freePara4},
\end{align}
The index $i$ is omitted in $\phi_{\cdot,j}$ since it is redundant for global differencing.
The constrained parameters in (\ref{independent_phi_formula})
are needed for bias correction,
whereas the free parameters in (\ref{eqt:freePara1})--(\ref{eqt:freePara4}) help reduce the variance.
In the case where all of $\{d_{(i)}\}_{i=1}^{N-m}$ satisfies the same set of constraints,
i.e., when $r=0$ or under an equidistant standard design,
we have as $n\rightarrow\infty$ that
\begin{equation*}
    \frac{1}{n-m}\tr\{A\Diag(A)\}\rightarrow 1
    \qquad \text{and}\qquad
    \frac{1}{n-m}\tr(A^2) \rightarrow 1+2\delta(d),
\end{equation*}
which imply
\begin{equation}\label{simple_asy_var}
    \Var\{\hat{\sigma}_*^2(m,0)\}
    =\Var\{\hat{\sigma}^2(m,0)\}
		= \frac{\sigma^4}{n-m}\left\{4\delta(d)+\lambda_4-1\right\}+O_p(n^{-2}),
\end{equation}
where $\delta(d)=\sum^m_{k=1}(\sum^{m-k}_{j=0}d_{j}d_{j+k} )^2$.
Therefore, the variance-optimal difference sequences in \citet{HKT1990} minimize $\delta(d)$, the asymptotic variance.
However,
(\ref{simple_asy_var}) does not hold in general
and there exist no simple asymptotic forms of $\Var\{\hat{\sigma}_*^2(m,r)\}$ and $\Var\{\hat{\sigma}_\Long^2(m,r)\}$.
We propose to optimize $\theta(A,\lambda_4)$ in $\Var\{\hat{\sigma}_*^2(m,r)\}$ and $\Var\{\hat{\sigma}_{\Long}^2(m,r)\}$,
the dominant term in Theorem \ref{thm_diff_var} and Theorem \ref{thm:long_diff_bias}.
It can be regarded as the finite-sample variance
in the sense that $A$ depends on the sample size $N$.
It is still challenging to minimize $\theta(A,\lambda_4)$ due to the nuisance parameter $\lambda_4$.
We suggest two approaches:
{approach (i) performs minimization over an upper bound of $\theta(A,\lambda_4)$, while
approach (ii) performs minimization over $\theta(A,\hat{\lambda}_4)$, where $\hat{\lambda}_4$ is a consistent estimator of $\lambda_4$.}

In approach (i),
it can be easily seen that $\tr(A^2)\geq\tr\{A\Diag(A)\}$.
Thus, it is equivalent to minimizing $\tr(A^2)$,
which is a simpler objective function as it does not depend on the data through the unknown $\lambda_4$.
In particular, the upper bound is sharp if $\lambda_4=3$.
Approach (ii) minimizes an estimate of $\theta(A,\lambda_4)$. We propose to estimate $\lambda_4$ and $\sigma^2$ by the following pilot estimators:
\begin{eqnarray*}
    \hat{\lambda}^\dag_4&=&\frac{1}{N-m^\dag}\sum^{N-m^\dag}_{i=1}\left\{\left(\sum^{m^\dag}_{j=0}d_{i,j}^*Y_{i+j}\right)^{4}
    -6\sum^{m^\dag-1}_{j=0}\sum^{m^\dag}_{h=j+1}d_{i,j}^{\dag 2}d_{i,h}^{\dag 2}(\hat{\sigma}^{\dag 2})^2\right\}
    /\left(\hat{\sigma}^{\dag 2}\right)^2,\\
    \hat{\sigma}^{\dag 2}&=&\frac{1}{N-m^\dag}\sum^{N-m^\dag}_{i=1}\left(\sum^{m^\dag}_{j=0}d^\dag_{i,j}Y_{i+j}\right)^{2},
\end{eqnarray*}
for some difference sequences $\{d^\dag_{(i)}\}^{N-m^\dag}_{i=1}$ of order $m^\dag$ and $d_{i,j}^*=d^\dag_{i,j}/\left(\sum^{m^\dag}_{j=0}d_{i,j}^{\dag 4}\right)^{1/4}$.
The following theorem states the consistency of $\hat{\lambda}^\dag_4$ and $\hat{\sigma}^{\dag 2}$.
\begin{theorem}\label{thm:lamda4_est}
   Suppose
    Assumption \ref{assumption:basic} holds.
    Assume
    $g(\cdot)$ is continuously differentiable and
    $\E(\epsilon_i^8/\sigma^8)<\infty$.
    Then, $\hat{\lambda}^\dag_4\rightarrow\lambda_4$ and $\hat{\sigma}^{\dag 2}\rightarrow\sigma^2$
    in probability.
\end{theorem}

In practice, the time cost for optimization of both approaches are comparable in one-off implementations.
However, approach (i) is recommended
when it is used repeatedly as an intermediate step in other procedures.
Because it only requires a single optimization before the experiment,
while approach (ii) requires tuning whenever a new dataset is observed.
We suggest to adopt the coordinate descent method, i.e., iteratively optimize each free parameter until convergence.

\section{Extensions and discussion \label{section_extension}}
First, our framework applies to multivariate data.
One may construct the difference statistics using $m$ nearest neighbors, i.e.,
$$
    \tilde{\sigma}^2_{\text{mNN}}=\frac{1}{n}\sum_{i=1}^n
    \left\{ \sum_{j=0}^md_jY_{\ell(i,j)} \right\}^2,
    \quad
    \ell(i,j)=
    \underset{1\leq p\leq n,p\neq\ell(i,1),\ldots,\ell(i,j-1)}{\argmin}
    \lVert X_{i}-X_p\rVert,
$$
{where
$Y_{\ell(i,0)}\equiv Y_i$ and $Y_{\ell(i,j)}$ is the $j$th nearest neighbour of the observe vector $Y_i$.}
See \citet{LCL2010} for more details.

Second, our framework applies to
regression models with time-series errors:
$
    Y_i=g(X_i)+\varepsilon_i
$
$(i=1,\ldots,n)$,
where $(\varepsilon_i)_{i\in\mathbb{Z}}$ is a zero-mean stationary time series noise sequence.
To estimate the long-run variance $v=\lim_{n\rightarrow\infty} n\Var(\sum^n_{i=1}\varepsilon_i/n)$, \citet{C2022} and \citet{wangchan2026tight} proposed
$$\hat{v}=\sum_{|k|<\ell}K(k/{\ell})
    \frac{1}{n}\sum^{n-mh-|k|}_{i=1}
     \left(\sum_{j=0}^md_jY_{i+jh} \right)
     \left(\sum_{j=0}^md_jY_{i+|k|+jh}\right),$$
where $K$ is a kernel,
$\ell$ is a bandwidth parameter, and $h$ is a lag parameter.

Third, the new difference sequences are also applicable to any methods using difference sequences to remove the nuisance trend;
see \citet{TW2005} and \citet{DTZ2017} for least square estimators which apply linear regression to
$\{\hat{\sigma}^2(i,r)\}_{i=1}^{m_n}$ for some bandwidth $m_n$ and fixed $r$.
Other variance estimation problems that fit into our framework
include, e.g., heteroscedasticity \citep{BL2007,MS1987},
complex surveys \citep{KG2002}, multivariate data \citep{MBWF2005},
trend variance test \citep{TC2023}, and variance change point test \citep{leungchan2026}.

Fourth,
in \ifnum\isXr=1{\S~\ref{miscell:SignOrder}}\else{\S~B.1}\fi,
we construct the difference sequences by
using their sign-inverted or order-reversal versions.
In
\ifnum\isXr=1{\S~\ref{miscell:DiffResid}}\else{\S~B.4}\fi,
we state a differenced-residual-based estimator
that has a higher bias-correcting power.
In
\ifnum\isXr=1{\S~\ref{miscell:PolyRegFit}}\else{\S~B.5}\fi,
we discuss variance estimators
with a parametric pre-fitting.

\section{Empirical Studies\label{section_experiment}}
\subsection{
Simulation experiments \label{sec:non_equ_design_exp}
}
We perform simulation experiments to compare the performance of
various variance estimators, including
(i) the proposed differencing method $\hat{\sigma}^2_*(3,1)$,
(ii) the variance-optimal differencing method $\hat{\sigma}^2(3,0)$,
(iii) the bias-optimal differencing method $\hat{\sigma}^2_*(3,2)$,
(iv) the least squares estimator $\hat{\sigma}^2_{\textsc{ls}}$ proposed in \citet{DTZ2017},
with parameters tuned by their suggested plateau method, and
(v) the nonparametric estimator based on local-linear fitting
$\hat{\sigma}^2_{\textsc{ll}}$, defined in the example in \S~\ref{section_review_diff}.

Let $\epsilon_i\sim\Normal(0,\sigma^2)$ independently,
where $\sigma\in \{0.2,0.6,1.0, \ldots,2.2\}$.
The design points $(X_i)_{i=1}^n$ are generated from
the beta distribution with parameters $0.5$ and $0.5$.
We consider
$n=\lfloor 75\times 1.5^v \rfloor$ for $v=0,\ldots,6$, and
set $s(1)=\cdots=s(n)=1$.
The following
regression functions
are studied:
(T1) monotonic function: $g_1(x)=13\exp(20x -10)^2/\left\{1+\exp(20x-10)\right\}^2$,
(T2) oscillating function: $g_2(x)=2\exp(1-x)\sin(4\pi x)$, and
(T3) oscillating function superimposed with an increasing trend: $g_3(x)= 5\exp(x)+\sin(ex)+\cos(8\pi x)$.

\begin{figure}[t]
  \begin{center}
  \includegraphics[width=.95\textwidth]{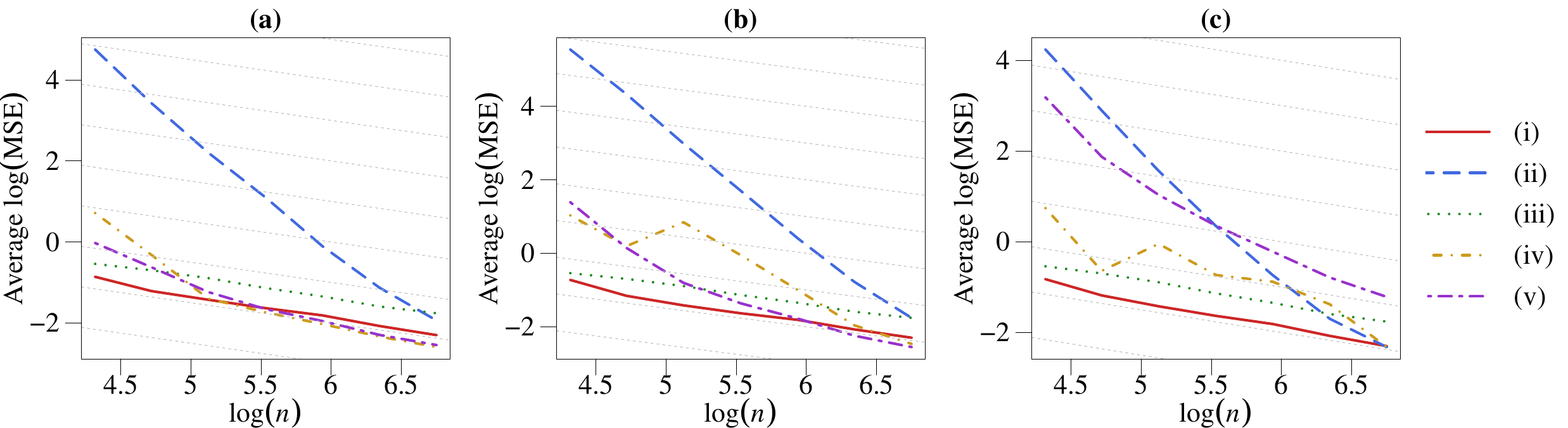}
    \vspace{-0.2cm}
  \captionsetup{font=small}
  \caption{
  The value of
  $\log \{ n^{1/2}\E(\cdot-\sigma^2)^2/\sigma^4 \}$
  averaged over $\sigma\in\{0.2,\ldots,2.2\}$ is computed for
  (i) proposed $\hat{\sigma}^2_*(3,1)$,
  (ii) $\hat{\sigma}^2(3,0)$,
  (iii) $\hat{\sigma}^2_*(3,2)$,
  (iv) $\hat{\sigma}^2_{\textsc{ls}}$,
  (v) $\hat{\sigma}^2_{\textsc{ll}}$.
  Plots (a)--(c) correspond to regression functions (T1)--(T3), respectively.
  The gray dotted lines are of slope $-1$
  corresponding to root-$n$ consistent estimators.
  }
  \label{fig:sim1}
  \end{center}
\end{figure}

The results are shown in Figure \ref{fig:sim1}.
The simulation results demonstrate the necessity of striking the balance between bias and variance:
variance-optimal differencing
lacks robustness
while
bias-optimal differencing lacks statistical efficiency.
In contrast, our proposed estimator $\hat{\sigma}^2_*(3,1)$,
which balances robustness and efficiency,
generally works well in all cases.
The least-squares estimator and the nonparametric estimator fail to tackle
complicated regression functions.

We repeat the experiments under two repetitional designs.
We compare our proposals $\hat{\sigma}^2_{\VB}(m_{\min},1)$
 and $\hat{\sigma}^2_{\Long}(m_{\min},1)$
to the sample variance estimator $\hat{\sigma}^2_{\textsc{sam}}$
and the sequencing estimator $\hat{\sigma}^2_{\textsc{seq}}$ proposed in \citet{DMTZ2015},
where $\hat{\sigma}^2_{\textsc{sam}}$ is defined differently as
the conventional sample variance.
We reuse the same setup as above
except that
$s(i)$'s are no longer equal to $1$.
We consider two cases:
\begin{itemize}
    \item[(L1)] let $s(1), \ldots, s(n) \sim \textsc{Unif}\{1,2,\ldots,5\}$
independently for $n=\lfloor 75\times 1.5^v \rfloor$;
    \item[(L2)] let $s(1), \ldots, s(n) \sim \textsc{Unif}\{v+1,\ldots,2(v+1)\}$ independently,
    and $n=200$,
\end{itemize}
for $v=0,\ldots,6$.
Figure \ref{fig:sim_longitudinal} shows the results.
The proposal $\hat{\sigma}^2_{\Long}(m_{\min},1)$
performs much better
than $\hat{\sigma}^2_{\textsc{seq}}$ in all cases.
Although $\hat{\sigma}^2_{\textsc{sam}}$ performs more comparably to
our proposals, it is not stable across models.
The MSE of $\hat{\sigma}^2_{\textsc{sam}}$ can be twice as large as that of $\hat{\sigma}^2_{\Long}(m_{\min},1)$.
The efficiency-boosted proposal $\hat{\sigma}^2_{\VB}(m_{\min},1)$ slightly improves the efficiency
at a price of
computational cost.

\begin{figure}
  \begin{center}
  \includegraphics[width=.9\textwidth]{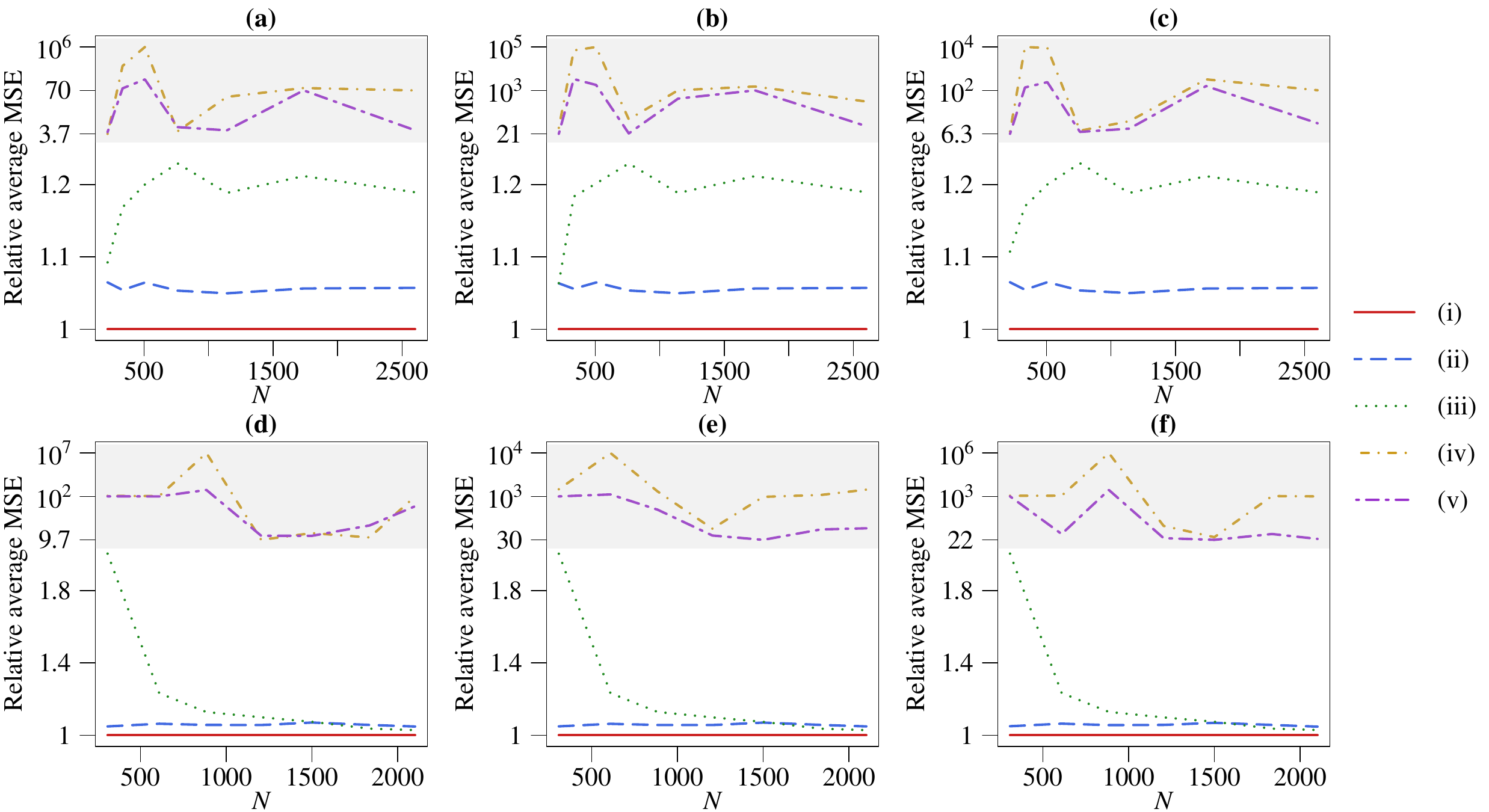}
    \vspace{-0.2cm}
  \captionsetup{font=small}
  \caption{
  The MSE
  $\E(\cdot-\sigma^2)^2$ averaged over $\sigma\in\{0.2,\ldots,2.2\}$ is computed for
  (i) proposed $\hat{\sigma}^2_\VB(m_{\min},1)$
  (ii) proposed $\hat{\sigma}^2_\Long(m_{\min},1)$,
  (iii) $\hat{\sigma}^2_{\textsc{sam}}$,
  (iv--v) $\hat{\sigma}^2_{\textsc{seq}}$ with bandwidths $\lfloor n^{1/2} \rfloor$ and $\lfloor n^{1/3} \rfloor$.
  The plots show the averaged MSE relative to estimator (i).
  Plots (a)--(c) represent regression functions (T1)--(T3) under (L1), respectively, while plots (d)--(f) correspond to the same functions under (L2).
  The minimum, median, and maximum relative errors among (iv) and (v) are labeled on the vertical axis of the gray regions in a nonlinear scale.
  }
  \label{fig:sim_longitudinal}
  \end{center}
\end{figure}

Lastly, we repeat the experiment with the same setup as above under two extremely imbalanced designs:
\begin{itemize}
    \item[(R1)]
    Let $q=\lceil 0.01n \rceil$ indices $i_1, \ldots, i_{q}$ are randomly selected from $\{1,\ldots,n\}$ without replacement.
    Then we set $s(i_1) = \ldots = s(i_q) = 100$ and
    generate $s(i)$ for $i\in\{1,\ldots,n\}\setminus\{i_1, \ldots, i_q\}$ independently
    from
    $\textsc{Unif}\{1,2,\ldots,5\}$
    for $n=\lfloor 75\times 1.5^v \rfloor$ for $v=0,\ldots,6$,
    as in (L1).
    \item[(R2)] The same setting is used as in (R1) except that the generating mechanism (L1) is replaced by (L2),
    i.e., $\textsc{Unif}\{v+1,\ldots,2(v+1)\}$ independently for $v=0,\ldots,6$ and $n=200$.
\end{itemize}

We compare our proposal $\hat{\sigma}^2_{\VB}(m'_{\min},1)$, which adjusts for the imbalanced design,
with $\hat{\sigma}^2_{\VB}(m_{\min},1)$, which does not adjust for the imbalanced design, and the sample variance estimator $\hat{\sigma}^2_{\textsc{sam}}$ in \citet{DMTZ2015}
under (R1) and (R2).
The threshold $s_{\max}$ is set to be $s_{0.95}$ for all scenarios.
Figure \ref{fig:long_supp_R} shows that
the proposed estimator $\hat{\sigma}^2_\VB(m'_{\min},1)$ is robust against the extremely imbalanced design points
while $\hat{\sigma}^2_\VB(m_{\min},1)$ has unstable performance.
Yet $\hat{\sigma}^2_\VB(m'_{\min},1)$ has better performance than
$\hat{\sigma}^2_{\textsc{sam}}$ in general.

\begin{figure}
  \begin{center}
  \includegraphics[width=0.9\textwidth]{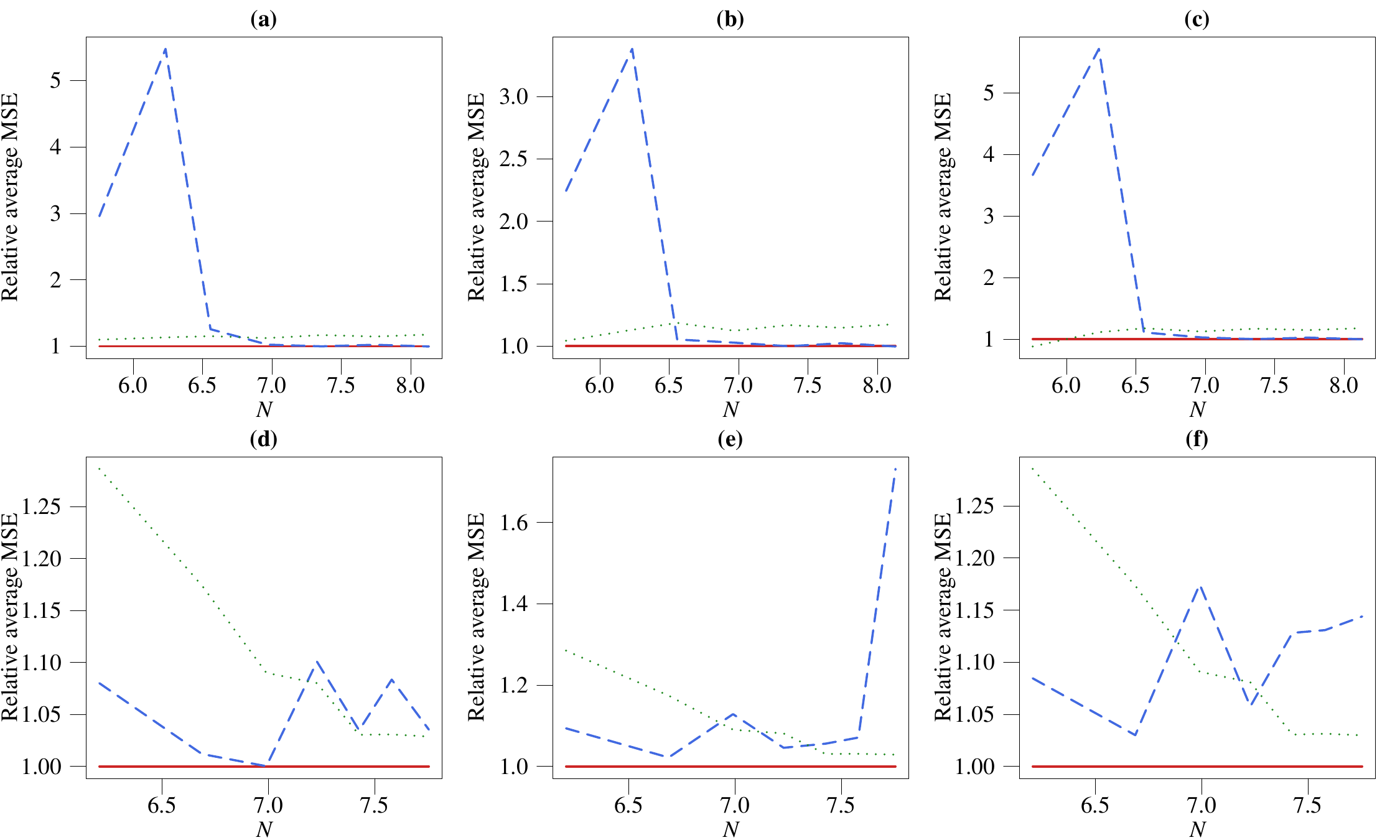}
  \captionsetup{font=small}
  \caption{The MSE
  $\E(\cdot-\sigma^2)^2$ averaged over $\sigma\in\{0.2,\ldots,2.2\}$ is computed for
  (i) proposed $\hat{\sigma}^2_\VB(m'_{\min},1)$ -- red solid line,
  (ii) the proposed $\hat{\sigma}^2_\VB(m_{\min},1)$ -- blue long dash line,
  (iii) $\hat{\sigma}^2_{\textsc{sam}}$ -- green dotted line,
  The plots show the value of the averaged MSE relative to estimator (i).
 Plots (a)--(c) represent regression functions (T1)--(T3) under (R1), respectively, while plots (d)--(f) correspond to the same functions under (R2).
  }
  \label{fig:long_supp_R}
  \end{center}
\end{figure}

\subsection{Application \label{sec:application}}
Consider a regression model with
potentially heteroscedastic noises:
$Y_i = g(X_i) + V^{1/2}(X_i) \xi_i$, $i=1,\ldots,n$,
where
$V(\cdot)$ is a variance function
and $(\xi_i)_{i=1}^n$ are independent noises of zero mean and unit variance.
\citet{Z2009} considered testing
\begin{equation}\label{eqt:H0H1_application}
   H_0:
   V(\cdot)\text{ is a constant}
   \quad\text{
   against}
   \quad
   H_1: V(\cdot)\text{ is non-constant}.
\end{equation}
So, $H_0$ assumes homoscedasticity.
The original test uses a residual-based variance estimator
\begin{eqnarray}\label{eqt:DTvarEst}
    \tilde{\sigma}^2_{\textsc{dt}}=\frac{1}{n}\sum_{i=1}^n\big\{Y_i-\hat{g}(X_i)\big\}^2,\quad \text{where}\quad\hat{g}(x)=\sum_{i=1}^n\frac{K\big\{(x-X_i)/h_n\big\}}{\sum_{j=1}^nK\big\{(x-X_j)/h_n\big\}}Y_i,
\end{eqnarray}
where $K$ is a kernel and $h_n$ is a bandwidth.
Let $\hat{u}_i=\big\{Y_i-\hat{g}(x)\big\}^2-\tilde{\sigma}^2_{\textsc{dt}}$, and
$K_{ij} = K\{ (x_i-x_j)/h_n \}$.
The test
statistic $T_n=nh_n^{1/2}W_n/\hat{s}$
is asymptotically normal under $H_0$, where
\begin{align*}
    W_n=
    \frac{1}{n(n-1)}
    \mathop{\sum\sum}_{1\leq i,j\leq n; i\neq j}\frac{1}{h_n}K_{ij} \hat{u}_i\hat{u}_j
    \qquad  \text{and} \qquad
    \hat{s}^2=\frac{2}{n(n-1)}\mathop{\sum\sum}_{1\leq i,j\leq n; i\neq j}
    \frac{1}{h_n} ( K_{ij}\hat{u}_i\hat{u}_j)^2.
\end{align*}
{However, the normal approximation is inaccurate in
finite samples and
leads to an under-size problem.
The author suggested that wild bootstrap can be used to mitigate the size inaccuracy.
Alternatively,
we propose to replace $\tilde{\sigma}^2_{\textsc{dt}}$ by our
proposal difference-based variance estimator.
We conjecture that it refines the test
due to its robustness.
We consider four regression functions (F1)--(F4) from \S~\ref{sec:non_equ_design_exp}.
(F1) $g(x)=\cos(4\pi x)\exp(x)$,
(F2) $g(x)=P^{(1,4)}_2(x)$,
(F3) $g(x)=\sin(2\pi x)+x\exp(1+x)-\mathbb{1}(x<0.3)+\mathbb{1}(x<0.6)$, and
(F4) $g(x)=\cos\{4\cos^{-1}(t)\}$,
where $P^{(\alpha,\beta)}_k(\cdot)$ is defined in Example \ref{eg:compare}.
We consider two variance functions:
(V1) $V(x)=1+\lambda\mathbb{1}\{t\in[0.25,0.5)\}/2-\lambda\mathbb{1}\{t\in[0.5,0.75)\}/2$ and
(V2) $V(x)=\{1+\lambda L(x,3)^2\}^2$
for some $\lambda\in\mathbb{R}$,
where $L(\cdot,p)$ is a Legendre polynomial of order $p$.
The noises $\xi_i\sim\Normal(0,1)$ independently.
$K(\cdot)$ is chosen to be Epanechnikov kernel, and
$h_n$ is selected by the \texttt{lpbwselect} function in R.
The parameter $\lambda$ in the $V(\cdot)$ controls
the degree of heteroscedasticity.
So, $H_0$ in (\ref{eqt:H0H1_application}) means
$\lambda=0$.}

Denote the original test as $T$.
Our suggested tests admit the same form as $T$ but replacing
$\tilde{\sigma}^2_{\textsc{dt}}$ with $\hat{\sigma}^2_{*}(m,r)$,
denoted as $T^{*}(m,r)$.
When there exists repeated measurements,
we replace $\tilde{\sigma}^2_{\textsc{dt}}$ with $\hat{\sigma}^2_{\VB}(m,r)$,
denoted as $T^{\VB}(m,r)$.
In addition, We also apply the test by replacing $\tilde{\sigma}^2_{\textsc{dt}}$ with $\hat{\sigma}^2_{\textsc{sam}}$ to serve as a competitor,
denoted as $T^{\textsc{sam}}$.
All tests are performed at size $5\%$
under $2048$ replications.
Table \ref{table:het_test_size} presents the type-I error rates under $H_0$,
whereas Table \ref{table:het_test_V1V2} displays
the rejection rates under $H_1$ at $\lambda=0.4,0.8$.

\begin{table}[t]
    \def~{\hphantom{0}}
    \caption{\label{table:het_test_size}Type-I error rates (\%) under regression functions F1--F4.
    The first column shows the design setups,
    where the upper section corresponds to standard design and the lower section corresponds to repeated design.
    The last column shows the mean absolute deviation (\%) of the type-I error rates
    from the nominal type-I error rates, i.e., $5\%$, among the four regression functions.
    The variance-optimal difference-based estimator is proposed by \citet{HKT1990}.
    }
    \begin{adjustbox}{width=\textwidth}
        \begin{tabular}{cccccccccc}
        \toprule
            &&& \multicolumn{4}{c}{regression function} & \\
            \cline{4-7}
            design & test & reference & F1 & F2 & F3 & F4 & deviation\\[2pt]
            \midrule
            $n=150$ & $T$ & \citet{Z2009} & $3.17$ & $2.44$ & $2.98$ & $2.83$ & $2.14$ \\
            $X_i\sim\textsc{Beta}(0.5,0.5)$ & $T^{*}(3,0)$ & variance-optimal & $4.44$ & $2.88$ & $3.66$ & $3.66$ & $1.34$  \\
            $s(i)=1$ & $T^{*}(3,2)$ & bias-optimal & $9.42$ & $9.72$ & $9.96$ & $11.08$ & $5.04$ \\
            & $T^{*}(3,1)$ & proposal & $4.49$ & $4.05$ & $3.91$ & $4.88$ & $0.67$ \\[2pt]
            \midrule
            $n=100$ & $T$ & \citet{Z2009} & $2.64$ & $3.47$ & $3.56$ & $2.78$ & $1.89$ \\
            $X_i=i/n$ & $T^{\textsc{sam}}$ & \citet{DMTZ2015} & $6.1$ & $8.06$ & $7.76$ & $7.08$ & $2.25$ \\
            $s(i)\sim\Unif\{1,2,3\}$ & $T^{\VB}(m_{\min},1)$ & proposal & $3.71$ & $4.88$ & $4.39$ & $3.76$ & $0.82$ \\
            \bottomrule
        \end{tabular}
    \end{adjustbox}
\end{table}

\begin{table}[t]
    \setlength{\tabcolsep}{2pt}
    \def~{\hphantom{0}}
    \caption{\label{table:het_test_V1V2}Power under regression functions F1--F4, variance functions V1--V2,
    and heteroscedasticity level $\lambda$.
    The first column shows the design setups,
    where the upper section corresponds to standard design and the lower section corresponds to repeated design.
    The last column shows the power averaged over
    the eight cases.
    The nominal type-I error is 5\%.
    }
    \begin{adjustbox}{width=\textwidth}
        \begin{tabular}{cccccccccccc}
        \toprule
            &&      & \multicolumn{8}{c}{regression function \& variance function} &\\
            \cline{4-11}
            design & test & $\lambda$ & F1-V1 & F2-V1 & F3-V1 & F4-V1 & F1-V2 & F2-V2 & F3-V2 & F4-V2 & power\\[2pt]
            \midrule
            $n=150$ & $T$ & $0.4$ & $4.20$ & $3.66$ & $4.83$ & $3.56$ & $16.11$ & $17.19$ & $19.78$ & $20.17$ & $11.19$\\
            $X_i\sim\textsc{Beta}(0.5,0.5)$ & $T^{*}(3,0)$ &  & $5.37$ & $4.64$ & $5.37$ & $4.10$ & $16.36$ & $18.12$ & $20.65$ & $21.44$ & $12.01$\\
            $s(i)=1$ & $T^{*}(3,2)$ & & $12.30$ & $12.60$ & $12.94$ & $13.04$ & $25.2$ & $28.71$ & $28.47$ & $31.10$ & $20.55$\\
            & $T^{*}(3,1)$ & & $5.66$ & $5.42$ & $5.81$ & $5.18$ & $18.36$ & $20.07$ & $22.36$ & $23.19$ & $13.26$\\[2pt]
            & $T$ & $0.8$ & $6.69$ & $8.01$ & $8.30$ & $6.93$ & $52.00$ & $50.00$ & $53.66$ & $50.15$ & $29.47$\\
            & $T^{*}(3,0)$ & & $7.86$ & $9.03$ & $9.67$ & $7.81$ & $53.08$ & $51.17$ & $54.93$ & $52.10$ & $30.71$\\
            & $T^{*}(3,2)$ & & $13.92$ & $20.46$ & $18.21$ & $18.36$ & $60.79$ & $60.16$ & $62.01$ & $60.99$ & $39.36$\\
            & $T^{*}(3,1)$ & & $8.74$ & $11.08$ & $10.4$ & $9.91$ & $54.93$ & $52.88$ & $57.47$ & $53.52$ & $32.37$\\[2pt]
            \midrule
            $n=100$ & $T$ & $0.4$   & $6.45$ & $8.40$ & $8.30$ & $7.42$ & $7.62$ & $7.08$ & $8.45$ & $7.57$ & $7.66$\\
            $X_i=i/n$ & $T^{\textsc{sam}}$ & & $11.72$ & $14.60$ & $13.96$ & $12.40$ & $13.13$ & $15.33$ & $14.60$ & $13.23$ & $13.62$\\
            $s(i)\sim\Unif\{1,2,3\}$ & $T^{\VB}(m_{\min},1)$ & & $7.81$ & $10.25$ & $10.64$ & $8.94$ & $9.33$ & $9.47$ & $10.30$ & $9.23$ & $9.50$\\
            [2pt]
            & $T$ & $0.8$ & $25.59$ & $33.15$ & $28.76$ & $31.84$ & $20.75$ & $20.75$ & $23.24$ & $21.48$ & $25.7$\\
            & $T^{\textsc{sam}}$ & & $34.03$ & $42.04$ & $36.38$ & $41.75$ & $28.96$ & $29.98$ & $30.32$ & $31.3$ & $34.34$\\
            & $T^{\VB}(m_{\min},1)$ & & $28.56$ & $36.47$ & $32.13$ & $35.3$ & $23.73$ & $23.39$ & $25.78$ & $25.59$ & $28.87$\\
            \bottomrule
    \end{tabular}
    \end{adjustbox}
\end{table}

For the standard design,
while $T^{*}(3,0)$ is size-accurate but least powerful,
$T^{ *}(3,2)$ is the most powerful yet size-inaccurate.
Our proposed $T^{ *}(3,1)$ strikes a good balance
and is the most powerful and size-accurate test.
When there are repeated measurements,
the difference-based test $T^{\VB}(m_{\min},1)$ is also the most powerful and size-accurate test.
The test $T^{\textsc{sam}}$, despite being most powerful,
is over-sized.
The application of a difference-based estimator
mitigates the inaccurate normal approximation and
reduces computation time substantially as bootstrap is no longer needed; see the last column of Table \ref{table:het_test_size}.
We note that applying bootstrap can improve all the tests in terms of size and power at the expense of computation costs.

\subsubsection{Real-data application: global surface temperature}
The National Center for Environmental Information
collects the surface temperature of the Earth.
We aim at testing whether the surface temperature data is heteroscedastic.
The monthly average land-and-ocean-combined global surface temperature data from 1850 to 2023 are obtained
from
\href{https://www.ncei.noaa.gov/data/noaa-global-surface-temperature/v5.1/access/timeseries/aravg.mon.land_ocean.90S.90N.v5.1.0.202312.asc}{their database}.
{The dataset contains no repeated measurements.
We applied $T$ and $T^{*}(3,1)$ to various subsamples
of the dataset that correspond to the past $150$, $100$, and $50$ years of data,
where our proposal $\hat{\sigma}^2_*(3,1)$ is used for variance estimation in $T^{*}(3,1)$.}
A summary of results is shown in Table \ref{table:co2_global_daily}.

\begin{table}[t]
    \def~{\hphantom{0}}
    \caption{\label{table:co2_global_daily}The $p$-values of $T$ and $T^{*}(3,1)$ based on the whole, past 150-year, 100-year and 50-year data sets.
    }
    \centering
    \begin{adjustbox}{width=0.6\textwidth}
        \begin{tabular}{ccccc}
            \toprule
            Data & 1850--2023 & 1873--2023 & 1923--2023 & 1973--2023\\
            \midrule
            $T$ & $<10^{-5}$ & $0.0012$ & $0.078$ & $0.042$ \\
            $T^{ *}(3,1)$ & $<10^{-200}$ & $<10^{-200}$ & $<10^{-100}$ & $<10^{-20}$ \\
            \bottomrule
        \end{tabular}
    \end{adjustbox}
\end{table}

The proposed test $T^{ *}(3,1)$ consistently indicates
compelling evidence of heteroscedasticity, whereas
the original test $T$ only
reveal less convincing conclusions and may fail to reject the homoscedasticity hypothesis at 5\% level
when the dataset is subsampled.
In particular, the conclusion based on $T$ is ambiguous
based on the datasets during 1923--2023 and 1973--2023.

\begin{figure}[t]
  \begin{center}
  \includegraphics[width=\textwidth]{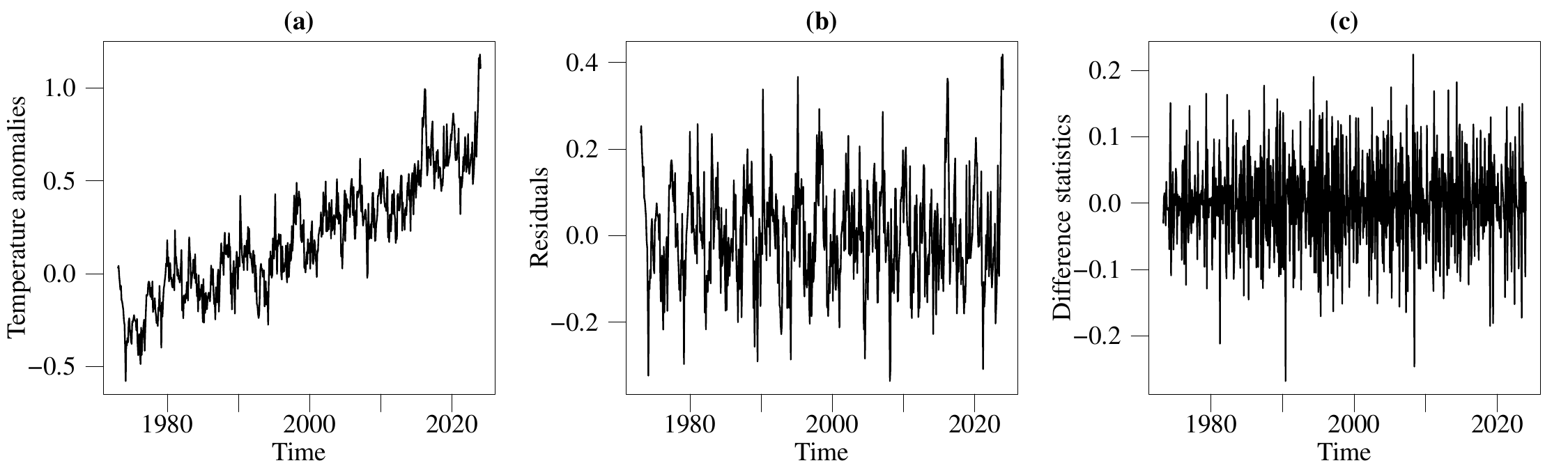}
  \captionsetup{font=small}
  \caption{Plot (a) shows the temperature anomalies relative to a 1971 to 2000 monthly climatology.
  Plot (b) shows the residuals $Y_i-\hat{g}(t)$ against time.
  Plot (c) shows the difference statistics $D_i$ against time.}
  \label{fig_temperature_residual}
  \end{center}
\end{figure}

Figure \ref{fig_temperature_residual} plots the residuals $\hat{\epsilon}_i = Y_i - \hat{g}(X_i)$
and the difference statistics $D_i$ based on the past 50 years of data.
The residuals $\hat{\epsilon}_i$ exhibit a more pronounced seasonal pattern,
indicating that they are less likely to be white noises,
whereas the difference statistics $D_i$ do not reveal this problem.
Furthermore, we observe that the variance estimates obtained
using $\tilde{\sigma}^2_{\textsc{dt}}$ based on datasets from
the past 150, 100, and 50 years are consistently larger than those obtained using $\hat{\sigma}^2_{*}(3,1)$.
We conjecture that this difference arises because $\tilde{\sigma}^2_{\textsc{dt}}$ is
more severely affected by nonstationarity
in the input dataset,
as a significant portion of the underlying trend $g(X)$ remains
in the residuals even after de-trending.
This phenomenon contributes to the stronger and more consistent
findings derived from $T^{ *}(3,1)$.
It highlights the contribution of the proposed variance estimator.

\subsubsection{Real-data application: precipitation in Alexandria}
The daily precipitation data collected from $5$ weather stations in Alexandria, Virginia, is downloaded from
the \href{https://www.ncei.noaa.gov/access/past-weather/}{Past Weather Tool} maintained by the National Center for Environmental Information.
The data contains $2$--$5$ repeated measurements in each day,
where the average number of repetitions is $4.71$.
We aggregate the daily precipitation of each station to compute the monthly total precipitation of each station.
We aim to test whether the monthly total rainfall data of Alexandria, Virginia, exhibits heteroscedasticity.
\begin{figure}[t]
  \begin{center}
  \includegraphics[width=0.7\textwidth]{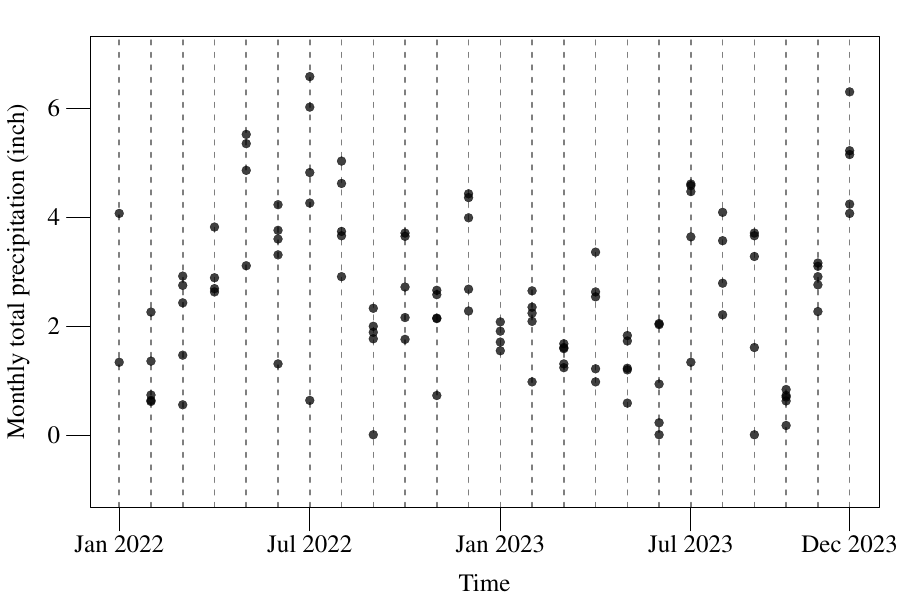}
  \captionsetup{font=small}
  \caption{Plot of monthly total precipitation against time.}
  \label{fig_prcp}
  \end{center}
\end{figure}

The $p$-value of $T$ is $0.101$ and the $p$-value of $T^{\VB}(3,1)$ is $0.00165$.
Figure \ref{fig_prcp} plots the monthly total precipitation data against time.
The data appears to exhibit some heteroscedasticity.
In particular,
the data has higher variance in May and June than other months.
However,
the original test $T$ fails to reject the null hypothesis even at $10\%$ level.
On the other hand,
$T_{\VB}(3,1)$ unveils convincing evidence of heteroscedasticity noise.
It demonstrates the value of the proposed variance estimator.

\begin{appendix}
\section{Proofs}
\subsection{Proof of Theorem {\ifnum\isXr=1{\ref{thm_diff_bias}}\else{4.1}\fi}}

(Local differencing)
By Proposition {\ifnum\isXr=1{\ref{prop:poly_cancelling}}\else{1}\fi},
we have $D^g_i=0$ for any $k$th degree polynomial trends $g(X)=\sum_{\ell=0}^ka_\ell X^\ell$,
where $0\leq k\leq r$,
hence $\Bias\big\{\hat{\sigma}^2_{*}(m,r)\big\}=0$.
Otherwise, $D_i^g=O_p(n^{-r-1})$.
By \citet{SGW1993},
\begin{align*}
    \Bias\big\{\hat{\sigma}^2_{*}(m,r)\big\}&=\frac{1}{n-m}\sum^{n-m}_{i=1}(D^g_i)^2\\
    &=\frac{1}{n-m}\sum^{n-m}_{i=1}\left\{\frac{g^{(r+1)}(X_i)}{(r+1)!}\sum_{j=0}^md_{i,j}(X_{i+j}-X_{i})^{r+1}\right\}^2+o_p(n^{-2r-2})\\
    &=O_p(n^{-2r-2}).
\end{align*}
The results remain the same when $X_i=i/n$ $(i=1,\ldots,n)$
since it satisfies Assumption {\ifnum\isXr=1{\ref{assumption:basic}}\else{2}\fi}.

(Global differencing)
By Proposition {\ifnum\isXr=1{\ref{prop:poly_cancelling}}\else{1}\fi}, we have $D^g_i=0$ in the case where $g(\cdot)$ is a constant function.
Hence $\Bias\left\{\hat{\sigma}^2(m,r)\right\}=0$.
Otherwise, $D_i^g=O_p(n^{-1})$ and
\begin{eqnarray*}
    \Bias\big\{\hat{\sigma}^2(m,r)\big\}&=&\frac{1}{n-m}\sum^{n-m}_{i=1}\left\{g^{(1)}(X_i)\sum_{j=0}^md_{j}(X_{i+j}-X_{i})\right\}^2+o_p(n^{-2})\\
    &=&O_p(n^{-2}).
\end{eqnarray*}
However, since Assumptions
{\ifnum\isXr=1{\ref{ass:typeDiffSeq}\ref{assumption:design_adapt_con} and \ref{ass:typeDiffSeq}\ref{assumption:regular_con}}\else{1(a) and 1(b)}\fi}
are equivalent when $X_i=i/n$ $(i=1,\ldots,n)$.
The results follow immediately from those results of local differencing.

\subsection{Proof of Theorem {\ifnum\isXr=1{\ref{thm_diff_var}}\else{4.2}\fi}}
Denote $1_n$ as a column vector of $1$ of length $n$,
$a_{i,j}$ as the $(i,j)$th element of matrix $A$
and $g=(g(X_1),\ldots,g(X_n))^\T$.
By \citet{SGW1993},
the variance of a difference-based estimator is
\begin{align}\label{eqt:var_SGW1993}
    \Var\big\{\hat{\sigma}_*^2(m,r)\big\}
    &=\frac{1}{\tr(A)^2}\Bigg(2\sigma^4\tr(A^2)+4\sigma^2g^\T A^2g \nonumber\\
    &\qquad+2\gamma_3\sigma^3\bigg[\tr\Big\{A\Diag(Ag{1}^\T_{n})\Big\}+g^\T A\Diag(A)1_n\bigg]  \\
    &\qquad +(\lambda_4-3)\sigma^4\tr\Big\{A\Diag(A)\Big\}\Bigg).\nonumber
\end{align}

Under either Assumption
{\ifnum\isXr=1{\ref{ass:typeDiffSeq}\ref{assumption:design_adapt_con} or \ref{ass:typeDiffSeq}\ref{assumption:regular_con}}\else{1(a) or 1(b)}\fi},
the second constraint in {\ifnum\isXr=1{(\ref{diff_poly_cond0})}\else{(7)}\fi} ensures that
\begin{align}\label{eqt:trA}
    \tr(A)&= \sum_{i=1}^{n-m}\sum_{j=0}^m d_{i,j}^2=n-m.
\end{align}

By Assumption {\ifnum\isXr=1{\ref{assumption:Dgi_mean}}\else{4}\fi}, we have
\begin{align*}
    \sum_{i=1}^{n-m}D^g_i/(n-m)&=O_p(n^{-\kappa}),\\
    \sum_{i=1}^{n-m}(D^g_i)^2/(n-m)&=O_p(n^{-2\kappa}),\\
    \sum_{\ell=1}^m\sum_{i=1}^{n-m-\ell}D^g_iD^g_{i+\ell}/(n-m)&=O_p(n^{-2\kappa}),
\end{align*}
which imply that
\begin{eqnarray*}
    \tr(A^2)&=&\sum^{n-m}_{i=1}\left(\sum^m_{j=0}d^2_{i,j}\right)^2+2\sum^m_{\ell=1}\sum^{n-m-\ell}_{i=1}\left(\sum^{m-\ell}_{j=0}d_{i,j+\ell}d_{i+\ell,j}\right)^2
    =O(n); \\
    g^\T A^2g&=&\sum^{n-m}_{i=1}\left(\sum^m_{j=0}d^2_{i,j}\right)(D^g_i)^2+2\sum^m_{\ell=1}\sum^{n-m-\ell}_{i=1}\left(\sum^{m-\ell}_{j=0}d_{i,j+\ell}d_{i+\ell,j}\right)D^g_iD^g_{i+\ell} \\
    &=&O_p(n^{-2\kappa+1}); \\
    \tr\left\{A\Diag(Ag{1}^\T_{n})\right\}&=&\sum^{n-m}_{i=1}\sum^m_{j=0}d^2_{i,j}\sum^{\min(i+j,n-m)}_{h=\max(1,i+j-m)}d_{h,i+j-h}D^g_h
    =O_p(n^{-\kappa+1}); \\
    g^\T A\Diag(A)1_n&=&\sum^{n-m}_{i=1}\left(\sum^m_{j=0}d_{i,j}a_{i+j,i+j}\right)D^g_i
    = O_p(n^{-\kappa+1}); \\
    \tr\left\{A\Diag(A)\right\}&=&\sum^{n-m}_{i=1}\sum^m_{j=0}d^2_{i,j}a_{i+j,i+j}
    =O(n),
\end{eqnarray*}
where we have used the fact that $a_{i,j}=O(m)$ and $m<\infty$
for all $i,j$.
Therefore, the variance of $\hat{\sigma}_*^2(m,r)$
can be written as
\begin{equation*}
    \Var\left\{\hat{\sigma}_*^2(m,r)\right\}=\frac{1}{(n-m)^2}\left[2\sigma^4\tr(A^2)+(\lambda_4-3)\sigma^4\tr\left\{A\Diag(A)\right\}\right]+O_p(n^{-\kappa-1}).
\end{equation*}

\subsection{Proof of Theorem {\ifnum\isXr=1{\ref{thm:long_diff_bias}}\else{4.3}\fi}}
We first drive the lower bound of $m$ such that
$D_i$ is repetitional-bias-corrected
for all $i$ and each $r$.
We call $D_{\ell(i)}$ $(i\in\ell_1\cup\ell_2)$ the $i$th leading difference statistics
because
\[
    D_{\ell(i)}=d_{i,0}Y_{1,i}+\sum_{j=i}^md_{i,j}Y_{i+j},
\]
where
the first component in $D_{\ell(i)}$ must be the
first observation $Y_{1,i}$ among $\{Y_{k,i}\}_{k=1}^{s(i)}$.
We call the remaining difference statistics $D_k$, where $k\neq \ell(i)$ for any $i\in\ell_1\cup\ell_2$, non-leading difference statistics.
Because their first component must not be the first observation from any distinct design points.

If $r=0$, then it can be easily seen that $m\geq 1$ is sufficient.
So, in this case, $m_{\min}=1$.
Now, we consider $r>0$.
Recall in \S~{\ifnum\isXr=1{\ref{sec:repeated}}\else{3}\fi},
intra-group differencing ensures
that $D^{ g}_i=0$
for $i\in\mathcal{I}_m$
because it can be easily seen that
\begin{align*}
    D^{ g}_i=\sum_{j=0}^md_{i,j}g(X_{i+j})
    =\sum_{j=0}^{m(i)}\sum_{t\in\Lambda_j(i)}d^\circ_{i,t}g(X^\star_{i+j})
    =\sum_{j=0}^{m(i)}g(X^\star_{i+j})\sum_{t\in\Lambda_j(i)}w_{i,t}/c_i=0
\end{align*}
Hence, it suffices to consider
inter-group differencing (i.e., $i\notin\mathcal{I}_m$).
The difference sequence $d_{(i)} = (d_{i,0}, \ldots, d_{i,m})^{\T}$
can only provide
$m(i)+1$ degrees of freedom for bias correction
even though it has $m+1$ components;
see Example {\ifnum\isXr=1{\ref{eg:diff_long}}\else{2}\fi} for an illustration.
The $m(i)+1$ degrees of freedom are used to ensure that
the $r+2$ constraints stated in
{\ifnum\isXr=1{(\ref{diff_poly_cond0})}\else{(7)}\fi} and {\ifnum\isXr=1{(\ref{diff_poly_random_cond})}\else{(11)}\fi}
are satisfied.
Thus,
provided $\sup_{i=1,\ldots,N-m}m(i)<\infty$,
for the results (a) and (b) in Proposition {\ifnum\isXr=1{\ref{prop:poly_cancelling}}\else{1}\fi} to hold for $i\notin\mathcal{I}_m$ under repetitional design, i.e.,
\begin{align*}
    D^{ g}_i=\sum_{j=0}^md_{i,j}g(X_{i+j})
    =\sum_{j=0}^{m(i)}g(X^\star_{i+j})\sum_{t'\in\Lambda_j(i)}w_{i,t}d_{i,t(i,t')}/c_i\\
    =\sum_{j=0}^{m(i)}d_{i,j}g(X^\star_{i+j})/c_i
    =O(n^{-r-1}),
\end{align*}
the order $m$ has to be sufficiently large such that
$m(i)\geq r+1$ for $i\notin\mathcal{I}_m$.

We move on to show that,
for the $i$th leading difference statistics $D_{\ell(i)}$ $(i\in\ell_1\cup\ell_2)$,
the value of $m$ must satisfy that
\begin{align}\label{eq:thm2_proof:1}
    m \geq M(i)\equiv \left\{
		\begin{array}{ll}
		M(1,i) \equiv S\left\{i,i+r-1\right\}\mathbb{1}(r>0)+1 & \quad \text{if } \min_{i'=i,\ldots,i+r-1}s(i')>1; \\
		M(2,i)\equiv S\left\{i,i+r\right\}\mathbb{1}(r>0)+\mathbb{1}(r=0) & \quad \text{if } \min_{i'=i,\ldots,i+r-1}s(i')=1 ,
		\end{array}
	\right.
\end{align}
in order to ensure that $D_{\ell(i)}$ is repetitional-bias-corrected,
where these two case in (\ref{eq:thm2_proof:1})
correspond to $\ell(i)\in\ell_1$ and $\ell(i)\in\ell_2$, respectively.
The expression $S(i,i+a)$ counts the number of observations
such that $(X_{\ell(i)},\ldots,X_{\ell(i)+S(i,i+a)-1})=(X_{1,i},\ldots,X_{s(i+a),i+a})$ consists of $a+1$ unique design points.
We then observe the following:
\begin{itemize}
    \item[(i)]
If $\min_{i'=i,\ldots,i+r-1}s(i')>1$,
the condition $m\geq M(1,i)$ implies
$m\left\{\ell(i)\right\}\geq r+1$
unless
$\ell(i)\in\mathcal{I}_m$.
It is because $m=M(1,i)$ implies that
$$(X_{\ell(i)},\ldots,X_{\ell(i)+M(1,i)})
=(X_{1,i},\ldots,X_{s(i+r-1),i+r-1},X_{\ell(i)+M(1,i)-1},X_{\ell(i)+M(1,i)}).$$
So, we have $m\left\{\ell(i)\right\}= r+1$
unless
$X_{\ell(i)+M(1,i)-1}=X_{\ell(i)+M(1,i)}$,
which means $\ell(i)\in\mathcal{I}_m$.
See the cartoon in Figure \ref{fig:thm2:1} for a graphical illustration.
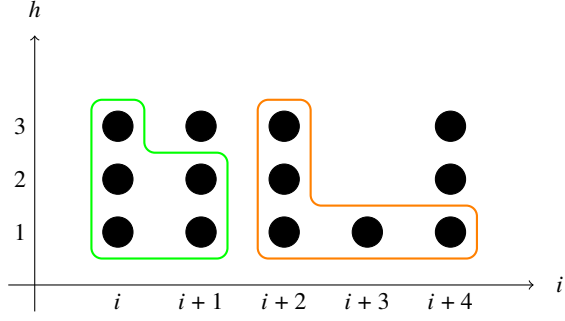
\begin{figure}[t]
        \centering
        \begin{tikzpicture}
            \tikzset{
                vertex/.style={circle,draw, inner sep=0pt, outer sep=0pt
                , minimum width=0.4cm
                ,fill=black
                },
            }
            \draw[->] (-\tikzTextMinDist,0)--(6*\tikzTextDist,0) node[below=1mm]{};
            \node at (6*\tikzTextDist+\tikzTextMinDist,0) [] {$i$};
            \draw[->] (0,-\tikzTextMinDist)--(0,3*\tikzTextDist) node[left]{};
            \node at (0,3*\tikzTextDist+\tikzTextMinDist) [] {$h$};
            \foreach \i in {1}{
                \pgfmathsetmacro\index{\i-4}
                \draw (\i*\tikzTextDist,0) node[below] {$i$};
            }
            \foreach \i in {2,...,5}{
                \pgfmathsetmacro\index{\i-1}
                \draw (\i*\tikzTextDist,0) node[below] {$i+\pgfmathprintnumber{\index}$};
            }
            \foreach \h in {1,...,3}{
                \draw (0,\h*\tikzNodeDist) node[left] {$\h$};
            }
            \foreach \i in {1,...,3}{
                \foreach \h in {1,2,3}{
                    \node at (\i*\tikzTextDist,\h*\tikzNodeDist) [vertex] {};
                }
            }
            \foreach \i in {4}{
                \foreach \h in {1}{
                    \node at (\i*\tikzTextDist,\h*\tikzNodeDist) [vertex] {};
                }
            }
            \foreach \i in {5}{
                \foreach \h in {1,2,3}{
                    \node at (\i*\tikzTextDist,\h*\tikzNodeDist) [vertex] {};
                }
            }
            \draw[rounded corners,green,thick]
            (1*\tikzTextDist-\tikzTextMinDist,4*\tikzNodeDist-\tikzNodeMinDist)
            --(1*\tikzTextDist-\tikzTextMinDist,1*\tikzNodeDist-\tikzNodeMinDist)
            --(2*\tikzTextDist+\tikzTextMinDist,1*\tikzNodeDist-\tikzNodeMinDist)
            --(2*\tikzTextDist+\tikzTextMinDist,3*\tikzNodeDist-\tikzNodeMinDist)
            --(1*\tikzTextDist+\tikzTextMinDist,3*\tikzNodeDist-\tikzNodeMinDist)
            --(1*\tikzTextDist+\tikzTextMinDist,4*\tikzNodeDist-\tikzNodeMinDist)
            --cycle  {};
            \draw[rounded corners,orange,thick]
            (3*\tikzTextDist-\tikzTextMinDist,4*\tikzNodeDist-\tikzNodeMinDist)
            --(3*\tikzTextDist-\tikzTextMinDist,1*\tikzNodeDist-\tikzNodeMinDist)
            --(5*\tikzTextDist+\tikzTextMinDist,1*\tikzNodeDist-\tikzNodeMinDist)
            --(5*\tikzTextDist+\tikzTextMinDist,2*\tikzNodeDist-\tikzNodeMinDist)
            --(3*\tikzTextDist+\tikzTextMinDist,2*\tikzNodeDist-\tikzNodeMinDist)
            --(3*\tikzTextDist+\tikzTextMinDist,4*\tikzNodeDist-\tikzNodeMinDist)
            --cycle  {};
        \end{tikzpicture}
        \caption{Suppose $s(i)=3$ for all $i$ except $s(i+3)=1$.
        Set $(m,r)=(M(1,i),1)=(4,1)$.
        The difference statistics $D_{\ell(i)}$ and $D_{\ell(i+2)}$, which are the green and orange group respectively, represent the case $\ell(i)\in\mathcal{I}_m$ and the case $m\{\ell(i+2)\}=r+1$.
        }

        \label{fig:thm2:1}
\end{figure}

\item[(ii)]
If $\min_{i'=i,\ldots,i+r-1}s(i')=1$,
then $\ell(i)\notin\mathcal{I}_m$.
So, we require $m\left\{\ell(i)\right\}\geq r+1$
to ensure that (a) and (b) in Proposition {\ifnum\isXr=1{\ref{prop:poly_cancelling}}\else{1}\fi} hold.
It is trivial that $m\geq M(2,i)$ implies $m(i)\geq r+1$.
See Remark \ref{remark:thm2:1} for the details why $M(1,i)$ is insufficient.
\end{itemize}

Next, we proceed to consider non-leading
difference statistics.
If $m$ satisfies (\ref{eq:thm2_proof:1}) for $i=\ell(j),\ell(j+1)$ and some $j$,
then $m(j')\geq r+1$ must also hold for $j'\in\{\ell(j)+1,\ldots,\ell(j+1)-1\}$,
unless $j'\in\mathcal{I}_m$.
So, $D_{\ell(j)+1},\ldots,D_{\ell(j+1)-1}$ are repetitional-bias-corrected.
Hence, it is sufficient to ensure that $m\geq m_{\min}$.
See the cartoon in Figure \ref{fig:thm2:2} for a graphical illustration.
\begin{figure}[t]
        \centering
        \begin{tikzpicture}
            \tikzset{
                vertex/.style={circle,draw, inner sep=0pt, outer sep=0pt
                , minimum width=0.4cm
                ,fill=black
                },
            }
            \draw[->] (-\tikzTextMinDist,0)--(6*\tikzTextDist,0) node[below=1mm]{};
            \node at (6*\tikzTextDist+\tikzTextMinDist,0) [] {$i$};
            \draw[->] (0,-\tikzTextMinDist)--(0,3*\tikzTextDist) node[left]{};
            \node at (0,3*\tikzTextDist+\tikzTextMinDist) [] {$h$};
            \foreach \i in {1}{
                \pgfmathsetmacro\index{\i-4}
                \draw (\i*\tikzTextDist,0) node[below] {$i$};
            }
            \foreach \i in {2,...,5}{
                \pgfmathsetmacro\index{\i-1}
                \draw (\i*\tikzTextDist,0) node[below] {$i+\pgfmathprintnumber{\index}$};
            }
            \foreach \h in {1,...,3}{
                \draw (0,\h*\tikzNodeDist) node[left] {$\h$};
            }
            \foreach \i in {1,...,5}{
                \foreach \h in {1,2,3}{
                    \node at (\i*\tikzTextDist,\h*\tikzNodeDist) [vertex] {};
                }
            }
            \draw[rounded corners,purple,thick]
            (1*\tikzTextDist-\tikzTextMinDist,4*\tikzNodeDist-\tikzNodeMinDist)
            --(1*\tikzTextDist-\tikzTextMinDist,2*\tikzNodeDist-\tikzNodeMinDist)
            --(2*\tikzTextDist-\tikzTextMinDist,2*\tikzNodeDist-\tikzNodeMinDist)
            --(2*\tikzTextDist-\tikzTextMinDist,1*\tikzNodeDist-\tikzNodeMinDist)
            --(2*\tikzTextDist+\tikzTextMinDist,1*\tikzNodeDist-\tikzNodeMinDist)
            --(2*\tikzTextDist+\tikzTextMinDist,4*\tikzNodeDist-\tikzNodeMinDist)
            --cycle  {};
            \draw[rounded corners,blue,thick]
            (3*\tikzTextDist-\tikzTextMinDist,4*\tikzNodeDist-\tikzNodeMinDist)
            --(3*\tikzTextDist-\tikzTextMinDist,3*\tikzNodeDist-\tikzNodeMinDist)
            --(4*\tikzTextDist-\tikzTextMinDist,3*\tikzNodeDist-\tikzNodeMinDist)
            --(4*\tikzTextDist-\tikzTextMinDist,1*\tikzNodeDist-\tikzNodeMinDist)
            --(5*\tikzTextDist+\tikzTextMinDist,1*\tikzNodeDist-\tikzNodeMinDist)
            --(5*\tikzTextDist+\tikzTextMinDist,2*\tikzNodeDist-\tikzNodeMinDist)
            --(4*\tikzTextDist+\tikzTextMinDist,2*\tikzNodeDist-\tikzNodeMinDist)
            --(4*\tikzTextDist+\tikzTextMinDist,4*\tikzNodeDist-\tikzNodeMinDist)
            --cycle  {};
        \end{tikzpicture}
        \caption{Suppose $s(i)=3$ for all $i$.
        Set $(m,r)=(m_{\min},1)=(4,1)$.
        The difference statistics $D_{\ell(i)+1}$ and $D_{\ell(i+3)-1}$, which are the purple and blue group respectively, represent the two possible difference statistics between two leading difference statistics,
        which satisfy $\left\{\ell(i)+1\right\}\in\mathcal{I}_m$ and $m\{\ell(i+3)-1\}\geq r+1$ respectively.
        It shows that if $m$ is large enough so that all the leading difference statistics are repetitional-bias-corrected,
        then all difference statistics are also repetitional-bias-corrected,
        i.e. satisfy either $i\in\mathcal{I}_m$ or $m(i)\geq r+1$.
        }

        \label{fig:thm2:2}
\end{figure}
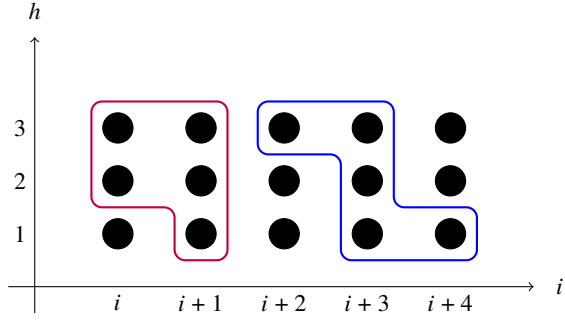
Recall that
$$
\mathcal{S}_1=\left\{
M(1,i)
:i\in \ell_1\right\},
\quad
\mathcal{S}_2=\left\{
M(2,i)
:i\in \ell_2\right\}.
$$
So, if $m\geq m_{\min} =\max(\mathcal{S}_1\cup\mathcal{S}_2)$, then $D_{\ell(i)}$ is repetitional-bias-corrected for all
$i\in\ell_1\cup\ell_2= \{1, \ldots, n-r+1\}$,
where $\{D_{\ell(1)},\ldots,D_{\ell(n-r+1)}\}$ is the set of all possible leading difference statistics.
It implies $D_i$ is repetitional-bias-corrected for all $i\in\{1,\ldots,N-m\}$.
See also Remark \ref{remark:thm2:2} for more details about the special cases $D_{\ell(n-r+1)}$ and $D_{\ell(n-r)}$, which are possibly ill-defined.

\begin{remark}\label{remark:thm2:1}
For (ii) above, it can be shown that $m=M(1,i)$ fails to guarantee sufficient degrees of freedom when $s(i')=1$ for some $i'\in\{i,\ldots,i-r+1\}$.
It is possible that
$m=M(1,i)$ leads to $m(i)<r+1$
while $i\notin\mathcal{I}_m$.
Because
$m=M(1,i)$ means $D_{\ell(i)}$ consists of $S(i,i+r-1)+2$ observations.
It can be seen that $m\{\ell(i)\}=r+\mathbb{1}(X_{\ell(i)+m-1}\neq X_{\ell(i)+m})$.
We have $m\{\ell(i)\}=r<r+1$ if
$X_{\ell(i)+m-1}=X_{\ell(i)+m}$.

\end{remark}

\begin{remark}\label{remark:thm2:2}
    Suppose $m=M_{1,\ell(n-r+1)}$ and $\min_{i=n-r+1,\ldots,n}s(i)>1$.
    It can be easily seen that $D_{\ell(n-r+1)}$ and $D_{\ell(n-r)}$ do not exist if $\min_{i=n-r+1,\ldots,n}s(i)>1$ and if $\min_{i=n-r+1,\ldots,n}s(i)=1$, respectively.
    The difference statistic
    \begin{align*}
        D_{\ell(n-r+1)}
        &=\sum_{j=0}^md_{\ell(n-r+1),j}Y_{\ell(n-r+1)+j}\\
        &=\sum_{j=0}^{m-1}d_{\ell(n-r+1),j}Y_{\ell(n-r+1)+j}
        +\sum_{j'=1}^2d_{\ell(n-r+1),m-2+j'}Y_{N+j'}
    \end{align*}
    is ill-defined
    since the observations $Y_{N+1}$ and $Y_{N+2}$ do not exist.
    However, the condition $m\geq m_{\min}\geq M_{1,\ell(n-r+1)}$ is needed
    so that the arguments above can be applied to $D_{i}$ for all $i$, i.e.,
    the value of $m$ satisfying (\ref{eq:thm2_proof:1}) for $i=\ell(n-r),\ell(n-r+1)$ implies $D_j$ is repetitional-bias-corrected for $j\in\{\ell(n-r),\ldots,N-m\}$.
    The same argument applies to $D_{\ell(n-r)}$ if $\min_{i=n-r+1,\ldots,n}s(i)=1$.
\end{remark}

We proceed to prove the convergence rate of $\Bias\{\hat{\sigma}^2_{\Long}(m,r)\}$,
given $m\geq m_{\min}$ and $r\geq 0$.
By \citet{SGW1993},
we have
\begin{align*}
    \Bias\left\{\hat{\sigma}^2_{\Long}(m,r)\right\}
    &=\frac{1}{N-m}\sum_{i=1}^{N-m}\left(D^{ g}_i\right)^2\\
    &= \mathcal{D}O_p(n^{-2r-2})/(N-m),
\end{align*}
where $\mathcal{D}$ be the number of $D_i$ that consists of
non-zero $D_{i}^{ g}$.
To derive the bias expression,
it suffices to derive the order of $\mathcal{D}$.
However, it is difficult to derive a general formula for $\mathcal{D}$.
Instead, we provide the following arguments to show when $\mathcal{D}=O(n+S_n m)$ holds.
Three cases are considered.

\begin{itemize}
    \item[(i)]
In this case, we suppose $s(1)=\cdots=s(n)=k\geq 2$.
Recall that the data is arranged in the following form:
\begin{equation*}
    \left(
    \begin{array}{ccccccccccc}
        Y_{1,1} & \cdots & Y_{k,1} & Y_{1,2} &\cdots & Y_{k,2} & \cdots & Y_{1,n} &\cdots & Y_{k,n}\\
        X_{1,1} & \cdots & X_{k,1} & X_{1,2} & \cdots & X_{k,2} & \cdots & X_{1,n} & \cdots & X_{k,n}
    \end{array}
    \right).
\end{equation*}
Observe that
$D^{ g}_i$ is non-zero only when
$Y_i=Y_{k,p}$
or
$Y_{i+m}=Y_{1,p}$
for some $p$.
See the cartoon in Figure \ref{fig:thm2:3} for a graphical illustration.
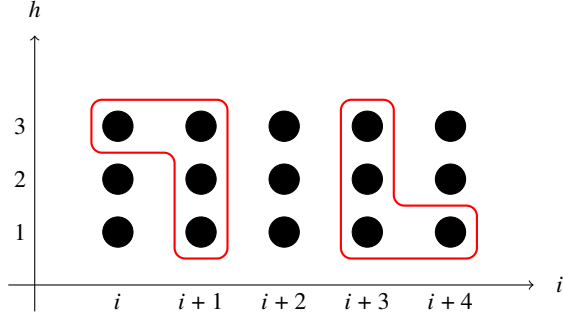
\begin{figure}[t]
        \centering
        \begin{tikzpicture}
            \tikzset{
                vertex/.style={circle,draw, inner sep=0pt, outer sep=0pt
                , minimum width=0.4cm
                ,fill=black
                },
            }
            \draw[->] (-\tikzTextMinDist,0)--(6*\tikzTextDist,0) node[below=1mm]{};
            \node at (6*\tikzTextDist+\tikzTextMinDist,0) [] {$i$};
            \draw[->] (0,-\tikzTextMinDist)--(0,3*\tikzTextDist) node[left]{};
            \node at (0,3*\tikzTextDist+\tikzTextMinDist) [] {$h$};
            \foreach \i in {1}{
                \pgfmathsetmacro\index{\i-4}
                \draw (\i*\tikzTextDist,0) node[below] {$i$};
            }
            \foreach \i in {2,...,5}{
                \pgfmathsetmacro\index{\i-1}
                \draw (\i*\tikzTextDist,0) node[below] {$i+\pgfmathprintnumber{\index}$};
            }
            \foreach \h in {1,...,3}{
                \draw (0,\h*\tikzNodeDist) node[left] {$\h$};
            }
            \foreach \i in {1,...,5}{
                \foreach \h in {1,2,3}{
                    \node at (\i*\tikzTextDist,\h*\tikzNodeDist) [vertex] {};
                }
            }
            \draw[rounded corners,red,thick]
            (1*\tikzTextDist-\tikzTextMinDist,4*\tikzNodeDist-\tikzNodeMinDist)
            --(1*\tikzTextDist-\tikzTextMinDist,3*\tikzNodeDist-\tikzNodeMinDist)
            --(2*\tikzTextDist-\tikzTextMinDist,3*\tikzNodeDist-\tikzNodeMinDist)
            --(2*\tikzTextDist-\tikzTextMinDist,1*\tikzNodeDist-\tikzNodeMinDist)
            --(2*\tikzTextDist+\tikzTextMinDist,1*\tikzNodeDist-\tikzNodeMinDist)
            --(2*\tikzTextDist+\tikzTextMinDist,4*\tikzNodeDist-\tikzNodeMinDist)
            --cycle  {};
            \draw[rounded corners,red,thick]
            (4*\tikzTextDist-\tikzTextMinDist,4*\tikzNodeDist-\tikzNodeMinDist)
            --(4*\tikzTextDist-\tikzTextMinDist,3*\tikzNodeDist-\tikzNodeMinDist)
            --(4*\tikzTextDist-\tikzTextMinDist,1*\tikzNodeDist-\tikzNodeMinDist)
            --(5*\tikzTextDist+\tikzTextMinDist,1*\tikzNodeDist-\tikzNodeMinDist)
            --(5*\tikzTextDist+\tikzTextMinDist,2*\tikzNodeDist-\tikzNodeMinDist)
            --(4*\tikzTextDist+\tikzTextMinDist,2*\tikzNodeDist-\tikzNodeMinDist)
            --(4*\tikzTextDist+\tikzTextMinDist,4*\tikzNodeDist-\tikzNodeMinDist)
            --cycle  {};
        \end{tikzpicture}
        \caption{Suppose $s(i)=3$ for all $i$ and $(m,r)=(3,0)$.
        The two difference statistics in red groups are examples of non-zero $D^{ g}_i$.
        }

        \label{fig:thm2:3}
\end{figure}
So there are at most $2(n-1)$ non-zero $D^{ g}_i$.
As a result, $\mathcal{D}=O(n)$ holds for balanced repetitional design.

\item[(ii)]
In this case, we suppose that
the data follows an imbalanced repetitional design and
$s(i)>1$ for all $i$.
Similarly, there are still at most $2(n-1)$ non-zero $D^{ g}_i$ following from the same arguments as in case (i),
i.e., $D^{ g}_i$ is non-zero only when $Y_i=Y_{k,p}$ or $Y_{i+m}=Y_{1,p}$ for some $p$.
Thus, $\mathcal{D}=O(n)$.

\item[(iii)]
In this case, we
suppose the data follows an imbalanced repetitional design and $s(i)\geq 1$ for all $i$.
Recall
$D^{ g}_i$ is non-zero
if there is at least one design point without duplication.
Observe that case (iii) is an extension of case (ii) where some non-duplicated design points are added.
Since one data point involves in at most $m+1$ difference statistics,
$(Y_{1,t},X_{1,t})$ adds at most $m-1$ to $\mathcal{D}$ if $s(t)=1$.
Note that $m+1-2=m-1$ is added
because $2$ non-zero $D^{ g}_i$, which satisfy $Y_i=Y_{1,t}$ or $Y_{i+m}=Y_{1,t}$,
have already been counted in $\mathcal{D}$, as argued in cases (i) and (ii).
So, $\mathcal{D} \leq 2(n-1)+S_n(m-1)$, where $S_n=\sum_{i=1}^n\mathbb{1}\left\{s(i)=1\right\}$.
\end{itemize}

In view of cases (i)--(iii) above,
we have $\mathcal{D}=O(n+S_nm)$.
Thus,

\begin{align*}
    \Bias\left\{\hat{\sigma}^2_{\Long}(m,r)\right\}
    &= O(n+S_nm)O_p(n^{-2r-2})/(N-m)\\
    &=O_p\left\{(n^{-2r-1}+n^{-2r-2}S_nm)/(N-m)\right\}.
\end{align*}

We continue to show the asymptotic variance of $\hat{\sigma}^2_{\Long}(m,r)$ remains the same as in Theorem {\ifnum\isXr=1{\ref{thm_diff_var}}\else{4.2}\fi}.
Recall that
$a_{i,j}$ is the $(i,j)$th element of the matrix $A$
and $g=(g(X_1),\ldots,g(X_N))^\T$.
Also recall that the maximum possible number of non-zero $D^{ g}_i$ is of order $O(n+S_nm)$ and $D^{ g}_i=O(n^{-r-1})\mathbb{1}(i\notin\mathcal{I}_m)$.
By similar arguments as in Theorem {\ifnum\isXr=1{\ref{thm_diff_var}}\else{4.2}\fi},
we show $\tr(A^2)$ and $\tr\left\{A\Diag(A)\right\}$ are the leading order terms.
First, we have
\begin{align}
    \tr(A^2)&=\sum_{i=1}^{N-m}\left(\sum^m_{j=0}d^{2}_{i,j}\right)^2+2\sum^m_{\ell=1}\sum_{i=1}^{N-m-\ell}\left(\sum^{m-\ell}_{j=0}d_{i,j+\ell}d_{i+\ell,j}\right)^2
    =O(Nm), \label{eq:thm2:2}
\end{align}
where (\ref{eq:thm2:2}) follows from the Cauchy--Schwarz
inequality that
\[
    \left(\sum^{m-\ell}_{j=0}d_{i,j+\ell}d_{i+\ell,j}\right)^2\leq \left(\sum^{m-\ell}_{j=0}d_{i,j+\ell}^{2}\right)\left(\sum^{m-\ell}_{j=0}d_{i+\ell,j}^{2}\right)\leq 1.
\]
For the term $ g^\T A^2g$, we have
\begin{align}
    g^\T A^2g&=\sum_{i=1}^{N-m}\left(\sum^m_{j=0}d^{2}_{i,j}\right)(D^{ g}_i)^2
    +2\sum^m_{\ell=1}\sum_{i=1}^{N-m-\ell}\left(\sum^{m-\ell}_{j=0}d_{i,j+\ell}d_{i+\ell,j}\right)D^{ g}_iD^{ g}_{i+\ell} \notag\\
    &=O(n+S_nm)O_p(n^{-2r-2}) + O(m)O(n+S_nm)O_p(n^{-2r-2}) \notag\\
    &=O(n^{-2r-1}m+n^{-2r-2}S_nm^2). \label{eq:thm2:5}
\end{align}
For the term $\tr\left\{A\Diag(Ag{1}^\T_{N})\right\}$, we have
\begin{align}
    \left|\tr\left\{A\Diag(Ag{1}^\T_{N})\right\}\right|
    &=\left|\sum_{i=1}^{N-m}\sum^m_{j=0}d^{2}_{i,j}\sum^{(i+j)\wedge (N-m)}_{h=1\vee (i+j-m)}d_{h,i+j-h}D^{ g}_h\right| \notag\\
    &\leq \sum_{i=1}^{N-m}\max_{0\leq j\leq m}\left|\sum^{(i+j)\wedge (N-m)}_{h=1\vee (i+j-m)}d_{h,i+j-h}D^{ g}_h\right|\sum^m_{j=0}d^2_{i,j} \notag\\
    &=O(n+S_nm)O(m)O_p(n^{-r-1}) \notag\\
    &=O_p(n^{-r}m+n^{-r-1}S_nm^2) .\notag
\end{align}
For the term $g^\T A\Diag(A)1_N$, we have
\begin{align}
    g^\T A\Diag(A)1_N&=\sum_{i=1}^{N-m}\left(\sum^m_{j=0}d_{i,j}a_{i+j,i+j}\right)D^{ g}_i \notag\\
    &=O(n+S_nm)O(m)O_p(n^{-r-1}) \label{eq:thm2:3}\\
    &=O_p(n^{-r}m+n^{-r-1}S_nm^{2}), \notag
\end{align}
where
(\ref{eq:thm2:3}) follows from the fact that $\sum^m_{j=0}d_{i,j}a_{i+j,i+j}=O(m)$ because
\begin{align*}
    \left\vert\sum^m_{j=0}d_{i,j}a_{i+j,i+j}  \right\vert
    &\leq \left(\sum^m_{j=0}d^{2}_{i,j} \right)^{1/2}\left(\sum^m_{j=0}a^2_{i+j,i+j}\right)^{1/2} \\
    &= \left(\sum^m_{j=0}a^2_{i+j,i+j}\right)^{1/2}
    = \left\{\sum^m_{j=0}\left(\sum_{h=0\vee(i-n+m)}^{m\wedge(i-1)}d_{i+j-h,h}^{2}\right)^2\right\}^{1/2}\\
    &\leq \left\{\left(\sum^m_{j=0}\sum_{h=0\vee(i-n+m)}^{m\wedge(i-1)}d_{i+j-h,h}^{2}\right)^2\right\}^{1/2}\\
    &= \sum^m_{j=0}\sum_{a=0\vee(i-n+m)-j}^{m\wedge(i-1)-j}d_{i-a,a+j}^{2} \quad(h-j=a)\\
    &\leq \sum_{a=0\vee(i-n+m)-m}^{m\wedge(i-1)}\sum^m_{j=0}d_{i-a,j}^{2}
    = O(m).
\end{align*}
For the term $\tr\left\{A\Diag(A)\right\}$, we have
\begin{align}
    \left|\tr\left\{A\Diag(A)\right\}\right|
    &=\left|\sum_{i=1}^{N-m}\sum^m_{j=0}d^{2}_{i,j}a_{i+j,i+j}\right| \notag
    \leq \sum_{i=1}^{N-m}\sum^m_{j=0}\left|1a_{i+j,i+j}\right| \notag\\
    &=O(N-m)O(m) \label{eq:thm2:4}\\
    &=O(Nm), \notag
\end{align}
where (\ref{eq:thm2:4}) follows from the fact that
\begin{align*}
    \sum^m_{j=0}\left|a_{i+j,i+j} \right|
    &= \sum^m_{j=0}\sum_{h=0\vee(i-n+m)}^{m\wedge(i-1)}d_{i+j-h,h}^{2}
    = \sum^m_{j=0}\sum_{a=0\vee(i-n+m)-j}^{m\wedge(i-1)-j}d_{i-a,a+j}^{2} \quad(h-j=a)\\
    &\leq \sum_{a=0\vee(i-n+m)-m}^{m\wedge(i-1)}\sum^m_{j=0}d_{i-a,j}^{2}
    =O(m).
\end{align*}
Collecting the above results,
we can rewrite the variance expression in (\ref{eqt:var_SGW1993})
for $\hat{\sigma}_\Long^2(m,r)$ as
\begin{align*}
    \Var\left\{\hat{\sigma}_\Long^2(m,r)\right\}&=\frac{1}{(N-m)^2}\left[2\sigma^4\tr(A^2)+(\lambda_4-3)\sigma^4\tr\left\{A\Diag(A)\right\}\right]\\
    &\quad+O_p\left\{\frac{n^{-r}m+n^{-r-1}S_nm^2}{(N-m)^2}\right\}.
\end{align*}

\subsection{Proof of Theorem {\ifnum\isXr=1{\ref{thm:repeat_varboost}}\else{4.4}\fi}}
First, we prove (i).
Without loss of generality,
let $\hat{\sigma}^2_{\Long,(h)}=Y^\T A_{(h)}Y/\tr(A_{(h)})$,
where $A_{(h)}=\DiffSeqMat_{(h)}^\T \DiffSeqMat_{(h)}$ in {\ifnum\isXr=1{(\ref{eq:Wmat_diffseq})}\else{(6)}\fi} under the $h$th re-arrangement of $(Y_i)_{i=1}^N$.
Let
$\DiffSeqMat=\DiffSeqMat_{(1)}$
and
$A=A_{(1)}=\DiffSeqMat^\T \DiffSeqMat$.
So, $\hat{\sigma}^2_{\Long}=\hat{\sigma}^2_{\Long,(1)}$.
Recall
$\DiffSeqMat=
\begin{pmatrix}
    v_1 & \cdots & v_N
\end{pmatrix}$.
It can be easily seen that
$\DiffSeqMat_{(h)}=
\begin{pmatrix}
    v_{\psi(1,h)} & \cdots & v_{\psi(N,h)}
\end{pmatrix}$.
In addition, the $(i,j)$th element of $A$ and $A_{(h)}$ is $v_i^\T v_j$ and $v_{\psi(i,h)}^\T v_{\psi(j,h)}$, respectively.
Then, we can rewrite $\hat{\sigma}^2_\VB$ as
\begin{align*}
    \hat{\sigma}^2_\VB &= \frac{1}{|\Psi|}\sum_{h=1}^{|\Psi|} \hat{\sigma}^2_{\Long,(h)}
    = \frac{1}{|\Psi|}\sum_{h=1}^{|\Psi|} \frac{Y^\T A_{(h)}Y}{\tr(A_{(h)})}
    = \frac{1}{N-m}Y^\T\left\{\sum_{h=1}^{|\Psi|} \frac{A_{(h)}}{|\Psi|}\right\}Y
\end{align*}
because $\tr(A_{(h)})=N-m$; see (\ref{eqt:trA}) for the explanation.
So we have $\Bar{A} =\sum_{h=1}^{|\Psi|}A_{(h)}/|\Psi|$.
By \citet{SGW1993},
the variance of $\hat{\sigma}^2_\VB(m,r)$ is
\begin{align*}
    \Var\left\{\hat{\sigma}^2_\VB(m,r)\right\}
    &=\frac{1}{\tr(\Bar{A})^2}\Bigg(2\sigma^4\tr(\Bar{A}^2)+4\sigma^2g^\T \Bar{A}^2g \\
    &\qquad+2\gamma_3\sigma^3\bigg[\tr\Big\{\Bar{A}\Diag(\Bar{A} g{1}^\T_{N})\Big\}+g^\T \Bar{A}\Diag(\Bar{A})1_N\bigg] \nonumber \\
    &\qquad +(\lambda_4-3)\sigma^4\tr\Big\{\Bar{A}\Diag(\Bar{A})\Big\}\Bigg).
\end{align*}
Recall $\Bar{a}_{i,j}$ is denoted as the $(i,j)$th element of $\Bar{A}$.
For the term $\tr(\Bar{A}^2)$, we have
\begin{align}
    \tr(\Bar{A}^2) &= \sum_{i=1}^N\sum_{j=1}^N \Bar{a}_{i,j}^2 \notag\\
    &= \sum_{i=1}^N \Bar{a}_{i,i}^2
    +\sum_{i=1}^N\sum_{j=1:i\neq j}^N \Bar{a}_{i,j}^2 \notag\\
    &= \sum_{i=1}^N\left[\sum_{k\in \mathcal{Q}_i}\frac{v_k^\T v_k}{s\left\{\alpha(i)\right\}}\right]^2 \notag\\
    &\quad+\sum_{i=1}^N\sum_{j=1:i\neq j}^N \left(
    \sum_{k\in\mathcal{Q}_i}\sum_{\ell\in\mathcal{Q}_j:\ell\neq k}\frac{v_k^\T v_\ell}{s\left\{\alpha(i)\right\}
    \left[s\left\{\alpha(j)\right\}-\mathbb{1}\left\{\alpha(i)=\alpha(j)\right\}\right]}
    \right)^2 \label{eq:thm:vb:1}\\
    &= \sum_{i=1}^n s(i)\left\{\sum_{k\in \mathcal{Q}^\star_i}\frac{v_k^\T v_k}{s(i)}\right\}^2 \notag
    \\
    &\quad+\sum_{i=1}^n\sum_{j=1}^n s(i)
    \left\{s(j)-\mathbb{1}\left(i=j\right)\right\}\left[
    \sum_{k\in\mathcal{Q}^\star_i}\sum_{\ell\in\mathcal{Q}^\star_j:\ell\neq k}\frac{v_k^\T v_\ell}{s(i)
    \left\{s(j)-\mathbb{1}\left(i=j\right)\right\}}
    \right]^2 \notag\\
    &= \sum_{i=1}^n \left\{\sum_{k\in \mathcal{Q}^\star_i}(v_k^\T v_k)^2 -\frac{1}{2s(i)}\sum_{u\in \mathcal{Q}^\star_i}\sum_{t\in \mathcal{Q}^\star_i}(v_u^\T v_u-v_t^\T v_t)^2\right\} \label{eq:thm:vb:2}
    \\
    &\quad+\sum_{i=1}^n\sum_{j=1}^n\left[
    \sum_{k\in\mathcal{Q}^\star_i}\sum_{\ell\in\mathcal{Q}^\star_j:\ell\neq k}(v_k^\T v_\ell)^2\right. \notag\\
    &\qquad-\left.\frac{\sum_{u\in\mathcal{Q}^\star_i}\sum_{t\in\mathcal{Q}^\star_j:t\neq u}
    \sum_{p\in\mathcal{Q}^\star_i}\sum_{q\in\mathcal{Q}^\star_j:q\neq p}(v_t^\T v_u-v_p^\T v_q)^2}{2s(i)
    \left\{s(j)-\mathbb{1}\left(i=j\right)\right\}}
    \right] \label{eq:thm:vb:3}\\
    &= \sum_{i=1}^N(v_i^\T v_i)^2 -\mathcal{A}_n
    +\sum_{k=1}^N\sum_{\ell=1:k\neq \ell}^N
    (v_k^\T v_\ell)^2
    -\mathcal{B}_n \notag\\
    &= \tr(A^2)-\mathcal{A}_n-\mathcal{B}_n, \notag
\end{align}
where (\ref{eq:thm:vb:1}) follows from {\ifnum\isXr=1{(\ref{eq:vb_Amat})}\else{(29)}\fi}, whereas
(\ref{eq:thm:vb:2}) and (\ref{eq:thm:vb:3}) follow from the formula:
\begin{align*}
    n\left(\frac{1}{n}\sum_{i=1}^nx_i\right)^2 =
    \sum_{i=1}^nx_i^2-\frac{1}{2n}\sum_{i=1}^n\sum_{j=1}^n(x_i-x_j)^2.
\end{align*}
For the term $\tr(\Bar{A}\Diag(\Bar{A}))$,
we have
\begin{align*}
    \tr(\Bar{A}\Diag(\Bar{A})) &= \sum_{i=1}^N\Bar{a}_{i,i}^2
    = \sum_{i=1}^N(v_i^\T v_i)^2 -\mathcal{A}_n
    = \tr\{ A\Diag(A) \} - \mathcal{A}_n.
\end{align*}
For the term $g^\T \Bar{A}^2 g$, we have
\begin{align*}
    g^\T \Bar{A}^2 g
    &= g^\T\left\{\sum_{h=1}^{|\Psi|} \frac{ A_{(h)}}{|\Psi|}\right\}^2 g
    = \frac{1}{|\Psi|^2}g^\T\left\{\sum_{h=1}^{|\Psi|} A_{(h)}\right\}^2 g  \\
    &= \frac{1}{|\Psi|^2}\sum_{h=1}^{|\Psi|}g^\T A_{(h)}^2 g
    +\frac{1}{|\Psi|^2}\sum_{k=1}^{|\Psi|}\sum_{\ell=1:\ell \neq k}^{|\Psi|}g^\T A_{(k)}A_{(\ell)} g\\
    &= O(n^{-2r-1}m+n^{-2r-2}S_nm^2),
\end{align*}
where the last line follows from (\ref{eq:thm2:5}) and the fact that
\begin{align}
    |g^\T A_{(k)}A_{(\ell)} g| &\leq \left(g^\T A_{(k)}^2 gg^\T A_{(\ell)}^2 g\right)^{1/2} \label{eq:thm:vb:6}\\
    &= O(n^{-2r-1}m+n^{-2r-2}S_nm^2), \label{eq:thm:vb:10}
\end{align}
where (\ref{eq:thm:vb:6}) follows from Cauchy--Schwarz inequality and (\ref{eq:thm:vb:10}) follows from
\begin{align*}
    g^\T A^2 g=g^\T A_{(1)}^2 g=\ldots = g^\T A_{(|\Psi|)}^2 g
    =O(n^{-2r-1}m+n^{-2r-2}S_nm^2).
\end{align*}

Let $D^g_{1:N-m}=\DiffSeqMat g
=
\begin{pmatrix}
    D^{ g}_1 & \cdots & D^{ g}_{N-m}
\end{pmatrix}^\T$.
For the term $g^\T \Bar{A}\Diag(\Bar{A})1_N$,
we have
\begin{align}
    g^\T \Bar{A}\Diag(\Bar{A})1_N &=
    g^\T \left\{\frac{1}{|\Psi|}\sum_{h=1}^{|\Psi|} \DiffSeqMat_{(h)}^\T \DiffSeqMat_{(h)}\right\}\Diag(\Bar{A})1_N \notag\\
    &= \left\{\frac{1}{|\Psi|}\sum_{h=1}^{|\Psi|} g^\T \DiffSeqMat_{(h)}^\T \DiffSeqMat_{(h)}\right\}
    \begin{pmatrix}
        \Bar{a}_{1,1}\\
        \vdots\\
        \Bar{a}_{N,N}
    \end{pmatrix} \notag\\
    &= \frac{1}{|\Psi|}\sum_{h=1}^{|\Psi|} (D^g_{1:N-m})^\T
    \sum_{i=1}^N \Bar{a}_{i,i}v_{\psi(i,h)} \label{eq:thm:vb:8}\\
    &= \frac{1}{|\Psi|}\sum_{h=1}^{|\Psi|} (D^g_{1:N-m})^\T
    \begin{pmatrix}
        \sum_{j=0}^m O(m)d_{1,j}\\
        \vdots\\
        \sum_{j=0}^m O(m)d_{N-m,j}
    \end{pmatrix} \label{eq:thm:vb:4}\\
    &= \frac{1}{|\Psi|}\sum_{h=1}^{|\Psi|}
    \sum_{i=1}^N D^g_i \sum_{j=0}^m O(m)d_{i,j} \notag\\
    &= O_p(n^{-r}m+n^{-r-1}S_nm^{2}), \label{eq:thm:vb:7}
\end{align}
where (\ref{eq:thm:vb:8}), (\ref{eq:thm:vb:4}) and (\ref{eq:thm:vb:7})
follows from the following observations:
\begin{itemize}
    \item (\ref{eq:thm:vb:8}) is obtained by noting that $\DiffSeqMat g=\DiffSeqMat_{(1)}g=\cdots=\DiffSeqMat_{(|\Psi|)}g=D^g_{1:N-m}$.
    \item (\ref{eq:thm:vb:4}) is obtained by noting that
\begin{align}
    0\leq \Bar{a}_{i,i}&=\frac{1}{s\left\{\alpha(i)\right\}}\sum_{k\in\mathcal{Q}_i}v_k^\T v_k
    \leq \frac{1}{s\left\{\alpha(i)\right\}} \left[m+s\left\{\alpha(i)\right\}\right] \label{eq:thm:vb:5}\\
    &\leq m+1, \notag
\end{align}
where the inequality in
(\ref{eq:thm:vb:5}) follows from the fact that sum of $b$ consecutive $v_k^\T v_k$ is upper bounded by $m+b$ because
\begin{align*}
    \sum_{\ell=0}^{b-1} v_{k+\ell}^\T v_{k+\ell}
    \leq
    \sum_{\ell=(k-m)\vee 0}^{(k+b-1)\wedge (N-m)}\sum_{j=0}^m d_{\ell,j}^{ 2}\leq m+b.
\end{align*}
    Hence, $\sup_{i=1, \ldots, N} |\bar{a}_{i,i}| = O(m)$.
    \item (\ref{eq:thm:vb:7}) is obtained from the fact that there are $O(n+S_n m)$ non-zero $D^{ g}_i$ of order $O(n^{-r-1})$
and that by Cauchy--Schwarz inequality,
\begin{align*}
    \left|\sum_{j=0}^mO(m)d_{i,j}\right|\leq \left\{O(m^2)\sum_{j=0}^md_{i,j}^{ 2}\right\}^{1/2}=O(m).
\end{align*}

\end{itemize}

For the term $\tr\left\{\Bar{A}\Diag(\Bar{A} g{1}^\T_{N})\right\}$, we have
\begin{align}
    \tr\left\{\Bar{A}\Diag(\Bar{A} g{1}^\T_{N})\right\} &= g^\T \Bar{A}\Diag(\Bar{A})1_N \label{eq:thm:vb:9}\\
    &= O_p(n^{-r}m+n^{-r-1}S_nm^{2}), \notag
\end{align}
where the identity (\ref{eq:thm:vb:9}) can be verified as follows:
\begin{align*}
    \tr\left\{\Bar{A}\Diag(\Bar{A} g{1}^\T_{N})\right\}&=
    \tr\left[\Bar{A}\Diag\left\{
    \begin{pmatrix}
        \sum_{j=0}^N \Bar{a}_{1,j}g_{j}\\
        \vdots\\
        \sum_{j=0}^N \Bar{a}_{N,j}g_{j}
    \end{pmatrix}
    {1}^\T_{N}\right\}\right]\\
    &=
    \tr\left\{\Bar{A}
    \begin{pmatrix}
        \sum_{j=0}^N \Bar{a}_{1,j}g_{j} & 0 & \cdots & 0 & 0\\
        \vdots & \vdots & \ddots & \vdots & \vdots\\
        0 & 0 & \cdots & 0 & \sum_{j=0}^N \Bar{a}_{N,j}g_{j}
    \end{pmatrix}\right\}
    = \sum_{i=1}^N \Bar{a}_{i,i}\sum_{j=0}^N \Bar{a}_{i,j}g_{j};\\
    g^\T \Bar{A}\Diag(\Bar{A})1_N  &=
    \begin{pmatrix}
        \sum_{j=0}^Ng_j \Bar{a}_{j,1} & \cdots & \sum_{j=0}^Ng_j \Bar{a}_{j,N}
    \end{pmatrix}
    \begin{pmatrix}
        \Bar{a}_{1,1}\\
        \vdots\\
        \Bar{a}_{N,N}
    \end{pmatrix}\\
    &=\sum_{i=1}^N \Bar{a}_{i,i}\sum_{j=0}^N \Bar{a}_{j,i}g_{j}
    =\sum_{i=1}^N \Bar{a}_{i,i}\sum_{j=0}^N \Bar{a}_{i,j}g_{j}
    = \tr\left\{\Bar{A}\Diag(\Bar{A} g{1}^\T_{N})\right\}.
\end{align*}
Collecting the above results, we have (i):
\begin{align*}
    \Var\left\{\hat{\sigma}^2_\VB(m,r)\right\}
    &=\frac{1}{(N-m)^2}\Bigg(2\sigma^4\tr(\Bar{A}^2)
    +(\lambda_4-3)\sigma^4\tr\Big\{\Bar{A}\Diag(\Bar{A})\Big\}\Bigg)\\
    &\quad+O_p\left\{\frac{n^{-r}m+n^{-r-1}S_nm^2}{(N-m)^2}\right\}.\\
    &=\frac{1}{(N-m)^2}\Bigg[2\sigma^4\tr(A^2)
    +(\lambda_4-3)\sigma^4\tr\Big\{A\Diag(A)\Big\}\\
    &\qquad-(\lambda_4-1)\sigma^4\mathcal{A}_n
    -2\sigma^4\mathcal{B}_n
    \Bigg]\\
    &\quad+O_p\left\{\frac{n^{-r}m+n^{-r-1}S_nm^2}{(N-m)^2}\right\}\\
    &=\frac{\sigma^4}{(N-m)^2}\theta(A,\lambda_4)
    -\frac{\sigma^4}{(N-m)^2}\left\{
    (\lambda_4-1)\mathcal{A}_n
    +2\mathcal{B}_n
    \right\}\\
    &\quad+O_p\left\{\frac{n^{-r}m+n^{-r-1}S_nm^2}{(N-m)^2}\right\}\\
    &=\Var\left\{\hat{\sigma}^2_\Long(m,r)\right\}
    -\frac{\sigma^4}{(N-m)^2}\left\{
    (\lambda_4-1)\mathcal{A}_n
    +2\mathcal{B}_n
    \right\}\\
    &\quad+O_p\left\{\frac{n^{-r}m+n^{-r-1}S_nm^2}{(N-m)^2}\right\}.
\end{align*}
Next, we prove (ii).
When $n=1$ and $0<m\leq N-1$, we have
$|\Psi|=N!$.
Then $\hat{\sigma}^2_{\VB}$ no longer depends on $r$ since intra-group differencing is not used.
We show that $\hat{\sigma}^2_{\VB}$ reduces to the sample variance for all $0<m\leq N-1$:
\begin{align*}
    \hat{\sigma}^2_\VB &= \frac{1}{|\Psi|}\sum_{h=1}^{|\Psi|} \hat{\sigma}^2_{\Long,(h)}\\
    &= \frac{1}{|\Psi|}\sum_{h=1}^{|\Psi|} \sum_{i=1}^{N-m} \frac{1}{N-m} \left(\sum_{j=0}^m d_{i,j} Y_{\psi(i+j,h)}\right)^2\\
    &= \frac{1}{|\Psi|} \frac{1}{N-m}\sum_{i=1}^{N-m}\sum_{h=1}^{|\Psi|} \left(\sum_{j=0}^m d_{i,j} Y_{\psi(i+j,h)}\right)^2\\
    &= \frac{1}{|\Psi|} \frac{1}{N-m} \sum_{i=1}^{N-m}(N-m-1)!
    \mathop{\mathop{\sum \cdots \sum}_{1\leq k_0,\ldots,k_m \leq N}}_{k_0, \ldots, k_m~\text{are distinct}}
    \left(\sum_{j=0}^m d_{i,j} Y_{k_{j}}\right)^2\\
    &= \frac{(N-m-1)!}{N!(N-m)} \sum_{i=1}^{N-m}
    \mathop{\mathop{\sum \cdots \sum}_{1\leq k_0,\ldots,k_m \leq N}}_{k_0, \ldots, k_m~\text{are distinct}}
    \left(\sum_{j=0}^m d_{i,j}^{ 2}Y_{k_{j}}^{ 2}
    +2\sum_{u=1}^m\sum_{t=0}^{u-1}d_{i,u}d_{i,t}Y_{k_u}Y_{k_t}\right).
\end{align*}
Let
\begin{align*}
    U_{i1}=\mathop{\mathop{\sum \cdots \sum}_{1\leq k_0,\ldots,k_m \leq N}}_{k_0, \ldots, k_m~\text{are distinct}}
    \sum_{j=0}^m d_{i,j}^{ 2}Y_{k_{j}}^{ 2},
    \quad
    U_{i2}=2\mathop{\mathop{\sum \cdots \sum}_{1\leq k_0,\ldots,k_m \leq N}}_{k_0, \ldots, k_m~\text{are distinct}}
    \sum_{u=1}^m\sum_{t=0}^{u-1}d_{i,u}d_{i,t}Y_{k_u}Y_{k_t}.
\end{align*}
Let $U_3=\sum_{i=1}^NY_i/N$ and $U_4=\sum_{i=1}^NY_i^{ 2}/N$.
By swapping the summation over $j$ and that over $k_0,\ldots,k_m$,
we have
\begin{align*}
    U_{i1}&=\sum_{j=0}^m\mathop{\mathop{\sum \cdots \sum}_{1\leq k_0,\ldots,k_m \leq N}}_{k_0, \ldots, k_m~\text{are distinct}}
     d_{i,j}^{ 2}Y_{k_{j}}^{ 2}\\
     &=\sum_{j=0}^m \sum_{k_j=1}^N
     \mathop{\mathop{\sum \cdots \sum}_{1\leq k_a \leq N:a\in\{1,\ldots,m\}\setminus \{j\}}}_{k_0, \ldots, k_m~\text{are distinct}}
     d_{i,j}^{ 2}Y_{k_{j}}^{ 2}
     =\sum_{j=0}^m \sum_{k_j=1}^N
     \frac{(N-1)!}{(N-m-1)!}
     d_{i,j}^{ 2}Y_{k_{j}}^{ 2}\\
     &=\frac{(N-1)!}{(N-m-1)!}\sum_{j=0}^m \sum_{k=1}^N
     d_{i,j}^{ 2}Y_{k}^{ 2}
     =\frac{N!}{(N-m-1)!}
     U_4,
\end{align*}
where the last line follows by noting that $\sum_{j=0}^m d_{i,j}^2=1$.
Similarly, for the term $U_{i,2}$, we have
\begin{align*}
    U_{i2}&=2\sum_{u=1}^m\sum_{t=0}^{u-1}
    \mathop{\mathop{\sum \cdots \sum}_{1\leq k_0,\ldots,k_m \leq N}}_{k_0, \ldots, k_m~\text{are distinct}}
    d_{i,u}d_{i,t}Y_{k_u}Y_{k_t}\\
    &=2\sum_{u=1}^m\sum_{t=0}^{u-1}
    \mathop{\mathop{\sum \sum}_{1\leq k_u, k_t \leq N}}_{k_u\neq k_t}
    \mathop{\mathop{\sum \cdots \sum}_{1\leq k_a \leq N:a\in\{1,\ldots,m\}\setminus \{u,t\}}}_{k_0, \ldots, k_m~\text{are distinct}}
    d_{i,u}d_{i,t}Y_{k_u}Y_{k_t}\\
    &=2\sum_{u=1}^m\sum_{t=0}^{u-1}
    \mathop{\mathop{\sum \sum}_{1\leq k_u, k_t \leq N}}_{k_u\neq k_t}
    \frac{(N-2)!}{(N-m-1)!}
    d_{i,u}d_{i,t}Y_{k_u}Y_{k_t}\\
    &=2\frac{(N-2)!}{(N-m-1)!}\sum_{u=1}^m\sum_{t=0}^{u-1}
    \mathop{\mathop{\sum \sum}_{1\leq k_a, k_b \leq N}}_{k_a\neq k_b}
    d_{i,u}d_{i,t}Y_{a}Y_{b}\\
    &=-\frac{(N-2)!}{(N-m-1)!}
    \left(\sum_{a=1}^N\sum_{b=1}^N Y_aY_b -\sum_{c}^N Y_c^{ 2}\right)\\
    &=-N\frac{(N-2)!}{(N-m-1)!}
    \left(NU_3^2 -U_{4}\right),
\end{align*}
where the second to the last line follows from $\sum_{u=1}^m\sum_{t=0}^{u-1}d_{i,u}d_{i,t}=-1/2$ because
\begin{align*}
    0 = \sum_{u=0}^m\sum_{t=0}^{m}d_{i,u}d_{i,t}
    = \sum_{u=0}^md_{i,u}^{ 2}+2\sum_{u=1}^m\sum_{t=0}^{u-1}d_{i,u}d_{i,t}
    \qquad \text{or} \qquad
    \sum_{u=1}^m\sum_{t=0}^{u-1}d_{i,u}d_{i,t}=-1/2.
\end{align*}
Collecting the above results, we have
\begin{align*}
    \hat{\sigma}^2_\VB
    &= \frac{(N-m-1)!}{N!(N-m)} \sum_{i=1}^{N-m}
    \left(U_{i,1}+U_{i,2}\right)\\
    &= \frac{(N-m-1)!}{N!(N-m)} \sum_{i=1}^{N-m}
    \left\{\frac{N!}{(N-m-1)!}
     U_4
     -N\frac{(N-2)!}{(N-m-1)!}
    \left(NU_3^2 -U_{4}\right)\right\}\\
    &= \frac{1}{(N-1)(N-m)} \sum_{i=1}^{N-m}
    \left\{NU_4
    -NU_3^2\right\}\\
    &=\frac{1}{N-1}\sum_{i=1}^N\left\{Y_i-\frac{1}{N}\sum_{j=1}^NY_{j}\right\}^2.
\end{align*}

\subsection{Proof of Theorem {\ifnum\isXr=1{\ref{thm:lamda4_est}}\else{5.2}\fi}}
By Theorems {\ifnum\isXr=1{\ref{thm_diff_bias}}\else{4.1}\fi} and {\ifnum\isXr=1{\ref{thm_diff_var}}\else{4.2}\fi}, it follows immediately that $\hat{\sigma}^{\dag 2}$ is weakly consistent, i.e., $\hat{\sigma}^{\dag 2}\rightarrow \sigma^2$ in probability.

Let $D_i^{g*}=\sum^{m^\dag}_{j=0}d_{i,j}^*g(X_{i+j})$,
$D_i^{\epsilon*}=\sum^{m^\dag}_{j=0}d_{i,j}^*\epsilon_{i+j}$,
$D_i^*=\sum^{m^\dag}_{j=0}d_{i,j}^*y_{i+j}=D_i^{g*}+D_i^{\epsilon*}$ and
$A_1=\sum^{N-m^\dag}_{i=1}D_i^{*4}/(N-m^\dag)$.
Since $r\geq 0$ and $g(\cdot)$ is continuously differentiable,
$D_i^{g*}$ is at least of order $O(N^{-1})$.
Recall that all expectations are conditioned upon $(X_i)_{i=1}^N$.
We have the following results:
\begin{eqnarray*}
    \E\big\{(D_i^{g*})^4\big\}&=&\left\{\sum^{m^\dag}_{j=0}d_{i,j}^*g(X_{i+j})\right\}^4
    =\left\{\sum^{m^\dag}_{j=0}d_{i,j}^\dag g(X_{i+j})\right\}^4/\sum^{m^\dag}_{j=0}d_{i,j}^{\dag 4}
    =O_p(N^{-4});\\
    \E\left\{(D_i^{g*})^3D_i^{\epsilon*}\right\}
    &=&0;\\
    \E\left\{(D_i^{g*})^2(D_i^{\epsilon*})^{2}\right\}&=&(D_i^{g*})^2\E\left\{(D_i^{\epsilon*})^{2}\right\}
    =(D_i^{g*})^2\sigma^2/\left(\sum^{m^\dag}_{j=0}d_{i,j}^{\dag 4}\right)^{1/2}
    =O_p(N^{-2});\\
    \E\left\{D_i^{g*}(D_i^{\epsilon*})^3\right\}
    &=&O_p(N^{-1})\E\left(\sum^{m^\dag}_{j=0}d_{i,j}^{\dag 3}\epsilon_{i+j}^{3}\right)
    =O_p(N^{-1})\sum^{m^\dag}_{j=0}d_{i,j}^{\dag 3}\lambda_3\sigma^3
    =O_p(N^{-1});\\
    \E\left\{(D_i^{\epsilon*})^4\right\}&=&\E\left\{\left(\sum^{m^\dag}_{j=0}d_{i,j}^\dag \epsilon_{i+j}\right)^{4}/\sum^{m^\dag}_{j=0}d_{i,j}^{\dag 4}\right\}\\
    &=&\E\left\{\sum_{k_0+\cdots+k_{m^\dag}=4;k_0,\ldots,k_{m^\dag}\geq 0}\binom{4}{k_0,\ldots,k_{m^\dag}}\prod^{m^\dag}_{j=0}(d_{i,j}^\dag \epsilon_{i+j})^{k_j}\right\}/\sum^{m^\dag}_{j=0}d_{i,j}^{\dag 4}\\
    &=&\E\left\{\sum^{m^\dag}_{j=0}d_{i,j}^{\dag 4}\epsilon_{i+j}^{4}+\sum^{m^\dag-1}_{j=0}\sum^{m^\dag}_{k=j+1}\frac{4!}{2!2!}(d_{i,j}^\dag\epsilon_{i+j}d_{i,k}^\dag\epsilon_{i+k})^2\right\}/\sum^{m^\dag}_{j=0}d_{i,j}^{\dag4}\\
    &=&\left(\sum^{m^\dag}_{j=0}d_{i,j}^{\dag 4}\lambda_4\sigma^4+\sum^{m^\dag-1}_{j=0}\sum^{m^\dag}_{k=j+1}6d_{i,j}^{\dag2}d_{i,k}^{\dag2}\sigma^4\right)/\sum^{m^\dag}_{j=0}d_{i,j}^{\dag4}\\
    &=&\lambda_4\sigma^4+\sum^{m^\dag-1}_{j=0}\sum^{m^\dag}_{k=j+1}6d_{i,j}^{\dag2}d_{i,k}^{\dag2}\sigma^4/\sum^{m^\dag}_{j=0}d_{i,j}^{\dag4}.
\end{eqnarray*}
Therefore, we have
\begin{eqnarray*}
    \E(A_1)&=&\frac{1}{N-m^\dag }\sum^{N-m^\dag }_{i=1}\E(D_i^{*4})\\
    &=&\frac{1}{N-m^\dag }\sum^{N-m^\dag }_{i=1}\left\{\lambda_4\sigma^4+\sum^{m^\dag-1}_{j=0}\sum^{m^\dag}_{k=j+1}6d_{i,j}^{\dag2}d_{i,k}^{\dag2}\sigma^4/\sum^{m^\dag}_{j=0}d_{i,j}^{\dag4}+O(N^{-1})\right\}\\
    &=&\lambda_4\sigma^4+\frac{1}{N-m^\dag }\sum^{N-m^\dag }_{i=1}\left(\sum^{m^\dag-1}_{j=0}\sum^{m^\dag}_{k=j+1}6d_{i,j}^{\dag2}d_{i,k}^{\dag2}\sigma^4/\sum^{m^\dag}_{j=0}d_{i,j}^{\dag4}\right)+O(N^{-1}).
\end{eqnarray*}
Now we show that $\Var(A_1)$ converges to $0$. Since $m$ is finite and fixed and $\E(\epsilon_i^8/\sigma^8)<\infty$, it is obvious that $\Var(D_i^{*4})=O(1)$. Also observe that
\begin{eqnarray*}
    \Cov(D_i^{*4},D_{i+p}^{*4})
        = \left\{
            \begin{array}{ll}
            O(1) & \text{if} \quad |p|\leq m, \\
            0 & \text{if} \quad |p|>m.
            \end{array}
        \right.
\end{eqnarray*}
Therefore,
\begin{eqnarray*}
    \Var(A_1)&=&\frac{1}{(N-m^\dag )^2}\sum^{N-m^\dag }_{i=1}\Var(D_i^{*4})+\frac{2}{(N-m^\dag )^2}\sum^{m}_{p=1}\sum^{N-m^\dag -p}_{i=1}\Cov(D_i^{*4},D_{i+p}^{*4})\\
    &=&O(N^{-1}).
\end{eqnarray*}
Hence, $A_1-\lambda_4\sigma^4+\sum^{N-m^\dag }_{i=1}\left(\sum^{m^\dag-1}_{j=0}\sum^{m^\dag}_{k=j+1}6d_{i,j}^{\dag2}d_{i,k}^{\dag2}\sigma^4/\sum^{m^\dag}_{j=0}d_{i,j}^{\dag4}\right)/(N-m^\dag)\rightarrow 0$ in probability.
Finally, as $\hat{\sigma}^{\dag2}$ and $A_1$ are consistent, by continuous mapping theorem, $\hat{\lambda}_4^\dag$ is a consistent estimator of $\lambda_4$,
i.e., $\hat{\lambda}_4^\dag\rightarrow \lambda_4$ in probability.

\subsection{Proof of Corollary {\ifnum\isXr=1{\ref{coro:equiv_diff_est}}\else{1}\fi}}
First we show equivalence between $\hat{\sigma}_*^2(m,0)$, $\hat{\sigma}^2(m,0)$ and variance-optimal differencing estimator by \citet{HKT1990}.
It suffices to show that
$\hat{\sigma}_*^2(m,0)$ and $\hat{\sigma}^2(m,0)$ uses the same global difference sequence as $\hat{\sigma}_\textsc{h}$.
Since $r=0$, the constraints {\ifnum\isXr=1{(\ref{diff_poly_random_cond})}\else{(11)}\fi} and {\ifnum\isXr=1{(\ref{diff_poly_cond})}\else{(12)}\fi} are equivalent and reduce to
\begin{align*}
    \sum_{j=0}^m d_j=0,
\end{align*}
which coincide with the second constraint in {\ifnum\isXr=1{(\ref{diff_poly_cond0})}\else{(7)}\fi}.
It can be easily seen that $\hat{\sigma}_*^2(m,0)$, $\hat{\sigma}^2(m,0)$ and $\hat{\sigma}_\textsc{h}$ are equivalent
since they are restricted by the same constraints and minimize the same quantity $\delta(d)$.

Second, we show equivalence between $\hat{\sigma}_*^2(2,1)$ and estimator proposed by \citet{GSJ1986}.
Given the data $(Y_i,X_i)_{i=1}^n$,
$\hat{\sigma}^2_{\textsc{g}}$ from \citet{GSJ1986} needs the following definitions:
\begin{eqnarray*}
    &\tilde{D}_i&=\frac{X_{i+2}-X_{i+1}}{(X_{i+2}-X_i)c_i}Y_{i}+\frac{-1}{c_i}Y_{i+1}+\frac{X_{i+1}-X_i}{(X_{i+2}-X_i)c_i}Y_{i+2},\\
    &c_i^2&=\left(\frac{X_{i+2}-X_{i+1}}{X_{i+2}-X_i}\right)^2+1+\left(\frac{X_{i+1}-X_i}{X_{i+2}-X_i}\right)^2,\\
    &\hat{\sigma}^2_{\textsc{g}}&=\frac{1}{n-2}\sum_{i=1}^{n-2}\tilde{D}_i^2.
\end{eqnarray*}
Let $A_{j}=X_{j}-X_i$.
The dependence in $i$ is omitted for simplicity of notations.
We can rewrite $\tilde{D}_i$ as
\begin{eqnarray*}
    \tilde{D}_i&=&\frac{A_{i+2}-A_{i+1}}{\left\{(A_{i+2}-A_{i+1})^2+A_{i+2}^2+A_{i+1}^2\right\}^{1/2}}\hspace{2pt}Y_i\\
    &&+\frac{-A_{i+2}}{\left\{(A_{i+2}-A_{i+1})^2+A_{i+2}^2+A_{i+1}^2\right\}^{1/2}}\hspace{2pt}Y_{i+1}
    +\frac{A_{i+1}}{\left\{(A_{i+2}-A_{i+1})^2+A_{i+2}^2+A_{i+1}^2\right\}^{1/2}}\hspace{2pt}Y_{i+2}.
\end{eqnarray*}
To show $\hat{\sigma}^2_{\textsc{g}}=\hat{\sigma}_*^2(2,1)$,
it is equivalent to show that
\begin{align}
        d_{i,0}&=\frac{A_{i+2}-A_{i+1}}{\left\{(A_{i+2}-A_{i+1})^2+A_{i+2}^2+A_{i+1}^2\right\}^{1/2}}, \label{eq:GSJeqProposal1}\\
        d_{i,1}&=\frac{-A_{i+2}}{\left\{(A_{i+2}-A_{i+1})^2+A_{i+2}^2+A_{i+1}^2\right\}^{1/2}}, \label{eq:GSJeqProposal2}\\
        d_{i,2}&=\frac{A_{i+1}}{\left\{(A_{i+2}-A_{i+1})^2+A_{i+2}^2+A_{i+1}^2\right\}^{1/2}}.   \label{eq:GSJeqProposal3}
\end{align}
Now we show (\ref{eq:GSJeqProposal1})--(\ref{eq:GSJeqProposal3}) hold.
According to \S~{\ifnum\isXr=1{\ref{section_algorithm}}\else{5}\fi}, let
\begin{equation*}
    \mathcal{L}_i=
    \begin{pmatrix}
        1 & 1 & 1\\
        0 & X_{i+1}-X_i & X_{i+2}-X_i
    \end{pmatrix}.
\end{equation*}
The reduced row echelon form of $\mathcal{L}_i$ yields
\begin{equation*}
    \mathcal{L}_i'=
    \begin{pmatrix}
        1 & 0 & \frac{A_{i+1}-A_{i+2}}{A_{i+1}}\\
        0 & 1 & \frac{A_{i+2}}{A_{i+1}}
    \end{pmatrix}.
\end{equation*}
By {\ifnum\isXr=1{(\ref{independent_phi_formula})}\else{(24)}\fi}, we have
\begin{eqnarray*}
    \phi_{i,2}=-\arctan\left(\frac{A_{i+1}}{A_{i+2}}\right)
    \qquad \text{and} \qquad
    \phi_{i,1}=-\arctan\left\{\frac{(A_{i+1}^2+A_{i+2}^2)^{1/2}}{A_{i+2}-A_{i+1}}\right\}.
\end{eqnarray*}
The parametrization {\ifnum\isXr=1{(\ref{patternvector})}\else{(23)}\fi} leads to
\begin{eqnarray*}
    d_{i,0}&=&\cos\phi_{i,1}\\
    &=&1/\left\{1+\frac{A_{i+1}^2+A_{i+2}^2}{(A_{i+2}-A_{i+1})^2}\right\}^{1/2}\\
    &=&\frac{A_{i+2}-A_{i+1}}{\left\{(A_{i+2}-A_{i+1})^2+A_{i+2}^2+A_{i+1}^2\right\}^{1/2}},\\
    d_{i,1}&=&\sin\phi_{i,1}\cos\phi_{i,2}\\
    &=&-\frac{(A_{i+1}^2+A_{i+2}^2)^{1/2}/(A_{i+2}-A_{i+1})}{\left\{1+(A_{i+1}^2+A_{i+2}^2)/(A_{i+2}-A_{i+1})^2\right\}^{1/2}}\frac{1}{\left(1+A_{i+1}/A_{i+2}\right)^{1/2}}\\
    &=&\frac{-1}{\left\{(A_{i+2}-A_{i+1})^2/A_{i+2}^2+1+A_{i+1}^2/A_{i+2}^2\right\}^{1/2}}\\
    &=&\frac{-A_{i+2}}{\left\{(A_{i+2}-A_{i+1})^2+A_{i+2}^2+A_{i+1}^2\right\}^{1/2}},\\
    d_{i,2}&=&\sin\phi_{i,1}\sin\phi_{i,2}\\
    &=&\frac{-A_{i+1}/A_{i+2}}{\left(1+A_{i+1}^2/A_{i+2}^2\right)^{1/2}}\frac{-(A_{i+1}^2+A_{i+2}^2)^{1/2}/(A_{i+2}-A_{i+1})}{\left\{1+(A_{i+1}^2+A_{i+2}^2)/(A_{i+2}-A_{i+1})^2\right\}^{1/2}}\\
    &=&\frac{A_{i+1}}{\left\{(A_{i+2}-A_{i+1})^2+A_{i+2}^2+A_{i+1}^2\right\}^{1/2}}.
\end{eqnarray*}
Therefore, $\hat{\sigma}^2_{\textsc{g}}\equiv\hat{\sigma}_*^2(2,1)$.

\subsection{Proof of Corollary {\ifnum\isXr=1{\ref{coro:long_adj_bias}}\else{2}\fi}}

First, we derive the bias for the case $m\geq m'_{\min}$.
Denote $\mathcal{I}$ as the set of indices $i$ such that $D_i$ is repetitional-bias-corrected.
Then we have
\begin{align}
    \Bias\big\{\hat{\sigma}^2_{\Long}(m,r)\big\}
    &=\frac{1}{|\mathcal{I}|}\sum_{i\in\mathcal{I}}\left(D^{ g}_i\right)^2 \notag\\
    &=\frac{1}{N-m-\mathcal{D}}
    \left\{
    \sum_{i\in\mathcal{I}_m}\left(D^{ g}_i\right)^2
    +\sum_{i\in\mathcal{I}\setminus\mathcal{I}_m}\left(D^{ g}_i\right)^2
    \right\} \notag\\
    &=\frac{1}{N-m-\mathcal{D}}\sum_{i\in\mathcal{I}\setminus\mathcal{I}_m}O_p(n^{-2r-2}) , \label{eq:coro1:bias1}
\end{align}
where
$\mathcal{D}$ is the number of difference statistics that are not repetitional-bias-corrected
and
(\ref{eq:coro1:bias1}) follows from the fact that $D^{ g}_i=0$ for $i\in\mathcal{I}_m$ and $D^{ g}_i=O_p(n^{-2r-2})$ for $i\in\mathcal{I}\setminus\mathcal{I}_m$.
To prove that (\ref{eq:coro1:bias1}) is equal to the right-hand side of {\ifnum\isXr=1{(\ref{eq:thm2:bias})}\else{(20)}\fi} asymptotically,
it suffices to show that $\mathcal{D}/N\rightarrow 0$.
It also suffices to consider the case where $r>0$
since the case $r=0$ is trivial.
Let
\begin{align}\label{eq:coro:1}
        \mathcal{U}=\bigcup_{i'\in\mathcal{T}_s}\mathcal{U}_{i'},
\end{align}
where $\mathcal{U}_{i'}$ is the set of indices $i$ such that
$D_{i}$ are not repetitional-bias-corrected due to the observations $\{Y_{1,i'},\ldots,Y_{s(i'),i'}\}$.
Note that $\mathcal{U}_{i'}$ may not be disjoint over $i'\in\mathcal{T}_s$.
The set $\mathcal{U}$ contains indices $i$ such that $D_{i}$ is not repetitional-bias-corrected,
hence
\begin{align}\label{eqt:Dp_upperBound}
\mathcal{D}\equiv|\mathcal{U}|\leq \sum_{i'\in\mathcal{T}_s}|\mathcal{U}_{i'}| .
\end{align}
We show $\mathcal{D}=O(S'_n)$ with the following cases.
From now on, we denote
    $Y_{i}$ as the leading observation and $Y_{i+m}$ as the ending observation
    in the difference statistic $D_{i}$.
Also recall that $\mathcal{T}_s=\{i\in[1,\ldots,n]\cap\mathbb{Z}:s(i)>s_{\max}\}$.
Two cases are considered.
\begin{itemize}
    \item[(i)] Suppose $s(i)>1$ for all $i$.
    Consider a particular index $i'\in\mathcal{T}_s$.
    Similar to the proof of Theorem {\ifnum\isXr=1{\ref{thm:long_diff_bias}}\else{4.3}\fi}, the difference statistic $D_i$ is possibly not repetitional-bias-corrected,
    i.e., $i\in\mathcal{U}_{i'}$,
    only when
    $Y_i=Y_{s(i'-a),i'-a}$
    or
    $Y_{i+m}=Y_{1,i'+a}$,
    for some $a=1,\ldots,r$,
    i.e., when either the leading observation is $Y_{s(i'-a),i'-a}$ or the ending observation is $Y_{1,i'+a}$.
    So, we have
    \begin{align}\label{eqt:upperbound_Uip}
        |\mathcal{U}_{i'}|\leq 2r .
    \end{align}
    Therefore, by (\ref{eqt:Dp_upperBound}) and (\ref{eqt:upperbound_Uip}), we have
    \begin{align*}
        \mathcal{D}\leq\sum_{i\in\mathcal{T}_s}|\mathcal{U}_{i'}|\leq 2r|\mathcal{T}_s|.
    \end{align*}
    See the cartoon in Figure \ref{fig:coro1:1} for a graphical illustration.
    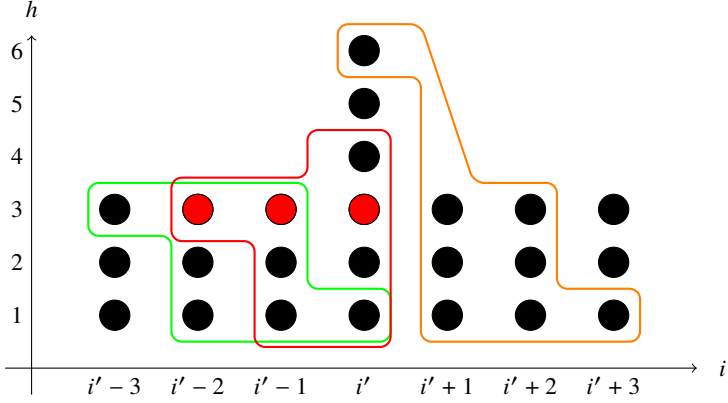
\begin{figure}[t]
        \centering
        \begin{tikzpicture}
            \tikzset{
                vertex/.style={circle,draw, inner sep=0pt, outer sep=0pt
                , minimum width=0.4cm
                ,fill=black
                },
            }
            \draw[->] (-\tikzTextMinDist,0)--(8*\tikzTextDist,0) node[below=1mm]{};
            \node at (8*\tikzTextDist+\tikzTextMinDist,0) [] {$i$};
            \draw[->] (0,-\tikzTextMinDist)--(0,4*\tikzTextDist) node[left]{};
            \node at (0,4*\tikzTextDist+\tikzTextMinDist) [] {$h$};
            \foreach \i in {1,...,3}{
                \pgfmathsetmacro\index{\i-4}
                \draw (\i*\tikzTextDist,0) node[below] {$i'\pgfmathprintnumber{\index}$};
            }
            \foreach \i in {4}{
                \pgfmathsetmacro\index{\i-4}
                \draw (\i*\tikzTextDist,0) node[below] {$i'$};
            }
            \foreach \i in {5,...,7}{
                \pgfmathsetmacro\index{\i-4}
                \draw (\i*\tikzTextDist,0) node[below] {$i'+\pgfmathprintnumber{\index}$};
            }
            \foreach \h in {1,...,6}{
                \draw (0,\h*\tikzNodeDist) node[left] {$\h$};
            }
            \foreach \i in {1,...,7}{
                \foreach \h in {1,2,3}{
                    \node at (\i*\tikzTextDist,\h*\tikzNodeDist) [vertex] {};
                }
            }
            \foreach \h in {4,...,6}{
                \node at (4*\tikzTextDist,\h*\tikzNodeDist) [vertex] {};
            }
            \foreach \i in {2,3,4}{
                \foreach \h in {3}{
                    \node at (\i*\tikzTextDist,\h*\tikzNodeDist) [vertex,fill=red] {};
                }
            }
            \draw[rounded corners,green,thick]
            (1*\tikzTextDist-\tikzTextMinDist,4*\tikzNodeDist-\tikzNodeMinDist)
            --(1*\tikzTextDist-\tikzTextMinDist,3*\tikzNodeDist-\tikzNodeMinDist)
            --(2*\tikzTextDist-\tikzTextMinDist,3*\tikzNodeDist-\tikzNodeMinDist)
            --(2*\tikzTextDist-\tikzTextMinDist,1*\tikzNodeDist-\tikzNodeMinDist)
            --(4*\tikzTextDist+\tikzTextMinDist,1*\tikzNodeDist-\tikzNodeMinDist)
            --(4*\tikzTextDist+\tikzTextMinDist,2*\tikzNodeDist-\tikzNodeMinDist)
            --(3*\tikzTextDist+\tikzTextMinDist,2*\tikzNodeDist-\tikzNodeMinDist)
            --(3*\tikzTextDist+\tikzTextMinDist,4*\tikzNodeDist-\tikzNodeMinDist)
            --cycle  {};
            \draw[rounded corners,orange,thick]
            (4*\tikzTextDist-\tikzTextMinDist,7*\tikzNodeDist-\tikzNodeMinDist)
            --(4*\tikzTextDist-\tikzTextMinDist,6*\tikzNodeDist-\tikzNodeMinDist)
            --(5*\tikzTextDist-\tikzTextMinDist,6*\tikzNodeDist-\tikzNodeMinDist)
            --(5*\tikzTextDist-\tikzTextMinDist,1*\tikzNodeDist-\tikzNodeMinDist)
            --(7*\tikzTextDist+\tikzTextMinDist,1*\tikzNodeDist-\tikzNodeMinDist)
            --(7*\tikzTextDist+\tikzTextMinDist,2*\tikzNodeDist-\tikzNodeMinDist)--(6*\tikzTextDist+\tikzTextMinDist,2*\tikzNodeDist-\tikzNodeMinDist)
            --(6*\tikzTextDist+\tikzTextMinDist,4*\tikzNodeDist-\tikzNodeMinDist)
            --(5*\tikzTextDist+\tikzTextMinDist,4*\tikzNodeDist-\tikzNodeMinDist)
            --(5*\tikzTextDist-\tikzTextMinDist,7*\tikzNodeDist-\tikzNodeMinDist)
            --cycle  {};
            \draw[rounded corners,red,thick]
            (2*\tikzTextDist-\tikzTextMinDist,3*\tikzNodeDist+1.2*\tikzNodeMinDist)
            --(2*\tikzTextDist-\tikzTextMinDist,3*\tikzNodeDist-1.2*\tikzNodeMinDist)
            --(3*\tikzTextDist-\tikzTextMinDist,3*\tikzNodeDist-1.2*\tikzNodeMinDist)
            --(3*\tikzTextDist-\tikzTextMinDist,1*\tikzNodeDist-1.2*\tikzNodeMinDist)
            --(4*\tikzTextDist+\tikzTextMinDist,1*\tikzNodeDist-1.2*\tikzNodeMinDist)
            --(4*\tikzTextDist+\tikzTextMinDist,5*\tikzNodeDist-\tikzNodeMinDist)
            --(3*\tikzTextDist+\tikzTextMinDist,5*\tikzNodeDist-\tikzNodeMinDist)
            --(3*\tikzTextDist+\tikzTextMinDist,3*\tikzNodeDist+1.2*\tikzNodeMinDist)
            --cycle;
        \end{tikzpicture}
        \caption{Suppose $s(i)=3$ for all $i$ except $s(i')=6$.
        Set $(m,r)=(m'_{\min},2)=(7,2)$.
        The red nodes corresponds to the indices $i\in\mathcal{U}_{i'}$.
         The red group is the difference statistics which
        satisfy $Y_{i}=Y_{3,i'-2}$ and thus is not repetitional-bias-corrected.
        The difference statistics in green group and orange group satisfies $Y_{i}=Y_{s(i'-r-1),i'-r-1}$ and $Y_{i+m}=Y_{s(i'+r+1),i'+r+1}$, respectively, but still repetitional-bias-corrected.
        It is noteworthy that $D_{\ell(i')-1}$ satisfies both $Y_{i}=Y_{s(i'-1),i'-1}$ and $Y_{i+m}=Y_{1,i'+1}$.
        So, $|\mathcal{U}_{i'}|=3<2r$.
        }

        \label{fig:coro1:1}
    \end{figure}

    \item[(ii)] Suppose $s(i)\geq 1$ for all $i$, which includes case (i) as a special case.
    Consider each index $i'\in\mathcal{T}_s$.
    \begin{itemize}
        \item If $s(i'-a)=1$ for some integer $1\leq a\leq r$,
    then
    the difference statistic $D_{i}$ is possibly not repetitional-bias-corrected for all $i\in[\ell(i'-r),\ell(i'-a)]\cap\mathbb{Z}$.
    So, there are at most
    $
        \sum_{j=i'-r}^{i'-a}\left\{s(j)-1\right\}
    $
    additional difference statistics that are not repetitional-bias-corrected
    compared with (\ref{eqt:upperbound_Uip}),
    where
    $1$ is subtracted because one possibly
    non-repetitional-bias-corrected difference statistics
    satisfying $Y_i=Y_{s(i'-b),i'-b}$ for each $b=a,\ldots,r$
    have been counted in the upper bound of (\ref{eqt:upperbound_Uip}).
    Hence, in this case
    \begin{align}
        |\mathcal{U}_{i'}| &\leq 2r +  \sum_{j=i'-r}^{i'-a}\left\{s(j)-1\right\} \notag\\
        &=2r +  \sum_{j=i'-r}^{i'-1}\left\{
        s(j)-1\right\}\mathbb{1}\left\{\min_{k=j,\ldots,i'-1}s(k)=1\right\}
         \notag\\
        &=r+\sum_{j=i'-r}^{i'-1}
        s(j)\mathbb{1}\left\{\min_{k=j,\ldots,i'-1}s(k)=1\right\}
        +r-\sum_{j=i'-r}^{i'-1}\mathbb{1}\left\{\min_{k=j,\ldots,i'-1}s(k)=1\right\} \notag\\
        &=r+\sum_{j=i'-r}^{i'-1}
        s(j)\mathbb{1}\left\{\min_{k=j,\ldots,i'-1}s(k)=1\right\}
        +\sum_{j=i'-r}^{i'-1}\mathbb{1}\left\{\min_{k=j,\ldots,i'-1}s(k)>1\right\}. \label{eqt:caseii_upperboundUip_1}
    \end{align}
        \item     Similarly,
    if $s(i'+a)=1$ for some integer $1\leq a\leq r$,
    then
    the difference statistics $D_{i}$ is possibly not repetitional-bias-corrected for all $i\in[\ell(i'+a)-m,\ell(i'+r)+s(i'+r)-1-m]\cap\mathbb{Z}$.
    So, there are at most $\sum_{j=i'+a}^{i'+r}\left\{s(j)-1\right\}$ more difference statistics that are not repetitional-bias-corrected.
    Hence, in this case,
    \begin{align}
        |\mathcal{U}_{i'}| &\leq 2r +\sum_{j=i'+a}^{i'+r}\left\{s(j)-1\right\}
         \notag\\
        &=r+\sum_{j=i'+1}^{i'+r}
        s(j)\mathbb{1}\left\{\min_{k=i'+1,\ldots,j}s(k)=1\right\}
        +\sum_{j=i'+1}^{i'+r}\mathbb{1}\left\{\min_{k=i'+1,\ldots,j}s(k)>1\right\} . \label{eqt:caseii_upperboundUip_2}
    \end{align}
    \end{itemize}

    See the cartoon in Figure \ref{fig:coro1:2} for a graphical illustration.
    \begin{figure}[t]
        \centering
        \begin{tikzpicture}
            \tikzset{
                vertex/.style={circle,draw, inner sep=0pt, outer sep=0pt
                , minimum width=0.4cm
                ,fill=black
                },
            }
            \draw[->] (-\tikzTextMinDist,0)--(8*\tikzTextDist,0) node[below=1mm]{};
            \node at (8*\tikzTextDist+\tikzTextMinDist,0) [] {$i$};
            \draw[->] (0,-\tikzTextMinDist)--(0,4*\tikzTextDist) node[left]{};
            \node at (0,4*\tikzTextDist+\tikzTextMinDist) [] {$h$};
            \foreach \i in {1,...,3}{
                \pgfmathsetmacro\index{\i-4}
                \draw (\i*\tikzTextDist,0) node[below] {$i'\pgfmathprintnumber{\index}$};
            }
            \foreach \i in {4}{
                \pgfmathsetmacro\index{\i-4}
                \draw (\i*\tikzTextDist,0) node[below] {$i'$};
            }
            \foreach \i in {5,...,7}{
                \pgfmathsetmacro\index{\i-4}
                \draw (\i*\tikzTextDist,0) node[below] {$i'+\pgfmathprintnumber{\index}$};
            }
            \foreach \h in {1,...,6}{
                \draw (0,\h*\tikzNodeDist) node[left] {$\h$};
            }
            \foreach \i in {1,3,4,6,7}{
                \foreach \h in {1,2,3}{
                    \node at (\i*\tikzTextDist,\h*\tikzNodeDist) [vertex] {};
                }
            }
            \foreach \h in {4,...,6}{
                \node at (4*\tikzTextDist,\h*\tikzNodeDist) [vertex] {};
            }
            \foreach \i in {2,5}{
                \foreach \h in {1}{
                    \node at (\i*\tikzTextDist,\h*\tikzNodeDist) [vertex,fill=red] {};
                }
            }
            \foreach \i in {4}{
                \foreach \h in {1,...,3}{
                    \node at (\i*\tikzTextDist,\h*\tikzNodeDist) [vertex,fill=red] {};
                }
            }
            \node at (3*\tikzTextDist,3*\tikzNodeDist) [vertex,fill=red] {};
            \draw[rounded corners,green,thick]
            (3*\tikzTextDist-\tikzTextMinDist,4*\tikzNodeDist-\tikzNodeMinDist+0.1)
            --(3*\tikzTextDist-\tikzTextMinDist,1*\tikzNodeDist-\tikzNodeMinDist-0.1)
            --(5*\tikzTextDist-\tikzTextMinDist-0.1,1*\tikzNodeDist-\tikzNodeMinDist-0.1)
            --(5*\tikzTextDist-\tikzTextMinDist-0.1,6*\tikzNodeDist-\tikzNodeMinDist)
            --(4*\tikzTextDist-\tikzTextMinDist-0.1,6*\tikzNodeDist-\tikzNodeMinDist)
            --(4*\tikzTextDist-\tikzTextMinDist-0.1,4*\tikzNodeDist-\tikzNodeMinDist+0.1)
            --cycle  {};

            \draw[rounded corners,orange,thick]
            (2*\tikzTextDist-\tikzTextMinDist,2*\tikzNodeDist-\tikzNodeMinDist)
            --(2*\tikzTextDist-\tikzTextMinDist,1*\tikzNodeDist-\tikzNodeMinDist)
            --(4*\tikzTextDist+\tikzTextMinDist,1*\tikzNodeDist-\tikzNodeMinDist)
            --(4*\tikzTextDist+\tikzTextMinDist,3*\tikzNodeDist-\tikzNodeMinDist)
            --(4*\tikzTextDist+\tikzTextMinDist,4*\tikzNodeDist-\tikzNodeMinDist)
            --(2*\tikzTextDist+\tikzTextMinDist,4*\tikzNodeDist-\tikzNodeMinDist)
            --(2*\tikzTextDist+\tikzTextMinDist,2*\tikzNodeDist-\tikzNodeMinDist)
            --cycle  {};

            \draw[rounded corners,red,thick]
            (4*\tikzTextDist-\tikzTextMinDist,7*\tikzNodeDist-\tikzNodeMinDist)
            --(4*\tikzTextDist-\tikzTextMinDist,2*\tikzNodeDist-\tikzNodeMinDist)
            --(5*\tikzTextDist-\tikzTextMinDist,2*\tikzNodeDist-\tikzNodeMinDist)
            --(5*\tikzTextDist-\tikzTextMinDist,1*\tikzNodeDist-\tikzNodeMinDist)
            --(5*\tikzTextDist-\tikzTextMinDist,1*\tikzNodeDist-\tikzNodeMinDist)
            --(6*\tikzTextDist+\tikzTextMinDist,1*\tikzNodeDist-\tikzNodeMinDist)
            --(6*\tikzTextDist+\tikzTextMinDist,3*\tikzNodeDist-\tikzNodeMinDist)
            --(5*\tikzTextDist+\tikzTextMinDist,3*\tikzNodeDist-\tikzNodeMinDist)
            --(4*\tikzTextDist+\tikzTextMinDist,7*\tikzNodeDist-\tikzNodeMinDist)
            --cycle  {};
        \end{tikzpicture}
        \caption{Suppose $s(i)=3$ for all $i$ except $s(i')=6$, $s(i'+1)=s(i'-2)=1$.
        Set $(m,r)=(m'_{\min},2)=(7,2)$.
        The difference statistics in green group are still repetitional-bias-corrected since it does not consist of the observation $Y_{1,i'-2}$.
        The difference statistics in orange group
        satisfies both that its leading observation is $Y_{s(i'-2),i'-2}$ and
        that it consists of the observation $Y_{1,i'-2}$ which does not have duplication.
        The difference statistics in red group is not repetitional-bias-corrected due to the observation $Y_{1,i'+1}$ which does not have duplication.
        }

        \label{fig:coro1:2}
    \end{figure}
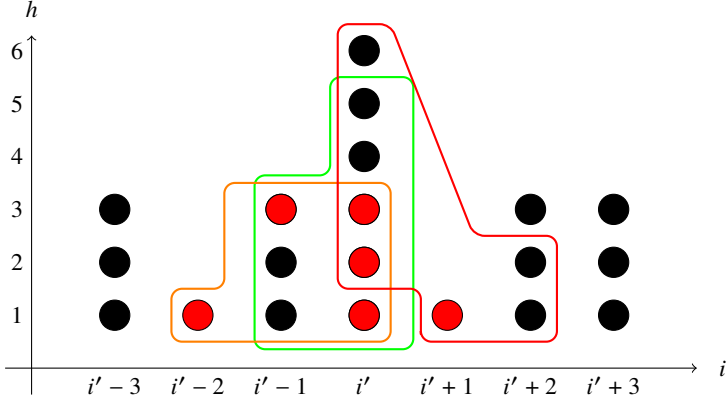
    Combining (\ref{eqt:caseii_upperboundUip_1}) and (\ref{eqt:caseii_upperboundUip_2}),
    we have
    \begin{align}
        |\mathcal{U}_{i'}|&\leq\sum_{h=1}^r
        s(i'-h)\mathbb{1}\left\{\min_{k=i'-h,\ldots,i'-1}s(k)=1\right\}+\sum_{h=1}^r\mathbb{1}\left\{\min_{k=i'-h,\ldots,i'-1}s(k)>1\right\} \notag\\
        &\quad+\sum_{h=1}^r
        s(i'+h)\mathbb{1}\left\{\min_{k=i'+1,\ldots,i'+h}s(k)=1\right\}
        +\sum_{h=1}^r\mathbb{1}\left\{\min_{k=i'+1,\ldots,i'+h}s(k)>1\right\} \notag\\
        &=\left\{\sum_{a=1}^rs(i'-a)^{b(i',-,a)}+s(i'+a)^{b(i',+,a)}\right\}. \label{eq:coro1:U}
    \end{align}
    Recall the function $b(i',\pm,a)=\mathbb{1}\left\{\min_{k=i\pm 1,\ldots,i\pm a}s(k)=1\right\}$.
    By definition, the right hand side of (\ref{eq:coro1:U}) reduces to
    \begin{align*}
        |\mathcal{U}_{i'}|\leq 2r,
    \end{align*}
    if $s(i)>1$ for $i\in[i-r,i+r]\cap\mathbb{Z}$.
\end{itemize}
Therefore, we have
\begin{align}
    \mathcal{D}&\leq \sum_{i'\in\mathcal{T}_s}|\mathcal{U}_{i'}|=S'_n, \notag\\
    \Bias\big\{\hat{\sigma}^2_{\Long}(m,r)\big\}
    &=\frac{1}{N-m-\mathcal{D}}\sum_{i\in\mathcal{I}\setminus\mathcal{I}_m}O_p(n^{-2r-2}) \notag\\
    &= O(n+S_nm-\mathcal{D})O_p(n^{-2r-2})/(N-m-\mathcal{D}) \label{eq:coro1:bias2}\\
    &=O_p\left\{(n^{-2r-1}+n^{-2r-2}S_nm)/N\right\}, \label{eq:coro1:bias3}
\end{align}
where (\ref{eq:coro1:bias2}) follows from $|\mathcal{I}\setminus\mathcal{I}_m|=O(n+S_nm-\mathcal{D})$
and (\ref{eq:coro1:bias3}) follows from $\mathcal{D}/N=S'_n/N\rightarrow 0$.

We continue to show the asymptotic variance remains the same as in Theorem {\ifnum\isXr=1{\ref{thm:long_diff_bias}}\else{4.3}\fi}.
Recall $a_{i,j}$ denotes the $(i,j)$th elements of $A$.
By similar arguments as in the proof of Theorem {\ifnum\isXr=1{\ref{thm:long_diff_bias}}\else{4.3}\fi},
we show $\tr(A^2)$ and $\tr\left\{A\Diag(A)\right\}$ are the leading order terms.
\begin{align}
    \tr(A^2)
    &=\sum_{i\in\mathcal{I}}\left(\sum^m_{j=0}d^{2}_{i,j}\right)^2+2\sum^m_{\ell=1}\sum_{i\in\mathcal{I}:i\leq N-m-\ell}\left(\sum^{m-\ell}_{j=0}d_{i,j+\ell}d_{i+\ell,j}\right)^2 \notag\\
    &=O(m)O(N-m-\mathcal{D}) \notag
    =O(Nm) ;  \notag \\
    g^\T A^2g
    &=\sum_{i\in\mathcal{I}}(D^{ g}_i)^2
    +2\sum^m_{\ell=1}\sum_{i\in\mathcal{I}:i\leq N-m-\ell}\left(\sum^{m-\ell}_{j=0}d_{i,j+\ell}d_{i+\ell,j}\right)D^{ g}_iD^{ g}_{i+\ell} \notag \\
    &=O(n+S_nm-\mathcal{D})O_p(n^{-2r-2})+2O(m)O(n+S_nm-\mathcal{D})O_p(n^{-2r-2}) \notag\\
    &=O(n^{-2r-1}m+n^{-2r-2}S_nm^2-n^{-2r-2}S'_n); \notag\\
    \left|\tr\left\{A\Diag(A)g{1}^\T_{n}\right\}\right|
    &=\left|\sum_{i\in\mathcal{I}}\sum^m_{j=0}d^{2}_{i,j}\sum^{(i+j)\wedge (N-m)}_{\substack{h=1\vee (i+j-m) \\
    h\in\mathcal{I}}}d_{h,i+j-h}D^{ g}_h\right| \notag\\
    &\leq \sum_{i\in\mathcal{I}}\max_{0\leq j\leq m}\left|\sum^{(i+j)\wedge (N-m)}_{\substack{h=1\vee (i+j-m) \\
    h\in\mathcal{I}}}d_{h,i+j-h}D^{ g}_h\right|\sum^m_{j=0}d^{2}_{i,j} \notag\\
    &=O(n+S_nm-\mathcal{D})O(m)O_p(n^{-r-1})
    =O_p(n^{-r}m+n^{-r-1}S_nm^2) ; \notag \\
    g^\T A\Diag(A)1_n
    &=\sum_{i\in\mathcal{I}}\left(\sum^m_{j=0}d_{i,j}a_{i+j,i+j}\right)D^{ g}_i \notag\\
    &=O_p(n+S_nm-\mathcal{D})O(m)O_p(n^{-r-1}) \label{eq:coro2:3} \\
    &=O_p(n^{-r}m+n^{-r-1}S_nm^{2}); \notag\\
    \left|\tr\left\{A\Diag(A)\right\}\right|
    &=\left|\sum_{i\in\mathcal{I}}\sum^m_{j=0}d^{2}_{i,j}a_{i+j,i+j}\right| \notag
    \leq \sum_{i\in\mathcal{I}}\sum^m_{j=0}\left|1a_{i+j,i+j}\right| \notag\\
    &=O(N-m-\mathcal{D})O(m)
    =O(Nm), \notag
\end{align}
where (\ref{eq:coro2:3}) follows from the fact that $\sum^m_{j=0}d_{i,j}a_{i+j,i+j}=O(m)$.
Therefore, in view of (\ref{eqt:var_SGW1993}),
the variance of $\hat{\sigma}_\Long^2(m,r)$
can be written as
\begin{equation*}
    \Var\left\{\hat{\sigma}_\Long^2(m,r)\right\}=\frac{1}{(N-m)^2}\left[2\sigma^4\tr(A^2)+(\lambda_4-3)\sigma^4\tr\left\{A\Diag(A)\right\}\right]+O_p\left\{\frac{n^{-r}m+n^{-r-1}S_nm^2}{(N-m)^2}\right\}.
\end{equation*}

\subsection{Proof of Proposition {\ifnum\isXr=1{\ref{prop:poly_cancelling}}\else{1}\fi}}
(a)
Under Assumption {\ifnum\isXr=1{\ref{ass:typeDiffSeq}\ref{assumption:design_adapt_con}}\else{1(a)}\fi},
for regression functions that admit the form $g(x)=\sum_{\ell=0}^ka_\ell x^\ell$,
we have
\begin{eqnarray*}
    \sum_{j=0}^md_{i,j}g(X_{i+j})&=&\sum^m_{j=0}d_{i,j}\sum^k_{\ell=0}a_\ell X_{i+j}^\ell\\
    &=&
    \sum^k_{\ell=0}a_\ell\sum^\ell_{h=0}\binom{\ell}{h}X_{i}^{\ell-h}\sum^m_{j=0}d_{i,j}(X_{i+j}-X_{i})^h\\
    &=& \sum^k_{\ell=r+1}a_\ell\sum^\ell_{h=r+1}\binom{\ell}{h}X_{i}^{\ell-h}\sum^m_{j=0}d_{i,j}(X_{i+j}-X_{i})^h,
\end{eqnarray*}
where the last line follows from the first constraint in {\ifnum\isXr=1{(\ref{diff_poly_cond0})}\else{(7)}\fi} and the constraints in {\ifnum\isXr=1{(\ref{diff_poly_random_cond})}\else{(11)}\fi}.
It is obvious that $d_{(i)}$ removes any $k$th degree polynomials exactly for any $k\leq r$.

(b)
Under Assumption {\ifnum\isXr=1{\ref{ass:typeDiffSeq}\ref{assumption:design_adapt_con}}\else{1(a)}\fi},
for an arbitrary function $g(x)$ that is $(r+1)$ times continuously differentiable,
Taylor expansion of $g(X_{i+j})$ about $X_i$ yields
\begin{eqnarray*}
    \sum_{j=0}^md_{i,j}g(X_{i+j})&=&
    \left\{\sum_{h=0}^{r+1}\frac{g^{(h)}(X_i)}{h!}\sum_{j=0}^md_{i,j}(X_{i+j}-X_{i})^h\right\}+o_p\big\{(X_{i+j}-X_{i})^{r+1}\big\}\\
    &=&\frac{g^{(r+1)}(X_i)}{(r+1)!}\sum_{j=0}^md_{i,j}(X_{i+j}-X_{i})^{r+1}+o_p\big\{(X_{i+j}-X_{i})^{r+1}\big\}\\
    &=&O_p(n^{-r-1}).
\end{eqnarray*}
Therefore, $d_{(i)}$ removes $g(x)$ up to order $O_p(n^{-r-1})$.

(c)
Under Assumption {\ifnum\isXr=1{\ref{ass:typeDiffSeq}\ref{assumption:regular_con}}\else{1(b)}\fi},
it can be easily seen that $D^g_i=0$ for any constant function $g$ due to the first constraint in {\ifnum\isXr=1{(\ref{diff_poly_cond0})}\else{(7)}\fi}.

(d)
Under Assumption {\ifnum\isXr=1{\ref{ass:typeDiffSeq}\ref{assumption:regular_con}}\else{1(b)}\fi},
for an arbitrary function $g(x)$ that is continuously differentiable,
Taylor expansion of $g(X_{i+j})$ about $X_{i}$ yields
\begin{align*}
    D^g_i&=\sum^m_{j=0}d_j\left\{g(X_{i})+\frac{g^{(1)}(X_{i})}{1!}(X_{i+j}-X_i)+o_p(n^{-1})\right\}\\
    &=g^{(1)}(X_{i})\sum^m_{j=0}d_j(X_{i+j}-X_i)+o_p(n^{-1})\\
    &=O_p(n^{-1})
\end{align*}
because the constraints {\ifnum\isXr=1{(\ref{diff_poly_cond0})}\else{(7)}\fi} and {\ifnum\isXr=1{(\ref{diff_poly_cond})}\else{(12)}\fi} cannot exactly remove the expression $\sum^m_{j=0}d_j(X_{i+j}-X_i)$.
If $g(x)$ that is Lipschitz continuous,
then
\begin{align*}
    |D^g_i|&=\left|\sum^m_{j=0}d_j\left\{g(X_{i})+g(X_{i+j})-g(X_i)\right\}\right|\\
    &=\left|\sum^m_{j=0}d_j\left\{g(X_{i+j})-g(X_i)\right\}\right|\\
    &\leq\sum^m_{j=0}\left|d_j\left\{g(X_{i+j})-g(X_i)\right\}\right|\\
    &\leq\sum^m_{j=0}|d_j|\left|X_{i+j}-X_i\right|\\
    &=O_p(n^{-1}),
\end{align*}
where the second to the last line follows from Lipschitz continuity.

\subsection{Proof of Lemma {\ifnum\isXr=1{\ref{lemma:diff_sol}}\else{5.1}\fi}}

Recall that $d_{(i)} = (d_{i,0}, \ldots, d_{i,m})^{\T}$.
By the property of the reduced row echelon form,
solving the linear system
$$\mathcal{L}_id_{(i)}=\begin{pmatrix}
    0 & \cdots & 0
\end{pmatrix}^\T$$
for $d_{(i)}$
can be done by solving
$$
    \mathcal{L}_{i,\beta}^{\prime\hspace{1pt}\T}d_{(i)} = 0
$$
iteratively for $\beta=r+1, \ldots, 1$,
where $\mathcal{L}_{i,\beta}^{\prime\hspace{1pt}\T}$ denotes the $\beta$th row of $\mathcal{L}'_i$.
Specifically, for $\beta=r+1$,
we have $0=
\mathcal{L}_{i,r+1}^{\prime\hspace{1pt}\T}
d_{(i)}$, i.e.,
\begin{align}
    0&=d_{i,r}+\sum^m_{j=r+1}L_{i,r+1,j+1}'d_{i,j}. \label{eq:prop1_proof:1}
\end{align}
Note that, for any $\phi_1,\ldots,\phi_m$,
the vector $(a_0,\ldots,a_m)^\T = u_0(\phi_1,\ldots,\phi_m)$
must satisfy that $\sum_{j=0}^m a_j^2=1$
by the definition of the function $u_0$.
So,
we can substitute $d_{(i)}=u_{0}(\phi_{i,1},\ldots,\phi_{i,m})$ into (\ref{eq:prop1_proof:1}) for some unknown $\phi_{i,1},\ldots,\phi_{i,m}$.
Then we solve for the unknown $\phi_{i,1},\ldots,\phi_{i,m}$, as follows:
\begin{align}
    0&=\left(\prod^{r}_{a=1}\sin\phi_{i,a}\right)\cos\phi_{i,r+1}+\sum^m_{j=r+1}L_{i,r+1,j+1}' \left( \prod^{j}_{b=1}\sin\phi_{i,b} \right) (\cos\phi_{i,j+1})^{\mathbb{1}(j\neq m)} \nonumber\\
    &=\left(\prod^{r}_{a=1}\sin\phi_{i,a} \right)\cos\phi_{i,r+1}+\sum^m_{j=r+1}L_{i,r+1,j+1}'
    \left( \prod^{r}_{c=1}\sin\phi_{i,c}  \right)
    \left(\prod^{j}_{b=r+1}\sin\phi_{i,b} \right)(\cos\phi_{i,j+1})^{\mathbb{1}(j\neq m)} . \label{eq:prop1_proof:2}
\end{align}
Dividing $\sin\phi_{i,r+1}\prod^{r}_{a=1}\sin\phi_{i,a}$ on both sides in (\ref{eq:prop1_proof:2}), we have
\begin{align*}
    0&=\frac{1}{\tan\phi_{i,r+1}}+L_{i,r+1,r+2}'\cos\phi_{i,r+2}+\sum^m_{j=r+2}L_{i,r+1,j+1}'
    \left( \prod^{j}_{b=r+2}\sin\phi_{i,b} \right) (\cos\phi_{i,j+1})^{\mathbb{1}(j\neq m)}.
\end{align*}
Let $u_{i,r+1} = u_{r+1}(\phi_{i,r+2},\ldots,\phi_{i,m})$.
Solving for $\phi_{i,r+1}$, we obtain
$$\phi_{i,r+1}=\arctan\left(\frac{-1}{{u}_{i,r+1}^\T\cdot L_{i,r+1}^{\prime}}\right).$$
Following similar procedures, for $1\leq\beta\leq r$,
solving
$0=
\mathcal{L}_{i,\beta}^{\prime\hspace{1pt}\T}
d_{(i)}$,
or equivalently $0=d_{i,\beta-1}+\sum^m_{j=r+1}L_{i,\beta,j+1}'d_{i,j}$ in its expanded form,
leads to
\begin{align}
    0
    &=\frac{1}{\tan\phi_{i,\beta}}+\sum^m_{j=r+1}L_{i,\beta,j+1}'
    \left( \prod^{j}_{b=\beta+1}\sin\phi_{i,b} \right) (\cos\phi_{i,j+1})^{\mathbb{1}(j\neq m)} \nonumber \\
    &=\frac{1}{\tan\phi_{i,\beta}}+
    \left( \prod^{r+1}_{b=\beta+1}\sin\phi_{i,b} \right)
    \begin{pmatrix}
        L_{i,\beta,r+2}' & \cdots & L_{i,\beta,m+1}'
    \end{pmatrix}
    \begin{pmatrix}
        \cos\phi_{i,r+2}\\
        \sin\phi_{i,r+2}\cos\phi_{i,r+3}\\
        \vdots\\
        \sin\phi_{i,r+2}\sin\phi_{i,r+3}\cdots\cos\phi_{i,m}\\
        \sin\phi_{i,r+2}\sin\phi_{i,r+3}\cdots\sin\phi_{i,m}
    \end{pmatrix}\nonumber\\
    &=\frac{1}{\tan\phi_{i,\beta}}+{u}_{i,r+1}^\T L_{i,\beta}'\prod^{r+1}_{b=\beta+1}\sin\phi_{i,b}. \nonumber
\end{align}
Solving for $\phi_{i,\beta}$, we obtain
\begin{align*}
    \phi_{i,\beta}&=\arctan\left(\frac{-1}{{u}_{i,r+1}^\T L_{i,\beta}'\prod^{r+1}_{b=\beta+1}\sin\phi_{i,b}}\right).
\end{align*}
Hence, any $d_{(i)}$ generated in this way will always satisfy the constraints given in $\mathcal{L}_i$ which is defined so that Assumption {\ifnum\isXr=1{\ref{ass:typeDiffSeq}\ref{assumption:design_adapt_con} or \ref{ass:typeDiffSeq}\ref{assumption:regular_con}}\else{1(a) or 1(b)}\fi}
 holds.
Derivation of difference sequence when $m=r+1$ follows similar procedures.

\section{Additional results}

\subsection{Sign and order reversal \label{miscell:SignOrder}}
The proposed difference sequence $d_{(i)}$ is always equivalent upon sign reversal, i.e.,
the estimators $\hat{\sigma}^2(m,r)$ computed based on
\begin{equation*}
    (d_{i,0},\ldots,d_{i,m})\quad\text{and}\quad(-d_{i,0},\ldots,-d_{i,m})
\end{equation*}
are numerically identical.
However, order reversal, i.e.,
changing $(d_{i,0},\ldots,d_{i,m})$ to\\
$(d_{i,m},\ldots,d_{i,0})$,
does not lead to an equivalent estimator in general.
It is trivial to see with a counter example:
Consider $m=2$, $r=1$ and $(X_1,X_2,X_3)=(1,3,4)$.
Then {\ifnum\isXr=1{(\ref{diff_poly_random_cond})}\else{(11)}\fi} becomes
\begin{align}
    \sum_{j=0}^2d_j(X_{1+j}-X_1)=2d_1+3d_2=0.\label{eq:signorder:1}
\end{align}
For the estimator with reversed difference sequence to be equivent to
the original estimator, we require
\begin{align}
    \sum_{j=0}^2d_{2-j}(X_{1+j}-X_1)=2d_1+3d_0=0.\label{eq:signorder:2}
\end{align}
Then (\ref{eq:signorder:1}) implies $d_1=-3d_2/2$.
Substituting $d_1=-3d_2/2$ into the first constraint in {\ifnum\isXr=1{(\ref{diff_poly_cond0})}\else{(7)}\fi}, we have
$d_0-3d_2/2+d_2=0$ or $d_0=d_2/2$.
Finally substituting $d_0=d_2/2$ and $d_1=-3d_2/2$ into (\ref{eq:signorder:2}), we have
\begin{align*}
    2d_1+3d_0=-3d_2+3d_2/2=-3d_2/2.
\end{align*}
So, order reversal leads to an equivalent estimator unless $d_2=0$.
However, $d_2=0$ will force $d_0=d_1=d_2=0$,
which violates the second constraints in {\ifnum\isXr=1{(\ref{diff_poly_cond0})}\else{(7)}\fi}.

However, order reversal always leads to an equivalent estimator
under equidistant design.
We prove by showing $\sum_{j=0}^md_jj^h=0$ if and only if $\sum_{j=0}^md_{m-j}j^h=0$ for $h=1,\ldots,r$.
First, suppose $\sum_{j=0}^md_jj^h=0$ holds for $h=1,\ldots,r$. Then
\begin{align*}
    \sum_{j=0}^md_{m-j}j^h
    &=\sum_{a=0}^md_{a}(m-a)^h
    =\sum_{a=0}^md_{a}\sum_{k=0}^h\binom{h}{k}m^{h-k}(-1)^k a^k\\
    &=\sum_{k=0}^h\binom{h}{k}m^{h-k}(-1)^k \sum_{a=0}^md_{a}a^k
    =0.
\end{align*}
Second, suppose $\sum_{j=0}^md_{m-j}j^h=0$ holds for $h=1,\ldots,r$. Then
\begin{align*}
    \sum_{j=0}^md_{j}j^h
    =\sum_{b=0}^md_{m-b}(m-b)^h
    =
    \sum_{k=0}^hm^{h-k}(-1)^k \sum_{b=0}^md_{m-b}b^k
    =0.
\end{align*}

\subsection{Residual-based variance estimators \label{rem:resBasedVarEst}}
The variance estimators based on detrending the data with
spline or kernel estimators are called residual-based
and their estimators generally admit the form of {\ifnum\isXr=1{(\ref{quadratic_form})}\else{(5)}\fi} with
$
A=(I-H)^\T (I-H),
$
where $I$ is an identity matrix and $H$ is the hat matrix such that
$\hat{Y}=HY$ is a linear estimate of $( g(X_1),\ldots,g(X_n) )^\T$.
Although \citet{HM1990} showed their residual-based method leads to
an smaller asymptotic mean squared error of $\Var(\epsilon^2)/n+o_p(n^{-1})$
comparing to
$C_1\Var(\epsilon^2)/n+o_p(n^{-1})$ with some $C_1>1$ of difference-based methods to be introduced.
Difference-based estimators are usually preferred because of the following reasons:
(a) ease of computation,
(b) tuning-free procedures,
(c) independence from trend estimators, and
(d) better-controlled bias.

\subsection{Change points \label{section:cp}}
This section discusses the robustness of $\hat{\sigma}^2_{*}(m,r)$ against regression function with jump discontinuities or change points.
\begin{definition}[piecewise-differentiable]
    A function $g(\cdot)$ is piecewise-differentiable if
    \begin{equation}\label{cp_trend}
        g(x)=g_{1}(x)\mathbb{1}(\beta_0\leq x\leq \beta_1)+\sum^B_{b=1}g_{b+1}(x)\mathbb{1}(\beta_b< x\leq \beta_{b+1}),
    \end{equation}
    where $B$ is finite, $\beta_0=0$, $\beta_{B+1}=1$ and $g_b(\cdot)$ is differentiable.
\end{definition}
\begin{proposition}
    \label{prop_cp_asymp}
    Suppose $(X_i)_{i=1}^n$ satisfies Assumption {\ifnum\isXr=1{\ref{assumption:basic}}\else{2}\fi} and
    $g(\cdot)$ is piece-wise-differentiable.
    Then
    \begin{align*}
            \Bias\big\{\hat{\sigma}_*^2(m,r)\big\}=O_p(n^{-1})
            \quad\text{and}\quad
            \Bias\big\{\hat{\sigma}^2(m,r)\big\}=O_p(n^{-1})
    \end{align*}
\end{proposition}
Proposition \ref{prop_cp_asymp} states the robustness of $\hat{\sigma}_*^2(m,r)$ and $\hat{\sigma}^2(m,r)$.
They remain asymptotically unbiased even if the regression function consists of change points,
however, the biases are no longer high-order corrected.

\begin{proof}
For any consecutive design points $(X_i,\ldots,X_{i+m})$
which contains $B^\star$ change points
such that $X_{i+v_{u}}\leq\beta_{b^\star_u}< X_{i+v_{u}+1}$
for $u=1,\ldots,B^\star$
and some integers $v_u,b^\star_u$, where $0\leq v_1<\cdots<v_{B^*}\leq m-1$,
\begin{eqnarray*}
    D^g_i
    &=&\sum^m_{j=0}d_j\Big\{\sum^{B^\star}_{u=1}g_{b^\star_{u}}(X_{i+j})\mathbb{1}(\beta_{b^\star_u-1}< X_{i+j}\leq \beta_{b^\star_u})\Big\}\\
    &=&\sum^{v_1}_{j=0}d_jg_{b^\star_1}(X_{i+j})
    +\sum^{{B^\star}}_{u=2}\sum^{v_{u}}_{j=v_{u-1}+1}d_jg_{b^\star_{u}}(X_{i+j})
    +\sum^m_{j=v_{B^\star}+1}d_jg_{b^\star_{B^\star+1}}(X_{i+j})\\
    &=&O_p(1).
\end{eqnarray*}
There are at most $mB$ $D^g_i$ that are of order $O_p(1)$. Thus, the bias of $\hat{\sigma}_*^2$ is inflated as follows:
\begin{eqnarray*}
    \Bias\big\{\hat{\sigma}_*^2(m,r)\big\}&=&\frac{1}{n-m}\sum^{n-m}_{i=1}(D^g_i)^2\\
    &=&O_p(n^{-2})+\frac{mB}{n-m}O_p(1)\\
    &=&O_p(n^{-1}).
\end{eqnarray*}
It can be easily verified that the variance remains the order of $O(n^{-1})$ as
\begin{align*}
    B_1 & =\frac{1}{n-m}\sum_{i=1}^{n-m}D^g_i = O_p(n^{-1}),\\
    B_2 & =\frac{1}{n-m}\sum_{i=1}^{n-m}(D^g_i)^2 = O_p(n^{-1}),\\
    B_3 & =\frac{1}{n-m}\sum_{\ell=1}^m\sum_{i=1}^{n-m-\ell}D^g_i D^g_{i+\ell} = O_p(n^{-1}).
\end{align*}
Thus, Assumption {\ifnum\isXr=1{\ref{assumption:Dgi_mean}}\else{4}\fi} holds with $\kappa=1/2$ and the variance expression follows from Theorem {\ifnum\isXr=1{\ref{thm_diff_var}}\else{4.2}\fi}.

\end{proof}

\subsection{Differenced-residuals \label{miscell:DiffResid}}
Difference-based estimator can be combined with regression techniques and yield a potentially better estimator.
\citet{TW2005} combined a first-order difference-based estimator with linear regression.
Here we propose a doubly-difference-based estimator $\hat{\sigma}^2_{\DR}$ which incorporates nonparametric regression technique to achieve bias of order $o_p(n^{-2r-2})$.

Suppose $g$ is a smooth and continuous function. By Taylor expansion of $g(X_{i+j})$ about $X_{i}$, we have
\begin{equation*}
    D^g_i=\frac{g^{(r+1)}(X_{i})}{(r+1)!}\sum^m_{j=0}d_{i,j}(X_{i+j}-X_{i})^{r+1}+o_p(n^{-r-1}).
\end{equation*}
We observe that the bias cancelling effect of traditional difference-based estimators depends solely on $X_{i+j}-X_{i}$ of order $O_p(n^{-1})$ for $i=1,\ldots,n-m$ and $j=0,\ldots,m$.
Inspired by the smoothing-type techniques, we propose
\begin{eqnarray}
    \hat{\sigma}^2_{\DR}(m,r)&=&\frac{1}{\eta}\sum^{n-m}_{i=1}\hat{D}_i^2,\label{diff_est_new}\\
    \hat{D}_i&=& \sum^m_{j=0}d_{i,j}\left\{y_{i+j}-\frac{\hat{g}^{(r+1)}_{(i)}}{(r+1)!}(X_{i+j}-X_{i})^{r+1}\right\},\label{D_estII}\\
    \eta&=&\left\{n-m-2\sum^{n-m}_{i=1}\sum^m_{\ell=0}d_{i,\ell}\frac{\omega_{i,r+1,i+\ell}}{(r+1)!}\sum^m_{j=0}d_{i,j}(X_{i+j}-X_{i})^{r+1}\right\}, \notag
\end{eqnarray}
where $\hat{g}^{(h)}_{(i)}=\sum^n_{v=1}\omega_{i,h,v}y_v$ is a linear estimate of $g^{(h)}(X_{i})$ and $\eta$ is the effective sample size.

We embed the idea of smoothing in the traditional difference-based estimator and
construct a differenced-residual statistics $\hat{D}_i$ which further smooths the $D_i$,
\begin{equation*}
    \hat{D}_i=\frac{g^{(r+1)}(X_{i})-\hat{g}^{(r+1)}_{(i)}}{(r+1)!}\sum^m_{j=0}d_{i,j}(X_{i+j}-X_{i})^{r+1}+o_p(n^{-r-1})+D^{\epsilon}_i.
\end{equation*}
The additional layer of differencing boosts the convergence rate on top of the effect of conventional difference sequence
as long as $\hat{g}^{(r+1)}_{(i)}$ is an consistent estimate of $g^{(r+1)}(x_{i})$.

We suggest to estimate the required derivatives by local polynomial regression of order
$p\geq r+1$
such that $p-r-1$ is odd.
Following the two assumptions concerning the kernel $K$ and bandwidth $h_n$ involved in the local polynomial estimator, Theorem \ref{thm_double_diff_asymp} states the improvement in bias.
\begin{assumption}\label{assumption:locpol}
    The kernel $K$ is compactly supported on $[-1,1]$ and bounded such that $\int_{\mathbb{R}} u^2K(u)\dd u$ is a non-zero constant.
    In addition, $\int_{\mathbb{R}} u^hK(u)\dd u=0$ for all odd $h$.
    The bandwidth parameter $h_n$ satisfies $h_n\rightarrow0$ and $nh_n^{2(r+1)+1}\rightarrow\infty$ as $n\rightarrow\infty$.
\end{assumption}
\begin{theorem}\label{thm_double_diff_asymp}
    Suppose Assumptions {\ifnum\isXr=1{\ref{assumption:basic}}\else{2}\fi} and \ref{assumption:locpol} hold. Under a random design, assume $g^{(p+2)}(\cdot)$ is continuous,
    $p\geq r+1$ and $p-r-1$ is odd, where $p$ is the order of local polynomial regression for derivatives estimation.
    Then
    \begin{align*}
        \Bias\big\{\hat{\sigma}^2_{\DR}(m,r)\}=&\frac{1}{\eta}\sum^{n-m}_{i=1}\Bigg\{\frac{\sum^m_{j=0}d_{i,j}(X_{i+j}-X_{i})^{r+1}}{(r+1)!}\Bigg\}^2\E\Biggl[\Bigg\{g^{(r+1)}(X_{i})-\hat{g}^{(r+1)}_{(i)}\Bigg\}^2\Biggr]\\
        &+o_p(n^{-2r-2})\\
        =& O_p\big\{n^{-2r-2}(h_n^{2p-2r}+n^{-1}h_n^{-2r-3})\big\}+o_p(n^{-2r-2}),\\
        \Var\big\{\hat{\sigma}^2_{\DR}(m,r)\big\}=&\frac{1}{\eta^2}\Big[(n-m)^2\Var\big\{\hat{\sigma}_*^2(m,r,0)\big\}+O_p(n^{-r+1/2}h_n^{1/2})\Big].
    \end{align*}
\end{theorem}
Proof of Theorem \ref{thm_double_diff_asymp} is deferred to the end of this section.
Theorem \ref{thm_double_diff_asymp} states
the additional differencing layer further lowers the bias to the order of $o_p(n^{-2r-2})$.
The bias of $\hat{\sigma}^2_{\DR}$ is dominated by a term depending on the mean squared error of $\hat{g}^{(r+1)}_{(i)}$.
In general, we recommend using $p=r+2$ which provide satisfactory improvements.

According to Theorem \ref{thm_double_diff_asymp}, $h_n$ influences the asymptotic properties through
\begin{equation*}
    \E\Big[\big\{g^{(r+1)}(x_{i})-\hat{g}^{(r+1)}_{(i)}\big\}^2\Big].
\end{equation*}
As a result, the optimal $h_n$ for $\hat{\sigma}^2_{\DR}$ is equivalent to the optimal bandwidth for estimating $g^{(r+1)}(x_{i})$. See \citet{FG1995}, \citet{WJ1995}, \citet{FG2018} and \citet{CCF2018} for some existing bandwidth selection methods. Also see \citet{CCF2019} for an R package for bandwidth selection. Alternatively, one may also obtain one bandwidth globally by minimizing the main terms in $\Bias(\hat{\sigma}^2_{\DR})$ with plug-in estimators for unknown terms.

\begin{proof}[Proof of Theorem \ref{thm_double_diff_asymp}]

We begin with the bias of $\hat{\sigma}^2_{\DR}(m,r)$.
By Taylor's expansion,
\begin{eqnarray*}
    \hat{\sigma}^2_{\DR}(m,r)&=&\frac{1}{\eta}\sum^{n-m}_{i=1}\hat{D}_i^2
    =\frac{1}{\eta}\sum^{n-m}_{i=1}\left[\sum^m_{j=0}d_{i,j}\Big\{Y_{i+j}-\frac{\hat{g}^{(r+1)}_{(i)}}{(r+1)!}(X_{i+j}-X_{i})^{r+1}\Big\}\right]^2\\
    &=&\frac{1}{\eta}\sum^{n-m}_{i=1}\left[\sum^m_{j=0}d_{i,j}\Big\{g(X_{i})+\sum^{r+1}_{h=1}\frac{g^{(h)}(X_{i})}{h!}(X_{i+j}-X_{i})^h+o(X_{i+j}-X_{i})^{r+1}\right.\\
    &&\left.-\frac{\hat{g}^{(r+1)}_{(i)}}{(r+1)!}(X_{i+j}-X_{i})^{r+1}+\epsilon_{i+j}\Big\}\right]^2\\
    &=&\frac{1}{\eta}\sum^{n-m}_{i=1}\Bigg[\sum^m_{j=0}d_{i,j}\Big\{\frac{g^{(r+1)}(X_{i})-\hat{g}^{(r+1)}_{(i)}}{(r+1)!}(X_{i+j}-X_{i})^{r+1}+o(X_{i+j}-X_{i})^{r+1}\Big\}\\
    &&\qquad+\sum^m_{j=0}d_{i,j}\epsilon_{i+j}\Bigg]^2\label{eq:1}.
\end{eqnarray*}

Let
\begin{eqnarray*}
    \mathcal{D}_{g,i}&=&\sum^m_{j=0}d_{i,j}\Big\{\frac{g^{(r+1)}(X_{i})-\hat{g}^{(r+1)}_{(i)}}{(r+1)!}(X_{i+j}-X_{i})^{r+1}+o(X_{i+j}-X_{i})^{r+1}\Big\}.
\end{eqnarray*}
Then
\begin{eqnarray*}
    \hat{\sigma}^2_{\DR}(m,r)&=&\frac{1}{\eta}\sum^{n-m}_{i=1}\left\{\mathcal{D}_{g,i}+D^\epsilon_i\right\}^2
    =\frac{1}{\eta}\sum^{n-m}_{i=1}\left\{\mathcal{D}_{g,i}^2+2\mathcal{D}_{g,i}D^\epsilon_i+(D^\epsilon_i)^2\right\}.
\end{eqnarray*}
By Theorem 4.2 in \cite{RW1994},
where the asymptotic properties of $\hat{g}^{(r+1)}_{(i)}$ under local polynomial regression are given, we have
\begin{align}
    \E(\mathcal{D}_{g,i}^2)
    &=\E\left[\Bigg\{\frac{g^{(r+1)}(X_{i})-\hat{g}^{(r+1)}_{(i)}}{(r+1)!}\sum^m_{j=0}d_{i,j}(X_{i+j}-X_{i})^{r+1}+o(X_{i+j}-X_{i})^{r+1}\Bigg\}^2\right] \notag\\
    &=\E\left[\Bigg\{\frac{g^{(r+1)}(X_{i})-\hat{g}^{(r+1)}_{(i)}}{(r+1)!}\sum^m_{j=0}d_{i,j}(X_{i+j}-X_{i})^{r+1}\Bigg\}^2\right. \notag\\
    &\quad\left.+2o_p(n^{-r-1})\Bigg\{\frac{g^{(r+1)}(X_{i})-\hat{g}^{(r+1)}_{(i)}}{(r+1)!}\sum^m_{j=0}d_{i,j}(X_{i+j}-X_{i})^{r+1}\Bigg\}\right]+o_p(n^{-2r-2}) \notag\\
    &=\E\left[\Bigg\{\frac{g^{(r+1)}(X_{i})-\hat{g}^{(r+1)}_{(i)}}{(r+1)!}\sum^m_{j=0}d_{i,j}(X_{i+j}-X_{i})^{r+1}\Bigg\}^2\right] \notag\\
    &\quad-2o_p(n^{-r-1})\Bigg[O_p(n^{-r-1})\E\left\{\hat{g}^{(r+1)}_{(i)}-g^{(r+1)}(X_{i})\right\}\Bigg]+o_p(n^{-2r-2}) . \label{eq:diff_resi:1}
\end{align}
The asymptotic bias formula is applied on $\E\left\{\hat{g}^{(r+1)}_{(i)}-g^{(r+1)}(X_{i})\right\}$ in (\ref{eq:diff_resi:1}),
which gives
\begin{align*}
    \E\left\{\hat{g}^{(r+1)}_{(i)}-g^{(r+1)}(X_{i})\right\}&=
    \int u^{p+1}K_{(r+1,p)}(u)\dd u\frac{g^{(p+1)}(X_{i})}{(p+1)!}h_n^{p-r-1+1}\{1+o_p(1)\}\\
    &=O(h_n^{p-r})\{1+o_p(1)\},
\end{align*}
where $K_{(r+1,p)}(t)$ satisfies
\begin{align*}
    \int_{\mathbb{R}} u^jK_{(r+1,p)}(u)\dd u=
    \left\{\begin{array}{ll}
        0,&\qquad 0 \leq j\leq p,j\neq r+1,\\
        (r+1)!,&\qquad j=r+1,\\
        \xi_{r+1,p},&\qquad j=p+1,\\
    \end{array}\right.
\end{align*}
for some non-zero constant $\xi_{r+1,p}$.
See \citet{RW1994} for more details.
Therefore, we have
\begin{align}
    \E(\mathcal{D}_{g,i}^2)
    &=\E\left[\Bigg\{\frac{g^{(r+1)}(X_{i})-\hat{g}^{(r+1)}_{(i)}}{(r+1)!}\sum^m_{j=0}d_{i,j}(X_{i+j}-X_{i})^{r+1}\Bigg\}^2\right] \notag\\
    &\quad-2o_p(n^{-r-1})O_p(n^{-r-1})O(h_n^{p-r})\{1+o_p(1)\}+o_p(n^{-2r-2}) \notag\\
    &=\E\left[\Bigg\{\frac{g^{(r+1)}(X_{i})-\hat{g}^{(r+1)}_{(i)}}{(r+1)!}\sum^m_{j=0}d_{i,j}(X_{i+j}-X_{i})^{r+1}\Bigg\}^2\right]+o_p(n^{-2r-2}) \notag\\
    &=\Bigg\{\frac{\sum^m_{j=0}d_{i,j}(X_{i+j}-X_{i})^{r+1}}{(r+1)!}\Bigg\}^2\E\Biggl[\Bigg\{g^{(r+1)}(X_{i})-\hat{g}^{(r+1)}_{(i)}\Bigg\}^2\Biggr]+o_p(n^{-2r-2}) \label{eq:diff_resi:2}\\
    &=O_p(n^{-2r-2})\Biggl[\Bigg\{\int u^{p+1}K_{(r+1,p)}(u)\dd u\Bigg\}^2\Bigg\{\frac{g^{(p+1)}(X_{i})}{(p+1)!}\Bigg\}^2h_n^{2(p-r-1+1)}\{1+o_p(1)\} \notag\\
    &\quad+\frac{\sigma^2}{nh_n^{2r+2+1}f_X(X_{i})}\int K_{(r+1,p)}(u)^{2}\dd u\{1+o_p(1)\}\Biggr]+o_p(n^{-2r-2}) \notag\\
    &=O_p(n^{-2r-2})O(h_n^{2p-2r}+n^{-1}h_n^{-2r-3})\{1+o_p(1)\}+o_p(n^{-2r-2}) \notag\\
    &=O_p\big\{n^{-2r-2}(h_n^{2p-2r}+n^{-1}h_n^{-2r-3})\big\}+o_p(n^{-2r-2}) \notag,
\end{align}
where the asymptotic bias and variance formulae in Theorem 4.2 by \cite{RW1994} are applied to (\ref{eq:diff_resi:2}).
\begin{eqnarray*}
    \E(\mathcal{D}_{g,i}D^\epsilon_i)
    &=&\E\Biggl[\sum^m_{\ell=0}d_{i,\ell}\epsilon_{i+\ell}\Bigg\{\sum^m_{j=0}d_{i,j}\frac{g^{(r+1)}(X_{i})-\hat{g}^{(r+1)}_{(i)}}{(r+1)!}(X_{i+j}-X_{i})^{r+1}\\
    &&\qquad +o(X_{i+j}-X_{i})^{r+1}\Bigg\}\Biggr]\\
    &=&\E\Biggl[\Bigg\{\frac{g^{(r+1)}(X_{i})-\sum^n_{v=1}\omega_{i,r+1,v}y_v}{(r+1)!}\sum^m_{j=0}d_{i,j}(X_{i+j}-X_{i})^{r+1}\Bigg\}\sum^m_{\ell=0}d_{i,\ell}\epsilon_{i+\ell}\Biggr]\\
    &=&\E\Biggl[\Bigg\{\frac{-\sum^n_{v=1}\omega_{i,r+1,v}y_v}{(r+1)!}\sum^m_{j=0}d_{i,j}(X_{i+j}-X_{i})^{r+1}\Bigg\}\sum^m_{\ell=0}d_{i,\ell}\epsilon_{i+\ell}\Biggr]\\
    &=&\E\Biggl[\Bigg\{\sum^n_{v=1}\frac{-\omega_{i,r+1,v}}{(r+1)!}\sum^m_{j=0}d_{i,j}(X_{i+j}-X_{i})^{r+1}y_v\Bigg\}\sum^m_{\ell=0}d_{i,\ell}\epsilon_{i+\ell}\Biggr]\\
    &=&\E\Biggl(\Bigg[\sum^n_{v=1}\frac{-\omega_{i,r+1,v}}{(r+1)!}\sum^m_{j=0}d_{i,j}(X_{i+j}-X_{i})^{r+1}\big\{g(X_{v})+\epsilon_v\big\}\Bigg]\sum^m_{\ell=0}d_{i,\ell}\epsilon_{i+\ell}\Biggr)\\
    &=&\E\Biggl[\Bigg\{\sum^n_{v=1}\frac{-\omega_{i,r+1,v}}{(r+1)!}\sum^m_{j=0}d_{i,j}(X_{i+j}-X_{i})^{r+1}\epsilon_v\Bigg\}\sum^m_{\ell=0}d_{i,\ell}\epsilon_{i+\ell}\Biggr]\\
    &=&\E\Biggl\{\sum^m_{\ell=0}d_{i,\ell}\frac{-\omega_{i,r+1,i+\ell}}{(r+1)!}\sum^m_{j=0}d_{i,j}(X_{i+j}-X_{i})^{r+1}\epsilon_{i+\ell}^2\Biggr\}\\
    &=&-\sum^m_{\ell=0}d_{i,\ell}\frac{\omega_{i,r+1,i+\ell}}{(r+1)!}\sum^m_{j=0}d_{i,j}(X_{i+j}-X_{i})^{r+1}\sigma^2.
\end{eqnarray*}
Therefore,
\begin{eqnarray*}
    \E(\hat{\sigma}^2_{\DR})&=&\frac{1}{\eta}\sum^{n-m}_{i=1}\left[\E(\mathcal{D}_{g,i}^2)+2(\mathcal{D}_{g,i}D^\epsilon_i)+\E\Big\{(D^\epsilon_i)^2\Big\}\right]\\
    &=&O_p\big\{n^{-2r-2}(h_n^{2p-2r}+n^{-1}h_n^{-2r-3})\big\}+o_p(n^{-2r-2})\\
    &&-\frac{2}{\eta}\sum^{n-m}_{i=1}\sum^m_{\ell=0}d_{i,\ell}\frac{\omega_{i,r+1,i+\ell}}{(r+1)!}\sum^m_{j=0}d_{i,j}(X_{i+j}-X_{i})^{r+1}\sigma^2
    +\frac{1}{\eta}\sum^{n-m}_{i=1}\sigma^2\\
    &=&O_p\big\{n^{-2r-2}(h_n^{2p-2r}+n^{-1}h_n^{-2r-3})\big\}+o_p(n^{-2r-2})\\
    &&+\frac{\sigma^2}{\eta}\left\{n-m-2\sum^{n-m}_{i=1}\sum^m_{\ell=0}d_{i,\ell}\frac{\omega_{i,r+1,i+\ell}}{(r+1)!}\sum^m_{j=0}d_{i,j}(X_{i+j}-X_{i})^{r+1}\right\}\\
    &=&O_p\big\{n^{-2r-2}(h_n^{2p-2r}+n^{-1}h_n^{-2r-3})\big\}+o_p(n^{-2r-2})+\sigma^2.
\end{eqnarray*}
Hence,
$$\Bias\big\{\hat{\sigma}^2_{\DR}(m,r)\big\}
=O_p\big\{n^{-2r-2}(h_n^{2p-2r}+n^{-1}h_n^{-2r-3})\big\}+o_p(n^{-2r-2}) . $$

Now we derive the variance of $\hat{\sigma}^2_{\DR}(m,r)$. In matrix form, we have
\begin{equation*}
    \hat{\sigma}^2_{\DR}=\frac{1}{\eta}(\DiffSeqMat Y-B)^\T(\DiffSeqMat Y-B),
\end{equation*}
where
\begin{align*}
    B&=\begin{pmatrix}
        \frac{\hat{g}^{(r+1)}_{(1)}}{(r+1)!}\sum^m_{j=0}d_{1,j}(X_{1+j}-X_{1})^{r+1}\\
        \vdots\\
        \frac{\hat{g}^{(r+1)}_{(n-m)}}{(r+1)!}\sum^m_{j=0}d_{n-m,j}(X_{n-m+j}-X_{n-m})^{r+1}\\
    \end{pmatrix}\\
    &=\Omega Y,\\
    \Omega&=\begin{pmatrix}
        \frac{\omega_{1,r+1,1}}{(r+1)!}\sum^m_{j=0}d_{1,j}(X_{1+j}-X_{1})^{r+1} & \cdots & \frac{\omega_{1,r+1,n}}{(r+1)!}\sum^m_{j=0}d_{1,j}(X_{1+j}-X_{1})^{r+1}\\
        \vdots & & \vdots\\
        \frac{\omega_{n-m,r+1,1}}{(r+1)!}\sum^m_{j=0}d_{n-m,j}(X_{1+j}-X_{1})^{r+1} & \cdots & \frac{\omega_{n-m,r+1,n}}{(r+1)!}\sum^m_{j=0}d_{n-m,j}(X_{1+j}-X_{1})^{r+1}
    \end{pmatrix}.
\end{align*}
Since $K(\cdot)$ is bounded, there are only $O(h_n)$ non-zero elements in each row and column of $\Omega$. Now we have
\begin{equation*}
    \Var\big\{\hat{\sigma}^2_{\DR}(m,r)\big\}=\frac{1}{\eta^2}\Var\left(Y^\T \DiffSeqMat^\T \DiffSeqMat Y+B^\T B-2B^\T \DiffSeqMat Y\right).
\end{equation*}
Let $\Omega_2=\Omega^\T \Omega$ and denote $1_n$ as a column vector of $1$ of length $n$.
By the theorem in \cite{UZ1992} which provides the formulae for the first and second moments of $\hat{\sigma}^2_{\text{QF}}$,
we have
\begin{align*}
    \E(B^\T B)^2
    &=\left\{
    g^\T \Omega_2g + \sigma^2 \tr(\Omega_2)
    \right\}^2\\
    &=\left\{
    O(n)O(n^{-2r-2}h_n^2)
    +
    \sigma^2 O(n)O(n^{-2r-2}h_n)
    \right\}^2\\
    &=\left\{
    O(n^{-2r-1}h_n)
    \right\}^2\\
    &=O(n^{-4r-2}h_n^2),\\
    \Var(B^\T B)&=\Var(Y^\T \Omega_2 Y)\\
    &=\sigma^4(\lambda_4-3)\tr\{\Diag(\Omega_2)^2\}+2\sigma^4\tr(\Omega_2^2)\\
    &\qquad+4\sigma^2g^\T \Omega_2^2 g+4\sigma^3\lambda_3g^\T\Omega_2\Diag(\Omega_2)1_n\\
    &=\sigma^4(\lambda_4-3)O_p(n^{-4r-3}h_n^{2})+2\sigma^4O_p(n^{-4r-2}h_n^2)\\
    &\qquad+4\sigma^2O_p(n^{-4r-3}h_n^{4})+4\sigma^3\lambda_3O_p(n^{-4r-3}h_n^3)\\
    &=O_p(n^{-4r-2}h_n^2),\\
    \E\left\{(B^\T B)^2\right\}&=\Var(B^\T B)+\E(B^\T B)^2\\
    &=O_p(n^{-4r-2}h_n^2).
\end{align*}
By Cauchy--Schwarz inequality,
we have
\begin{align}
    0\leq(B^\T \DiffSeqMat Y)^2\leq B^\T B Y^\T \DiffSeqMat^\T \DiffSeqMat Y\label{eq:DR:1}.
\end{align}
Taking expectation in (\ref{eq:DR:1}) yields
\begin{align}
    0\leq\E\Big\{(B^\T \DiffSeqMat Y)^2\Big\}&\leq \E(B^\T B Y^\T \DiffSeqMat^\T \DiffSeqMat Y) \nonumber\\
    &\leq\left[\E\Big\{(B^\T B)^2\Big\} \E\Big\{(Y^\T \DiffSeqMat^\T \DiffSeqMat Y)^2\Big\}\right]^{\frac{1}{2}} \label{eq:DR:2}\\
    &=\left\{O_p(n^{-4r-2}h_n^2)
    O(n^2)
    \right\}^{\frac{1}{2}} \nonumber\\
    &=O_p(n^{-2r}h_n) \nonumber,
\end{align}
where (\ref{eq:DR:2}) follows from Cauchy--Schwarz inequality again.
As a result,
we have
\begin{align*}
    0\leq\Var(B^\T \DiffSeqMat Y)&\leq\E\Big\{(B^\T \DiffSeqMat Y)^2\Big\} \\
    &=O_p(n^{-2r}h_n).
\end{align*}
Therefore, the variance of $\hat{\sigma}^2_{\DR}(m,r)$ is asymptotically equal to that of $\hat{\sigma}_*^2(m,r)$:
\begin{eqnarray*}
    \Var\big\{\hat{\sigma}^2_{\DR}(m,r)\big\}&=&\frac{1}{\eta^2}\Var\left(Y^\T \DiffSeqMat^\T \DiffSeqMat Y+B^\T B-2B^\T \DiffSeqMat Y\right)\\
    &=&\frac{1}{\eta^2}\Big[(n-m)^2\Var\big\{\hat{\sigma}_*^2(m,r)\big\}+O_p(n^{-r+1/2}h_n^{1/2})\Big].
\end{eqnarray*}
\end{proof}

\subsection{Discussion: Global polynomial fitting \label{miscell:PolyRegFit}}
The idea of cancelling a polynomial trend may similarly be achieved by a parametric approach.
One can fit an $r$th order polynomial globally to obtain the residuals,
then the noise variance can be estimated by the average of squared residuals.
If the true regression function is an $r$th order polynomial,
Lemma B.10 in \citet{TC2023} ensures that the estimated coefficients are accurate up to $O(n^{-1/2})$ errors.

However, using $\hat{\sigma}_*^2$ is superior in terms of the types of regression functions that it can cancel.
According to Theorem {\ifnum\isXr=1{\ref{thm_diff_bias}}\else{4.1}\fi}, $\hat{\sigma}_*^2$ still performs well for arbitrary functions that are sufficiently smooth.
On the contrary, \cite{FG2018} has illustrated that even a global polynomial regression up to $4$th order fails to produce visually satisfying fits in real data.
Moreover, local differencing offers protection against the lack of knowledge about the potentially non-polynomial regression functions.

In addition, difference-based estimators are still valid under the existence of change points according to Proposition \ref{prop_cp_asymp}.
While $\hat{\sigma}_*^2$ suffers from inflation of bias yet remains consistent,
the global polynomial fitting approach suffers worst problems.
Suppose the true regression function is
\begin{equation*}
    g(x)=\sum^r_{\ell=0}a_{1,\ell} x^\ell \mathbb{1}(x\leq0.5)+\sum^r_{\ell=0}a_{2,\ell} x^\ell \mathbb{1}(x>0.5).
\end{equation*}
Fitting a $r$th order polynomial always results in
\begin{equation*}
    g(x)=\sum^r_{\ell=0}(a_{1,\ell}-\hat{a}_{\ell}) x^\ell \mathbb{1}(x\leq0.5)+\sum^r_{\ell=0}(a_{2,\ell}-\hat{a}_{\ell}) x^\ell \mathbb{1}(x>0.5),
\end{equation*}
where $\hat{a}_{\ell}$ are the fitted coefficients. The global fitting can never get rid of the indicator function, which is the change point. Basically, the global polynomial fitting approach is invalid. In contrast, local differencing is still robust under the existence of change points.

\subsection{Reduced row echelon form}\label{miscell:RREF}
To define reduced row echelon form,
we first define row echelon form.
By \citet[Chapter~2.1]{M2023},
an $m\times n$ matrix $E$ with rows $E_{i\bullet}$ and column $E_{\bullet j}$ is said to be in row echelon form if the following two conditions hold.
\begin{itemize}
    \item If $E_{i\bullet}$ consists entirely of zeros, then all rows below $E_{i\bullet}$ are also entirely zero, i.e., all zero rows are at the bottom.
    \item If the first nonzero entry in $E_{i\bullet}$ lies in the $j$th position, then all entries below the $i$th position in columns $E_{\bullet 1}$, $E_{\bullet 2}$, \ldots , $E_{\bullet j}$ are zero.
\end{itemize}
Let $*$ be any real number.
An example of matrix in row echelon form is
\[
\begin{pmatrix}
    * & * & * & * & *\\
    0 & * & * & * & *\\
    0 & 0 & 0 & * & *\\
    0 & 0 & 0 & 0 & *\\
    0 & 0 & 0 & 0 & 0\\
\end{pmatrix}.
\]

According to \citet[Chapter~2.2]{M2023},
an $m\times n$ matrix $E$ is said to be in reduced row echelon form if the following three conditions hold.
\begin{itemize}
    \item $E$ is in row echelon form.
    \item In each row, the pivot (i.e., the first nonzero entry) is $1$.
    \item All entries above each pivot are 0.
\end{itemize}
An example of matrix in reduced row echelon form is
\[
\begin{pmatrix}
    1 & 0 & * & 0 & 0\\
    0 & 1 & * & 0 & 0\\
    0 & 0 & 0 & 1 & 0\\
    0 & 0 & 0 & 0 & 1\\
    0 & 0 & 0 & 0 & 0\\
\end{pmatrix}.
\]
Let $M$ be an $m\times n$ matrix where $m<n$.
Denote the row operations to transform $M$ to its reduced row echelon form as $R$,
i.e., $M'=RM$.
Write $M$ as
\[
M = \begin{pmatrix}
    A_{m\times m} & B_{m\times (n-m)}
\end{pmatrix},
\]
where $A_{m\times m}$ is an $m\times m$ matrix and $B_{m\times (n-m)}$ is an $m\times (n-m)$ matrix.
If $A_{m\times m}$ is non-singular,
then $M'$ admits the form:
\[
M' = \begin{pmatrix}
    1 & \cdots & 0 & * & \cdots & *\\
    \vdots & \ddots & \vdots & \vdots & \ddots & \vdots\\
    0 & \cdots & 1 & * & \cdots & *\\
\end{pmatrix},
\]
i.e., the first $m$ columns of $M'$ is an identity matrix.
Because
\[
M' = RM
= R\begin{pmatrix}
    A_{m\times m} & B_{m\times (n-m)}
\end{pmatrix}
=\begin{pmatrix}
    RA_{m\times m} & RB_{m\times (n-m)}
\end{pmatrix}
=\begin{pmatrix}
    I_{m\times m} & RB_{m\times (n-m)}
\end{pmatrix}.
\]
\end{appendix}

\begin{acks}[Acknowledgments]
The authors would like to thank the anonymous referees, an associate editor, and the editor for their constructive comments that improved the quality of the article.
The authors would also like to thank Siu Hang Edmund Ip for his preliminary work contributed at the early stage of this project.
\end{acks}

\begin{funding}
This research was partially supported by grants GRF-14306421 and 14307922 provided by
the Research Grants Council of HKSAR.
\end{funding}

\bibliographystyle{imsart-nameyear.bst}
\bibliography{myRef}

\end{document}